%% file: main.tex
\documentclass[11pt]{article}

\usepackage{geometry}
\usepackage{graphicx}
\usepackage{array}
\usepackage{caption}
\usepackage[authordate-trad,noibid,backend=biber,natbib]{biblatex-chicago}
\usepackage{amsmath}
\usepackage{amssymb}
\usepackage{amsthm}
\usepackage[nopatch]{microtype}
\usepackage{booktabs}
\usepackage{longtable}
\usepackage{pdflscape}
\usepackage{placeins}
\usepackage{tikz}                                  % design flowchart, Figure 2
\usetikzlibrary{shapes, arrows, arrows.meta, calc}
\usepackage[hidelinks]{hyperref} % hidelinks: no colored boxes around every
\usepackage{xurl}      % lets \href/\url break at any character instead of
\renewcommand\appendix{\par
  \setcounter{section}{0}
  \setcounter{subsection}{0}
  \gdef\thesection{Appendix \arabic{section}}
}

\theoremstyle{plain}
\newtheorem{assumption}{Assumption}
\newtheorem{prop}{Proposition}
\newtheorem{lem}{Lemma}

\theoremstyle{remark}
\newtheorem{rem}{Remark}

\newcommand{\advice}[1]{%
  \par\smallskip
  \noindent{\sffamily\itshape #1.}\enspace\ignorespaces
}

\usepackage{xcolor}

\title{Analyzing Within-Subject Experiments: Identification, Testing, and
Sensitivity\footnote{Authors are listed in alphabetical order.}}

\author{%
  Shiyao Liu\thanks{China Center for Economic Research, Institute of
  South-South Cooperation and Development, National School of Development,
  Peking University, Beijing, China; Governance Lab, Massachusetts Institute
  of Technology, Cambridge, Massachusetts, USA; shiyaoliu@nsd.pku.edu.cn}
  \and Junni Zhang\thanks{China Center for Economic Research, National
  School of Development and Center for Statistical Science, Peking
  University, Beijing, China}
}
\date{}

\begin{document}

%% Literature-review macros (LitReviewZeta*, LitReviewAcrossPaperDelta*) are
%% generated by simu/litreview_appendix_table.R and used in both Section 4
%% (covariance-plausibility paragraph) and Section 5 (Table 4 summary), so
%% they must load before Section 4, not just before Table 4's \input.
\input{tables/litreview_sensitivity_macros}

\maketitle

\begin{abstract}
Recent work encourages political scientists to move from post-only toward within-subject designs for improved precision from repeated measurements. We formalize a potential-outcomes framework for two-period within-subject designs that allows for unequal allocation and heterogeneous treatment and carryover effects. We characterize the pooled estimator and evaluate the carryover test used to justify pooling. We find: first, pooling identifies the average treatment effect only when the gap in the average carryover effects is zero across the two treatment sequences. The unit-clustered standard error for the pooled estimator is identical to its design-based counterpart. Second, under mild conditions, the carryover test has strictly less power than the average-treatment-effect test with post-only data. The resulting two-step procedure, which pools only after a nonrejected test, produces confidence intervals that typically undercover. When the gap is zero, undercoverage occurs if and only if pooling is more efficient than post-only analysis, precisely when the within-subject design is worthwhile. When the gap is nonzero, undercoverage is typical unless the gap or sample size is large. Third, we derive a sensitivity analysis and find published conclusions robust to plausible carryover gaps. We therefore endorse within-subject designs but recommend justifying a zero carryover gap substantively and reporting sensitivity to departures.
\end{abstract}


\section{Introduction}

Researchers in experimental social science work with limited budgets and treatment effect sizes that are often far smaller than those typical of clinical trials, making efficiency a central concern. They have increasingly been advised to adopt repeated-measures
designs, in which each respondent contributes more than one outcome
measurement \parencite{clifford2021increasing}. Two variants are
commonly recommended: in a pre-post design, respondents report the
outcome both before and after treatment. In a within-subject design,
respondents are sequentially exposed to two opposing treatment conditions
and the outcome is measured after each.

This advice has gained traction. \textcite{JORDAN_OLLERENSHAW_TREXLER_2026} identify 83 peer-reviewed studies that cited \textcite{clifford2021increasing} to justify a repeated-measures design in the first four years since its publication, including 33, roughly 40\%, that adopted the within-subject variant. Yet the conventional post-only design, where respondents are exposed to a single treatment, remains the disciplinary default \parencite{diaz2026balancing}. This difference is puzzling because a within-subject design often adds a second measurement for the same respondents at modest cost. Its first round remains a valid post-only experiment, allowing researchers to preserve that analysis while gaining the option to use data from both periods.

The main obstacle is carryover: first-period treatment may change how respondents react to second-period treatment \parencite{willan1986carryover}. Figure~\ref{fig:Replication-of-Study} illustrates why a carryover test does not resolve this concern. The left panel reproduces the post-only and pooled estimates from Study 1 of \textcite{clifford2021increasing}, where the carryover test does not reject a zero carryover gap. In the right panel, we subtract 0.20 scale points from the second-period outcome for respondents assigned treatment and then control, to fabricate a carryover effect, but nothing for the reverse sequence. The same test again fails to reject a zero carryover gap.

\begin{figure}
\centering \includegraphics[width=0.47\textwidth]{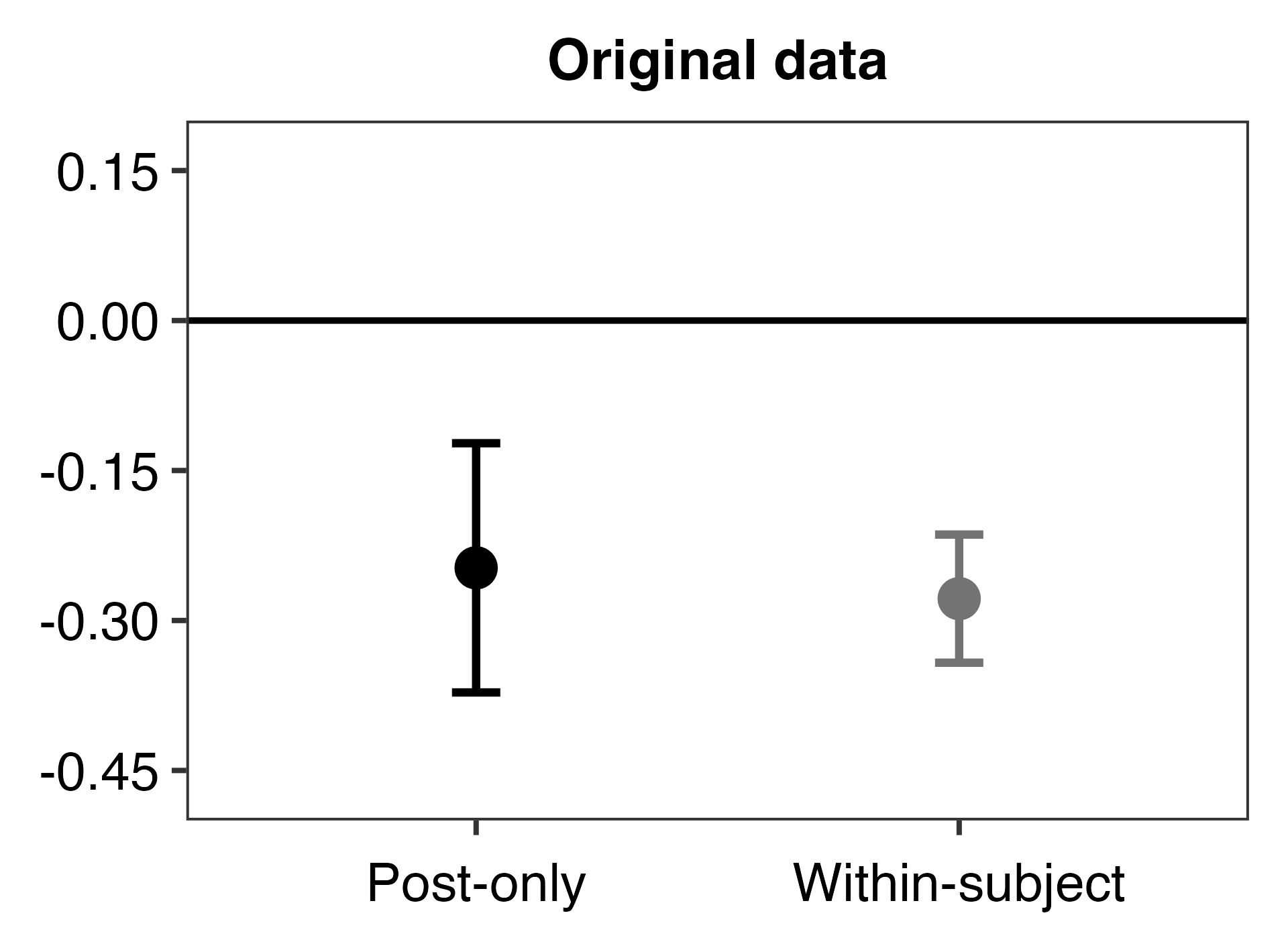}\hfill
\includegraphics[width=0.47\textwidth]{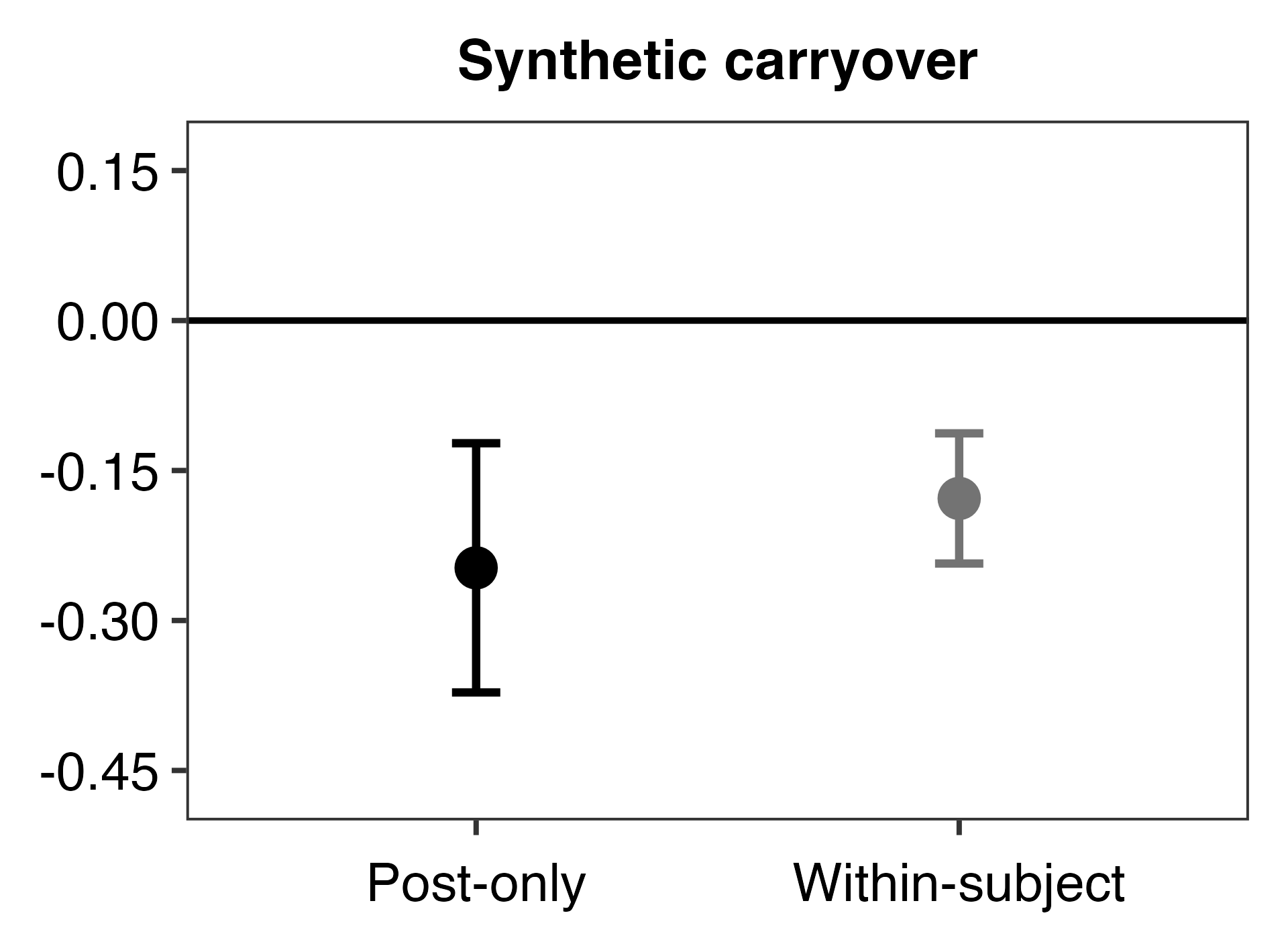}

\caption{Post-only and pooled estimates from Study 1 of \textcite{clifford2021increasing} using the original data (left) and after introducing a synthetic difference in average carryover effects between treatment sequences (right).}\label{fig:Replication-of-Study}
\end{figure}

An analyst therefore pools in both panels, although only the original data, definitely not the synthetic data, may satisfy the identifying condition. In the synthetic data, the test fails to detect the difference in average carryover effects, and the pooled estimator remains biased for the average treatment effect even in large samples. Consistent with this concern, \textcite{JORDAN_OLLERENSHAW_TREXLER_2026} indeed find that repeated-measures estimates are attenuated relative to post-only estimates in substantially larger experiments.

We formalize these concerns. Our first contribution clarifies what pooling requires. In a general potential-outcomes framework for the two-period within-subject design that allows for unequal allocation and heterogeneous treatment and carryover effects, pooling identifies the average treatment effect when the
gap in the average carryover effects is zero across the two treatment sequence. This condition does not require carryover to be zero for every unit, or even on average within either sequence. In fact, it requires only that average carryover be equal across sequences. \textcite{shi2024behavioral} derive the same identifying condition under unequal allocation and heterogeneous carryover effects. We extend the design-based analysis and also study when common regression specifications reproduce the design-based point estimate and how their clustered standard errors compare with the design-based standard error.

Our second contribution evaluates the carryover test and the two-step procedure that pools only when the carryover test fails to reject. Under mild conditions, the carryover test is less powerful than the corresponding average-treatment-effect test with the post-only data. Estimator selection based on the carryover test also produces confidence intervals that typically undercover. Extending \textcite{freeman1989performance} beyond an additive normal model with equal allocation, we show that, when the carryover gap is zero, undercoverage occurs precisely when pooling is more efficient than post-only analysis, the setting in which the within-subject design is most valuable. With a fixed nonzero carryover gap, undercoverage is typical for parameter values common in social science: as the sample size grows, the pooled interval initially preferred by the carryover test narrows around a biased value (the sum of treatment effect and half of the carryover gap), and coverage recovers only when the carryover test reliably selects the post-only analysis.

Our third contribution is a sensitivity analysis that asks how large the carryover gap must be to overturn a conclusion based on pooling. Applying it to published within-subject studies suggests that many conclusions are robust to plausible gaps.

Taken together, our results support the case for within-subject designs while changing how researchers should analyze them. The post hoc carryover test should not determine whether to pool. Researchers should instead justify equality of average carryover across sequences as the identifying assumption, explain why any departure should be substantively small, use the appropriate regression specifications, and report sensitivity to larger departures.

\section{What Does Pooling Estimate in a Within-Subject Design?}\label{sec:Within-Subject-Design-in}

We use Study 1 of \textcite{clifford2021increasing} \parencite{clifford2021data} to illustrate a typical
within-subject design. The experiment revisits the \textcite{smith1987we}
question-wording study: respondents report their support for government
spending on a three-point scale, with the target described as either
``welfare'' or ``assistance to the poor.''

For this application, we define ``welfare'' as the treatment condition, $D=1$, and ``assistance to the poor'' as the control condition, $D=0$. We use “control” only as a label: here, it refers to the alternative wording rather than the absence of an intervention. More generally, $D=1$ can denote any focal condition and $D=0$ its comparator.

Let $n$ denote the total number of respondents. As Figure~\ref{fig:flowchart} illustrates, the experiment randomly
assigns $n_1$ respondents to sequence $(1,0)$, in which they receive the
treatment in Period 1 and the control in Period 2. The remaining
$n_2=n-n_1$ respondents follow sequence $(0,1)$, receiving the control first
and the treatment second. The outcome is measured immediately after each
exposure. The first measurement alone forms a valid post-only experiment.
Using both measurements may improve precision, but it changes what the
resulting estimator identifies when first-period exposure carries over into
the second period.

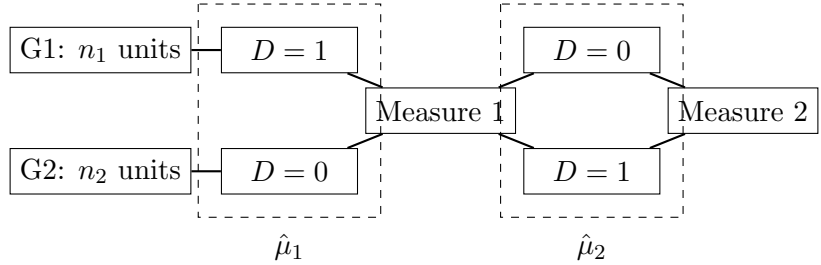
\begin{figure}[h]
\centering
\begin{tikzpicture}[
    > = Stealth,
    every node/.style={rectangle, draw, minimum width=1.8cm, minimum height=0.6cm, align=center},
]
\node (s1) at (0, 0.8) [minimum width=2cm] {G1: $n_1$ units};
\node (s2) at (0,-0.8) [minimum width=2cm] {G2: $n_2$ units};

\node (g11) at (2.5, 0.8) {$D=1$};
\node (g21) at (2.5,-0.8) {$D=0$};
\node (m1)  at (4.5, 0)   {Measure 1};

\node (g12) at (6.5, 0.8) {$D=0$};
\node (g22) at (6.5,-0.8) {$D=1$};
\node (m2)  at (8.5, 0)   {Measure 2};

\draw[dashed] ($(g11.north west)+(-0.3,0.3)$) rectangle ($(g21.south east)+(0.3,-0.3)$);
\node[draw=none] at (2.5, -1.8) {$\hat{\mu}_1$};

\draw[dashed] ($(g12.north west)+(-0.3,0.3)$) rectangle ($(g22.south east)+(0.3,-0.3)$);
\node[draw=none] at (6.5, -1.8) {$\hat{\mu}_2$};

\draw[-, thick] (s1)--(g11);
\draw[-, thick] (s2)--(g21);
\draw[-, thick] (g11)--(m1);
\draw[-, thick] (g21)--(m1);
\draw[-, thick] (m1)--(g12);
\draw[-, thick] (m1)--(g22);
\draw[-, thick] (g12)--(m2);
\draw[-, thick] (g22)--(m2);
\end{tikzpicture}

\caption{The two treatment sequences in a within-subject design}\label{fig:flowchart}
\end{figure}

\paragraph*{Potential Outcomes}

Let $D_{it}\in\{0,1\}$ denote unit $i$'s treatment status in period $t$. There are only two possible treatment sequences: $(D_{i1},D_{i2})=(1,0)$ and $(D_{i1},D_{i2})=(0,1)$.
Under the Stable Unit Treatment Value Assumption (SUTVA;
Assumption~\ref{as:sutva} in \ref{app:assumptions}), write the potential
outcome for unit $i$ in period 
$t$, given that it receives $d_1$ in the first period and $d_2$ in the second, as
\begin{equation}
Y_{it}(d_1,d_2)
=
\pi_{it}+d_t\tau_i+(t-1)\delta_{i,d_1}.
\label{eq:saturated-model}
\end{equation}
Here, $\pi_{it}$ is the baseline outcome, which may vary across units
and periods; $\tau_i$ is the unit treatment effect for $i$ 
relative to the control condition; and $\delta_{i,d_1}$ is unit $i$'s
second-period carryover after receiving $d_1$ in the first period.
Treatment and carryover effects may both vary across units, and
carryover may differ between treatment sequences. Because carryover arises
from prior exposure, it enters only the second-period outcome. Table~\ref{tab:Saturated-Model-of} summarizes the potential outcomes for the two treatment sequences.

\begin{table}[!h]
\begin{centering}
\begin{tabular}{ccc}
\hline
 & Period 1 & Period 2\tabularnewline
\hline
Potential Sequence $(0,1)$ & $Y_{i1}(0,1)=\pi_{i1}$ & $Y_{i2}(0,1)=\pi_{i2}+\tau_i+\delta_{i,0}$\tabularnewline
Potential Sequence $(1,0)$ & $Y_{i1}(1,0)=\pi_{i1}+\tau_i$ & $Y_{i2}(1,0)=\pi_{i2}+\delta_{i,1}$\tabularnewline
\hline
\end{tabular}
\par\end{centering}
\caption{Potential outcomes under the two treatment sequences}\label{tab:Saturated-Model-of}
\end{table}

\paragraph*{What the Estimators Identify}

Let $Y_{it}=Y_{it}(D_{i1},D_{i2})$ denote the observed outcome for unit $i$ in period $t$. Let $\tau\equiv\mathbb E[\tau_i]$ denote the average treatment effect.
Random assignment (Assumption 2 in Appendix 1.1) implies that the first-period difference in average value of $Y_{i1}$ between treatment and control identifies $\tau$. The second-period contrast in average value of $Y_{i2}$ may identify
something different because it also contains the carryover gap between sequences. Define that gap as
\[
\delta
\equiv
\mathbb E[\delta_{i,0}]
-
\mathbb E[\delta_{i,1}].
\]
From Table~\ref{tab:Saturated-Model-of}, the population contrasts in the two periods can be summarized as
\[
\begin{aligned}
\mu_1
&\equiv
\mathbb E[Y_{i1}\mid (D_{i1},D_{i2})=(1,0)]
-
\mathbb E[Y_{i1}\mid (D_{i1},D_{i2})=(0,1)]
=\tau,\\
\mu_2
&\equiv
\mathbb E[Y_{i2}\mid (D_{i1},D_{i2})=(0,1)]
-
\mathbb E[Y_{i2}\mid (D_{i1},D_{i2})=(1,0)]
=\tau+\delta.
\end{aligned}
\]
The pooled population contrast is
\begin{equation}
\mu_{co}\equiv
\frac{1}{2}(\mu_1+\mu_2)
=\tau+\frac{\delta}{2}.
\label{eq:pooled-contrast}
\end{equation}
Thus, using only the first period reproduces a post-only analysis. Using only
the second period adds the full carryover gap. Pooling both periods adds
one-half of that gap. We reserve formal statements and identification proofs in
\ref{app:identification}.

Pooling identifies the average treatment effect if and only if $\delta=0$. Importantly, this condition does not rule out carryover: it requires only that average carryover be equal across the two treatment sequences. Carryover may vary across units, and both sequence averages may be nonzero, provided that $\mathbb E[\delta_{i,0}]=\mathbb E[\delta_{i,1}]$.

For each unit, define the potential treated-minus control contrast as
\[
\Delta_i(1,0)=Y_{i1}(1,0)-Y_{i2}(1,0),
\qquad
\Delta_i(0,1)=Y_{i2}(0,1)-Y_{i1}(0,1).
\]
The corresponding observed treated-minus-control contrast is
\[
\Delta_i
\equiv
\begin{cases}
Y_{i1}-Y_{i2}, & (D_{i1},D_{i2})=(1,0),\\
Y_{i2}-Y_{i1}, & (D_{i1},D_{i2})=(0,1).
\end{cases}
\]
We can also write $\mu_{co}$ as
\[
\mu_{co}=\frac{1}{2}(\mathbb E[\Delta_i\mid (D_{i1},D_{i2})=(1,0)]+\mathbb E[\Delta_i\mid (D_{i1},D_{i2})=(0,1)]
\]
Therefore pooling averages within-unit contrasts from two independent sequence groups.

\paragraph*{Estimation and Asymptotic Behavior}

Let $\hat\mu_1$, $\hat\mu_2$ and $\hat\mu_{co}$ denote the sample analogs of $\mu_1$, $\mu_2$ and $\mu_{co}$. Define the variance scales
\[
\begin{aligned}
\sigma_1^2
&\equiv
\frac{n}{n_1}\mathbb V\!\left\{Y_{i1}(1,0)\right\}
+
\frac{n}{n_2}\mathbb V\!\left\{Y_{i1}(0,1)\right\},\\
\sigma_2^2
&\equiv
\frac{n}{n_2}\mathbb V\!\left\{Y_{i2}(0,1)\right\}
+
\frac{n}{n_1}\mathbb V\!\left\{Y_{i2}(1,0)\right\},\\
\sigma_{co}^2
&\equiv
\frac{n}{n_1}\mathbb V\!\left\{\Delta_i(1,0)\right\}
+
\frac{n}{n_2}\mathbb V\!\left\{\Delta_i(0,1)\right\},
\end{aligned}
\]
where $\mathbb V$ denotes population variances.

Under Assumptions~\ref{as:sutva}--\ref{as:regularity} (\ref{app:assumptions}),
\[
\begin{aligned}
\frac{\sqrt n(\hat\mu_1-\mu_1)}{\sigma_1}
&\overset{d}{\longrightarrow}\mathcal N(0,1),\\
\frac{\sqrt n(\hat\mu_2-\mu_2)}{\sigma_2}
&\overset{d}{\longrightarrow}\mathcal N(0,1),\\
\frac{2\sqrt n(\hat\mu_{co}-\mu_{co})}{\sigma_{co}}
&\overset{d}{\longrightarrow}\mathcal N(0,1).
\end{aligned}
\]

We only sketch the proof for the pooled estimator here and provide full derivations in~\ref{app:sampling}. Random assignment partitions respondents into two disjoint sequence groups, $(1,0)$ and $(0,1)$. The pooled estimator is one-half the sum of the groups’ mean values of $\Delta_i$. Applying the central limit theorem within each group establishes asymptotic normality, while independence across groups allows their variances to be added.

The design-based standard errors follow directly from sample analogs of the
population variances. Each component of $\sigma_1^2$, $\sigma_2^2$, and $\sigma_{co}^2$ is a variance of an observed outcome or within-unit contrast in one of the two sequence groups. All three are
therefore consistently estimable by replacing the population moments with
their corresponding sample moments. \ref{app:sampling} gives the
details.

\paragraph*{When Pooling Improves Precision} Now we compare the efficiency of the pooled estimator ($\hat\mu_{co}$) and the post-only estimator ($\hat\mu_1$), under zero carryover gap, i.e. $\delta=0$. Pooling is more efficient when
\[
\sigma_{co}^2<4\sigma_1^2.
\]

Expanding the variance of the within-unit differences ($\Delta_i(1,0)$ and $\Delta_i(0,1)$), we obtain
\[
\sigma_{co}^2=\sigma_1^2+\sigma_2^2-2\sigma_{12},
\]
where $\sigma_{12}$ combines the cross-period covariances of the two sequence
groups,
\[
\sigma_{12}
\equiv
\frac{n}{n_1}
\operatorname{Cov}\!\left\{Y_{i1}(1,0),Y_{i2}(1,0)\right\}
+
\frac{n}{n_2}
\operatorname{Cov}\!\left\{Y_{i1}(0,1),Y_{i2}(0,1)\right\}.
\]
Therefore, pooling is more efficient when
\[
\sigma_{12}>\frac{1}{2}(\sigma_2^2-3\sigma_1^2).
\]

Positive cross-period covariance lowers $\sigma_{co}^2$ and therefore makes
this condition more likely to hold. By contrast, unusually noisy
second-period outcomes or weak or negative cross-period covariance can
eliminate the precision gain. Moreover, $\sigma_2^2<\sigma_1^2$ is sufficient
for pooling to be more efficient regardless of the cross-period covariance,
since $|\sigma_{12}|\le\sigma_1\sigma_2$
(Proposition~\ref{prop:pooling-sufficient} in \ref{app:sampling}).

In most social-science applications, respondents are believed to behave
similarly across periods, giving positive cross-period covariance, and
Period-2 outcomes are usually comparable in magnitude to Period-1 outcomes.
Both favor pooling: among the papers with significant pooled estimates
reviewed in Section~\ref{sec:Sensitivity-Analysis}
(Table~\ref{tab:litreview-sensitivity-summary}), pooling is more efficient
in all.

Further, this efficiency condition has another important inferential consequence.
Section~\ref{sec:testing-carryover} shows that, when $\delta=0$,
it is exactly the condition under which confidence intervals reported after a
post hoc carryover test undercover.
Before turning to that result, the next section connects the design-based estimators and standard errors to regression specifications researchers
commonly use.

\section{Regression Implementation of Within-Subject Estimators}\label{sec:regression-implementation}

A single within-subject sample supplies both the Period-1 estimator
$\hat\mu_1$ and the pooled estimator $\hat\mu_{co}$. This section explains
how standard regressions can implement these estimators and how their point and
variance estimates relate to the design-based quantities in Section
\ref{sec:Within-Subject-Design-in}.

\paragraph{Research Specifications} Researchers often implement pooling with a treatment-only regression,
\[
y_{ij}=\beta_{0}+\beta_{1}D_{ij}+e_{ij}.
\]
A treatment-only OLS model leaves the within-unit dependence unrestricted when
inference is clustered by unit. A treatment-only random-effects model instead
posits
$e_{ij}=u_{i}+\varepsilon_{ij}$, where $u_{i}$ is a unit-specific
random intercept and $\varepsilon_{ij}$ is idiosyncratic, period-specific
noise. Thus the two approaches impose the same mean specification
but differ in their model for the errors. 

The treatment-only coefficient is a size-weighted average of the two
sequence-specific treated-minus-control contrasts,
$\hat\beta_{1}=\frac{n_{1}}{n}\bar{\Delta}_{(1,0)}+\frac{n_{2}}{n}\bar{\Delta}_{(0,1)}$, where $\bar{\Delta}_{(1,0)}$ and $\bar{\Delta}_{(0,1)}$ are the sample means of $\Delta_i$ for units in treatment sequences (1,0) and (0,1) (\ref{app:bare-pooled}). Under equal allocation, i.e.,  $n_1= n_2$, as in
\textcite{clifford2021increasing} and \gentextcite{JORDAN_OLLERENSHAW_TREXLER_2026},
the treatment-only coefficient equals $\hat\mu_{co}$, and the cluster-robust
variance estimator is weakly conservative relative to the design-based variance.

With unequal allocation, i.e.,  $n_1\ne n_2$, the treatment-only
coefficient therefore recovers $\mu_{co}$ only
under an additional restriction that is difficult to justify\footnote{Formally,
$(n_{1}-n_{2})\mathbb{E}[\pi_{i1}-\pi_{i2}]+n_{2}\mathbb{E}[\delta_{i,0}]-n_{1}\mathbb{E}[\delta_{i,1}]=0$;
see Equation~\eqref{eq:bare-pooled-restriction} in \ref{app:bare-pooled} for
the derivation. This restriction mixes the difference in the two sequence-specific average baseline outcomes with
the two sequence-specific average carryover effects, making it more
difficult to defend than the identifying condition $\delta=0$ in Section
\ref{sec:Within-Subject-Design-in}.}. An allocation-robust implementation is to use a period-interacted regression that includes the period dummy and its interaction with the treatment indicator,
\begin{equation}
y_{ij}=\beta_{0}+\beta_{1}D_{ij}+\beta_{2}I\left(j=2\right)+\beta_{3}D_{ij}I\left(j=2\right)+e_{ij},
\label{eq:period-interacted}
\end{equation}
and to use standard errors clustered at the unit level.

With the above regression specification, our findings are as follows. First, the connections between the coefficient estimates and the treatment effect estimators are:
\[
\hat\mu_1=\hat\beta_1,\quad \hat\mu_2=\hat\beta_1+\hat\beta_3,\quad
\hat\mu_{co}=\hat\beta_1+\tfrac12\hat\beta_3.
\]
Second, regardless of allocation balance, the unadjusted cluster-robust sandwich variance estimators (with clustering at the unit level) equal the design-based variances. Third, the OLS and random-effects regression models have the same coefficient estimates and the same unadjusted cluster-robust sandwich variance estimators. Software-specific small-sample corrections may apply to the variance estimators.

We reserve analytical details in \ref{app:bare-pooled} and \ref{app:saturated}, and use the within-subject arm of \gentextcite{clifford2021increasing} Study 1 as an illustration.

\paragraph{Illustration of Pooled Regressions}

Table \ref{tab:Comparison-of-Sample} compares design-based estimation, the paired $t$-test (labeled (0)), four pooled regression specifications (treatment-only OLS,
treatment-only random effects, period-interacted OLS, and
period-interacted random effects, labeled (1)-(4)), and an OLS that uses only period-1 data (labeled (5)). The comparison is based on two versions of
the within-subject arm of Study 1 in \citet{clifford2021increasing}.

Panel A uses the original data after dropping one unit observed
in only one period, leaving a nearly balanced allocation ($n_1=227$,
$n_2=226$). The paired $t$-test and the four pooled regressions closely track the design-based
point estimate ($-0.278$) and standard error ($0.033$): the
period-interacted specifications, (3) and (4), match exactly, while
the paired $t$-test and the two treatment-only specifications, (0)--(2),
differ by only $0.02\%$ and $0.08\%$ due to the one-unit imbalance.

Panel B duplicates
units assigned to $(1,0)$ to create a $2{:}1$ allocation ($n_1=454$,
$n_2=226$). The gap becomes substantive. the paired $t$-test and
treatment-only specifications now differ from the design-based point
estimate by $3.27\%$ and from its standard error by $4.81\%$, while the
period-interacted specifications remain exact regardless of allocation balance.

In both panels, the design-based standard error is smaller than that of the Period-1-only estimator, demonstrating that the efficiency roughly doubles with the within-subject design.

The simulations in \ref{app:regression-simulation} vary the carryover
gap and the allocation ratio, and further confirm the results. 
Period-interacted OLS and random effects with unit-clustered standard
errors reproduce the design-based estimator in every case, while the
paired $t$-test and treatment-only specifications match it only under
equal allocation. 

We therefore recommend period-interacted OLS or random effects
with unit-clustered standard errors as the workhorse implementation, when
researchers are willing to assume no gap in average carryover effects across
sequences ($\delta=0$). When $\delta\ne0$, the pooled estimator is biased
and its confidence intervals are invalid (Table~\ref{tab:regression-specs-coverage}
in \ref{app:regression-simulation}). The robustness of the treatment
effect's sign to a nonzero $\delta$ is discussed in
Section~\ref{sec:Sensitivity-Analysis}.

\begin{table}[!h]
\begin{centering}
\input{tables/spec_comparison}
\par\end{centering}
\caption{Comparison of Sample Estimators with a Reanalysis of \cite{clifford2021increasing}}\label{tab:Comparison-of-Sample}

\end{table}

\section{Carryover Test: Power and the Two-Step Procedure}\label{sec:testing-carryover}

Section~\ref{sec:Within-Subject-Design-in} shows that pooling identifies the
average treatment effect only when the carryover gap $\delta=0$.
\textcite{clifford2021increasing} and \textcite{JORDAN_OLLERENSHAW_TREXLER_2026}
test this condition by comparing the pooled estimate against an independent
post-only sample. This section instead asks whether the within-subject
sample can test the condition on its own, and how that test's power compares
with the average-treatment-effect test using only the first period of the same within-subject sample, i.e., the data that a post-only design would have collected.

We show that, under plausible conditions, the carryover test for $\delta=0$ is the less
powerful of the two, sharpening the dilemma facing researchers who adopt a
within-subject design for its efficiency gain: confirming that pooling is
warranted with adequate power takes, if anything, a larger sample than the
post-only design needs to detect the treatment effect itself with the same
power. Put differently: by the time a study has enough power for its own
carryover test to be trustworthy, that same study already has enough power
to detect the treatment effect directly, without pooling. 

A researcher who
reports a null carryover test and treats it as license to pool is therefore
in one of two positions: either the test was underpowered, in which case the
null result carries little information about whether carryover is actually
zero, or the test was adequately powered, in which case the study never
needed the within-subject design's efficiency gain to begin with. Either
way, the null carryover test cannot do the justificatory work researchers
typically ask of it\footnote{\textcite{brown1980crossover} reaches the same
conclusion from a different angle: in a classical normal-theory
analysis-of-variance model, the sample size that the carryover test needs
to detect a carryover effect large enough to meaningfully bias the treatment
estimate is, in his worked numerical example, ten times what a post-only
design needs for the treatment effect itself. Our result differs in being a general
conclusion for values of $n_1$, $n_2$, $\delta$ and $\mu_1$ that are typical in most social-science applications, derived
in a potential outcome framework with design-based standard errors, rather
than a sample-size ratio computed under one chosen calibration of an assumed
random-effects model.}.

\paragraph{Test for $\delta$}To test $H_0:\delta=0$, it is natural to use $\hat\delta\equiv\hat\mu_2-\hat\mu_1$,
since $\delta=\mu_2-\mu_1$ (Section~\ref{sec:Within-Subject-Design-in}).
Proposition~\ref{prop:delta-asymptotic} in \ref{app:sampling} shows that
\[
\frac{\sqrt n(\hat\delta-\delta)}{\sigma_{\delta}}\overset{d}{\longrightarrow}
\mathcal N(0,1).
\]
Therefore we reject $H_0$ when
$|\sqrt{n}\hat\delta/\hat{\sigma}_{\delta}|>z_{1-\alpha/2}$, where
$\hat{\sigma}_{\delta}$ is the sample analog of $\sigma_{\delta}$ derived in
\ref{app:sampling}, and $z_{1-\alpha/2}$ is the $1-\alpha/2$ quantile of the standard normal distribution.

In practice, $\hat\delta=\hat\beta_3$ from the period-interacted regression
of Equation~\eqref{eq:period-interacted},
noting $\hat\mu_2-\hat\mu_1=(\hat\beta_1+\hat\beta_3)-\hat\beta_1$. Researchers therefore can equivalently implement the carryover test as a $t$-test of the interaction term being zero, using a period-interacted OLS with clustered standard errors. We reserve the equivalence results in \ref{app:saturated}.

\paragraph{Power of the Carryover Test versus the Average-Treatment-Effect Test}

Ignoring the remainder term, the carryover test's power is
\[
\psi_{\delta}\left(\delta\right)\approx
1-\Phi\!\left(z_{1-\alpha/2}-\frac{\sqrt{n}\delta}{\sigma_{\delta}}\right)+\Phi\left(-z_{1-\alpha/2}-\frac{\sqrt{n}\delta}{\sigma_{\delta}}\right).
\]
The power of testing $H_{0}:\mu_{1}=0$, the average-treatment-effect test with
post-only data, is
\[
\psi_{1}\left(\mu_{1}\right)\approx
1-\Phi\!\left(z_{1-\alpha/2}-\frac{\sqrt{n}\mu_1}{\sigma_1}\right)+\Phi\left(-z_{1-\alpha/2}-\frac{\sqrt{n}\mu_1}{\sigma_1}\right).
\]
Because two-sided normal power is increasing in the magnitude of the
noncentrality parameter (Proposition~\ref{prop:power-comparison} in
\ref{app:power}), comparing $\psi_{\delta}$ and $\psi_{1}$ amounts to
comparing the magnitude of their noncentrality parameters,
\[
\frac{|\delta|}{\sigma_{\delta}}\quad\text{versus}\quad\frac{|\mu_{1}|}{\sigma_{1}}.
\]

Two conditions make the power of the average-treatment-effect test larger. First,
$\sigma_\delta>\sigma_1$ whenever the cross-period
covariance is nonnegative within each sequence, the expected case in most social-science applications. That is, a unit's potential outcomes across periods is believed to be positively correlated under different stimuli, since both reflect that unit's own underlying
disposition. Our review of published studies also confirm this scenario: across the \LitReviewZetaCount{} recoverable study-outcomes (see \ref{app:litreview-parameters}), the
standardized cross-period covariance $\sigma_{12}/\sigma_{1}^{2}$ ranges from
\LitReviewZetaMin{} to \LitReviewZetaMax{} (10th--90th percentile:
\LitReviewZetaPTen{} to \LitReviewZetaPNinety{}), with a median of
\LitReviewZetaMedian{}, and is nonnegative in \LitReviewZetaNonnegCount{}
of the \LitReviewZetaCount{}.

Second, it is commonly presumed
that the between-sequence average carryover gap is smaller in magnitude than
the average treatment effect: $|\delta|\le|\mu_{1}|$, equivalently
$|\delta|\le|\mathbb{E}[\tau_{i}]|$. Carryover is a partial, attenuated echo
of the treatment's influence into the following period, not a force expected
to exceed the treatment effect that produced it. Further, $\delta$ is the average carryover
gap between sequences, which nets out whatever the two sequences' carryover
levels share in common. This is broadly consistent
with \textcite{JORDAN_OLLERENSHAW_TREXLER_2026}'s own meta-analytic finding
of roughly 20\% attenuation in repeated-measures designs.

Combining both conditions, with the same size of respondents $n$,
\[
\frac{|\delta|}{\sigma_{\delta}}\ \le\ \frac{|\mu_{1}|}{\sigma_{\delta}}\ <\ \frac{|\mu_{1}|}{\sigma_{1}}.
\]
Hence, $\psi_{\delta}\left(\delta\right)<\psi_{1}\left(\mu_{1}\right)$: under
these two stated conditions, the carryover test is less powerful than the test for the average
treatment effect, as claimed above.

\begin{figure}[htbp] % [htbp] controls where the image is placed (here, top, bottom, page)
    \centering
    \includegraphics[width=0.4\textwidth]{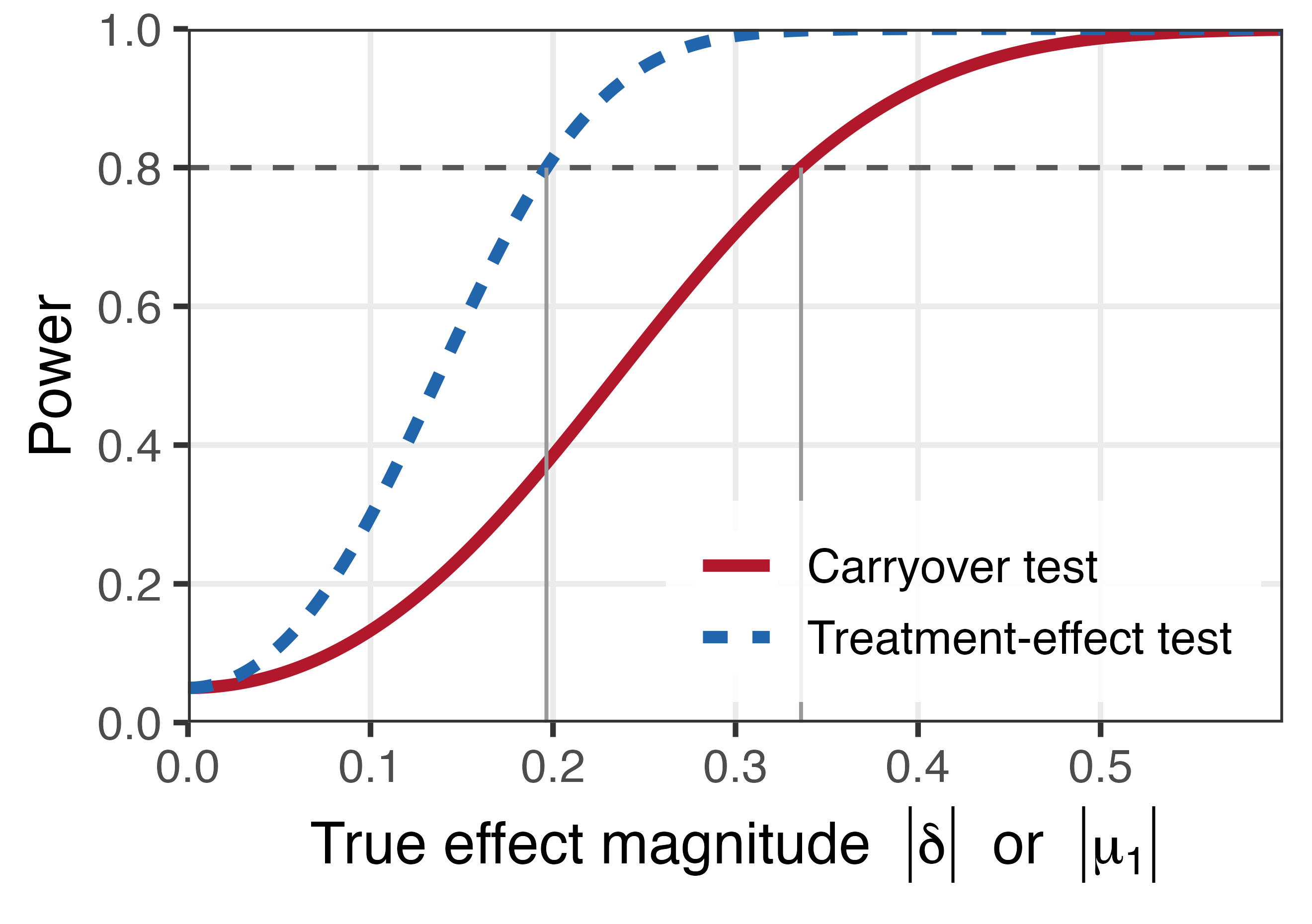} % Omit .png extension if you prefer
    \caption{Power of the Carryover Test and the Treatment-Effect Test, Calibrated to
    Study 1 of \textcite{clifford2021increasing}}
    \label{fig:power_function}
\end{figure}
%{\color{red}\small[TK-REVIEW]} hand verified in Stata

Figure \ref{fig:power_function} calibrated the above results to the
within-subject arm of Study 1 in \gentextcite{clifford2021increasing}. Because
$\sigma_{\delta}/\sigma_{1}=1.71$ there, a carryover gap must be $1.71$ times
as large as a treatment effect for the two tests to have equal power to detect their respective effects: at $80\%$ power the smallest detectable treatment effect is $0.20$ scale
points, while the smallest detectable carryover gap is $0.34$. The fact that the blue
curve (the treatment-effect test for $\mu_1=0$) lies above the red curve (the carryover
test for $\delta=0$) throughout the figure shows that this power ordering holds at every
effect magnitude.

Nevertheless, researchers may base their decision on whether to pool on the
significance of the carryover test for $\delta$. We show this practice to
be problematic even apart from the power issue.

\paragraph{The Two-step Procedure}

The two-step procedure works as follows. First test
$H_{0}:\delta=0$. If the test fails to reject, pool the two periods and report $\hat{\mu}_{co}$. If it rejects, fall back to the post-only estimator $\hat{\mu}_{1}$ instead. In either case, researchers report that estimator's
ordinary confidence interval with no adjustment for having been chosen by a
preliminary test. The formal procedure is given in \ref{app:grizzle}.

This is exactly the
procedure that \citet{willan1986carryover} examine under the name Grizzle's
procedure, and the literature has already shown it to be problematic.
\textcite{freeman1989performance} finds that the reported interval undercovers
its nominal level even when there is no differential carryover, and
considerably more so once carryover is present, and concludes that the
two-step analysis is too misleading to be of practical use. Freeman's
calculation assumes an additive normal model with homoskedasticity and equal allocation. We
extend it to allow non-normal outcomes, heteroskedasticity, and unequal
allocation in \ref{app:grizzle} for theory and \ref{app:undercoverage-simulation} for simulation.

To illustrate, in Figure~\ref{fig:csp-coverage}, we demonstrate the performance of the two-step procedure with parameters calibrated from the within-subject arm of \gentextcite{clifford2021increasing} Study
1 (details in \ref{app:csp-calibration})\footnote{The
undercoverage pattern is not unique to this study: Figure~\ref{fig:literature-two-step-coverage}
in \ref{app:undercoverage-simulation} shows a similar pattern for other papers
we review.}. The two-step interval covers
with probability $0.924$ at $\delta=0$, falling to as little as
$0.433$ at a carryover gap of just $0.19$ points on their three-point scale.
At the estimated gap of $|\hat\delta|=0.02$, though noisy and not
significantly different from zero, the coverage probability is $0.913$, lower than the nominal
level of 95\%.

\begin{figure}[htbp]
    \centering
    \includegraphics[width=0.4\textwidth]{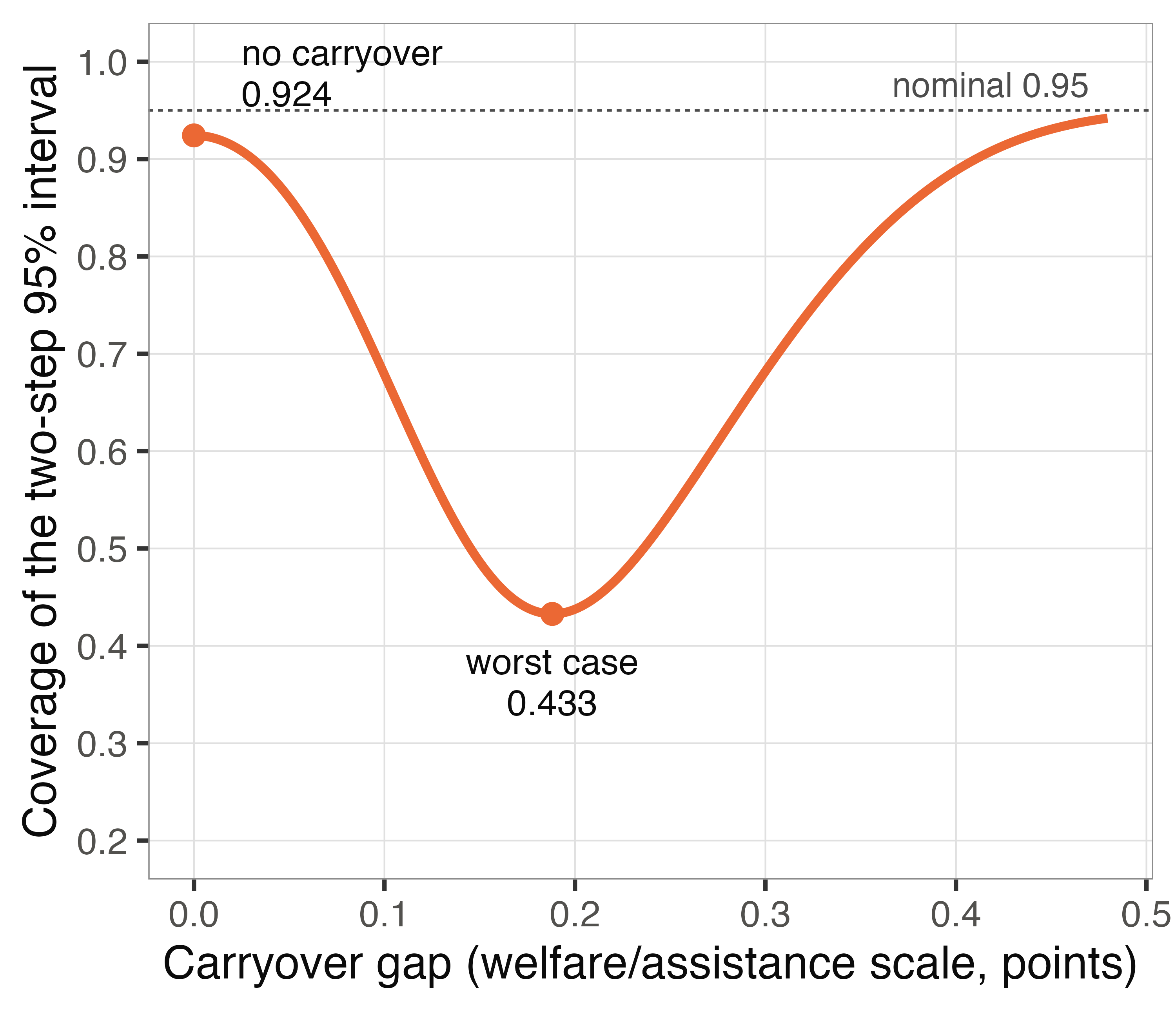}
    \caption{Coverage of the Two-Step Interval, Calibrated to the
    Within-Subject Arm of Study 1 of \textcite{clifford2021increasing}}
    \label{fig:csp-coverage}
\end{figure}

There are two sources of the shortfall. First, if the truth is $\delta\neq0$, the test itself may not be powerful enough to detect it, and thus leads to the wrong specification. Second, and independently, reporting
an estimator's ordinary confidence interval after it has been chosen by a
statistical test beforehand is invalid on its own terms: selection changes what the
interval actually covers.

The second source alone is enough to produce undercoverage. At $\delta=0$, the reported interval still undercovers its nominal level, exactly when pooling is more efficient
than the post-only estimator, i.e. $\sigma_{co}^2<4\sigma_1^2$. \ref{app:grizzle} shows the formal results. Away from $\delta=0$, the calibrations relevant here usually yield
undercoverage. \ref{app:undercoverage-simulation} conducts simulations based on calibration to published social science studies.

Finally, Figure~\ref{fig:jot-coverage-n} examines the two-step procedure's performance
as $n$ grows, holding fixed the carryover gap $\delta=-0.119$ calibrated to
Study 4 of \textcite{JORDAN_OLLERENSHAW_TREXLER_2026} \parencite{jordan2026data} (the equal-weighted
average across their AmeriSpeak, Prolific, and Lucid samples, as in
Table~\ref{tab:litreview-sensitivity-summary}), and comparing it against
always reporting the post-only or the pooled estimator.

\begin{figure}[htbp]
    \centering
    \includegraphics[width=0.66\textwidth]{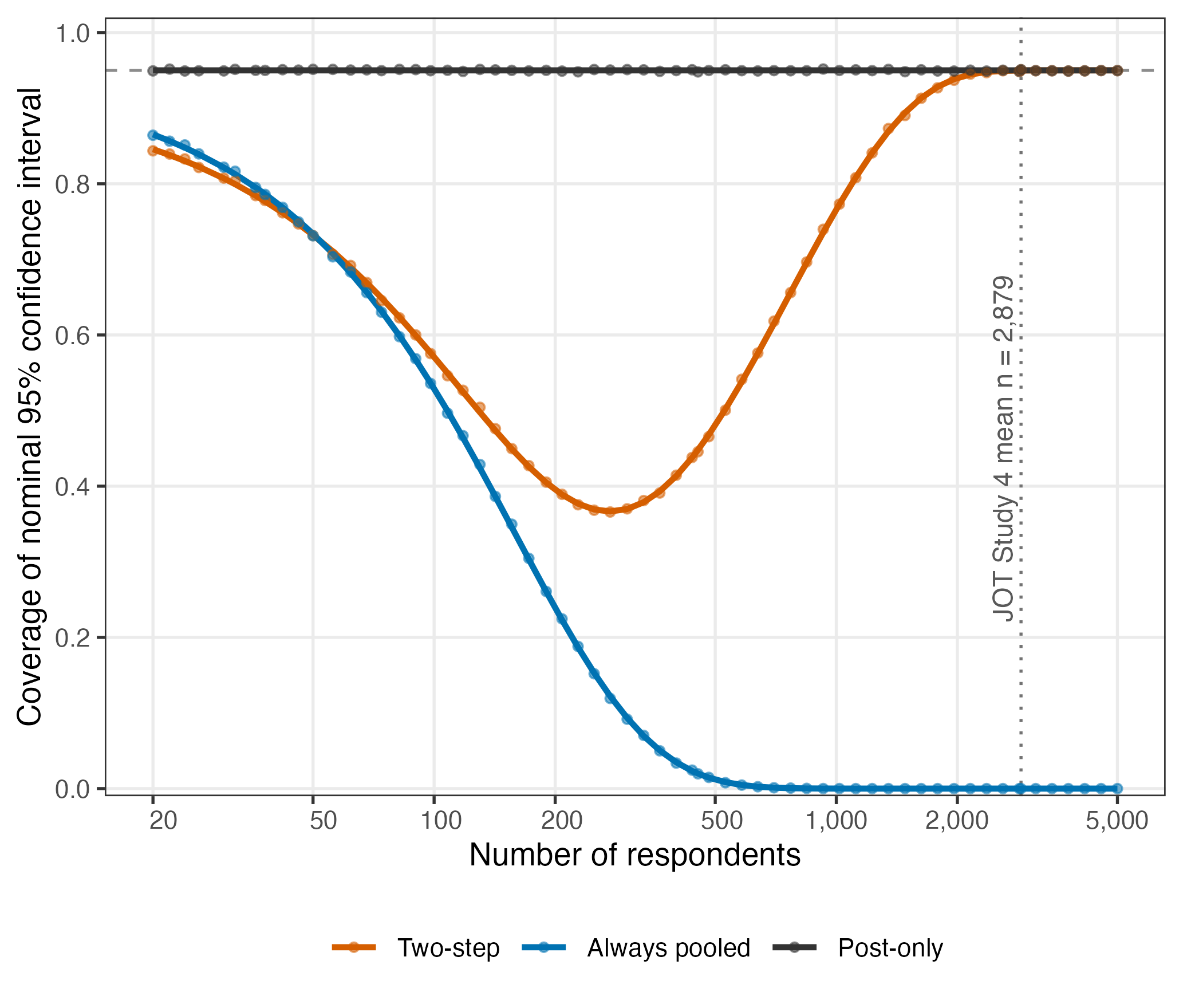}
    \caption{Coverage as Sample Size Increases under the Study 4 Carryover
    Calibration of \textcite{JORDAN_OLLERENSHAW_TREXLER_2026}. Points are based
    on 50,000 Monte Carlo replications at each sample size, and solid curves are
    the corresponding analytic coverage probabilities. The horizontal dashed
    line marks nominal 95\% coverage; the vertical dotted line marks the mean
    complete-pair sample size across the three Study 4 samples.}
    \label{fig:jot-coverage-n}
\end{figure}

The two-step curve is U-shaped. Its coverage initially falls to about $0.37$,
then returns to the nominal level only because the carryover test eventually
rejects with probability approaching one and the procedure therefore reports
the post-only estimator. At the mean Study 4 sample size of $2{,}879$, the
simulation selects the post-only estimator more than 99\% of the time. If the
within-subject arm of \gentextcite{clifford2021increasing} Study 1 had instead
shared this JOT Study 4 carryover and variance structure, its sample size of
$227+226=453$ would yield two-step coverage of only $0.448$.

\paragraph{The Problem with Always Pooling}  By contrast,
sample size alone does not validate pooling. As $n$ grows, coverage of the
pooled interval approaches zero. The reason is that, when $\delta\neq0$, the pooled
estimator becomes increasingly precise for $\tau+\delta/2$ rather than for
$\tau$. Its confidence interval therefore concentrates more tightly
around the wrong target, and its probability of covering $\tau$ vanishes.

Taken together, a large sample rescues the two-step procedure only by
abandoning the pooled estimator and its precision gain, rather than making
conventional pooled
inference valid. A further discussion can be found as Remark~\ref{rem:fixed-delta} in \ref{app:grizzle}. 

% For Appendix
% Recalling that $\hat{\mu}_{1}=\frac{1}{n_{1}}\sum^{n_{1}}_{i=1}Y_{i1}-\frac{1}{n_{2}}\sum^{n_{1}+n_{2}}_{j=n_{1}+1}Y_{j1}$ and 
% $\hat{\mu}_{2}=\frac{1}{n_{2}}\sum^{n_{1}+n_{2}}_{j=n_{1}+1}Y_{j2}-\frac{1}{n_{1}}\sum^{n_{1}}_{i=1}Y_{i2}$, we have
% \begin{align*}
% \hat{\delta} & \equiv\hat{\mu}_{2}-\hat{\mu}_{1}\\
%  & =\frac{1}{n_{2}}\left(\sum^{n_{1}+n_{2}}_{j=n_{1}+1}Y_{j1}+\sum^{n_{1}+n_{2}}_{j=n_{1}+1}Y_{j2}\right)-\frac{1}{n_{1}}\left(\sum^{n_{1}}_{i=1}Y_{i1}+\sum^{n_{1}}_{i=1}Y_{i2}\right)
% \end{align*}

% So $\mathbb{E}\left[\hat{\delta}\right]=\mu_{2}-\mu_{1}$, and $\mathbb{V}\left(\hat{\delta}\right)=\mathbb{V}\left(\frac{1}{n_{2}}\left(\sum^{n_{1}+n_{2}}_{j=n_{1}+1}Y_{j1}+\sum^{n_{1}+n_{2}}_{j=n_{1}+1}Y_{j2}\right)\right)+\mathbb{V}\left(\frac{1}{n_{1}}\left(\sum^{n_{1}}_{i=1}Y_{i1}+\sum^{n_{1}}_{i=1}Y_{i2}\right)\right)$because
% of the independence of sequence assignments. Independence between
% units lead to 
% \[
% \hat{\delta}\overset{d}{\rightarrow}\mathcal{N}\left(\delta,\sigma^{2}_{\delta}\right);\ \ \sigma^{2}_{\delta}\equiv
% \]

\section{Sensitivity to Unobserved Carryover Effects}\label{sec:Sensitivity-Analysis}

Section 4 shows that the carryover test can have little power, and that the two-step procedure it motivates can
produce confidence intervals that do not achieve nominal coverage. In this Section, we recommend defending the pooled estimator with a sensitivity
analysis instead. 

If the two sequences differ in
average carryover by $\delta$, the pooled estimate is biased by $\delta/2$
(Equation~\eqref{eq:pooled-contrast}). Treating $\delta$ as a
sensitivity parameter, we ask how large this gap would need to be, relative
to the pooled estimate, to overturn an originally significant result. In our
review of published social-science applications, we find the sign of the estimated treatment
effect is generally robust to relatively large average carryover gaps.

\paragraph{Sensitivity Analysis for the Pooled Estimator}

Under the hypothesis that the carryover gap equals $\delta^\star$,
the bias-adjusted estimator is
$\hat\tau(\delta^\star)\equiv\hat\mu_{co}-\delta^\star/2$, with standard error
$s_{co}\equiv\widehat{\mathrm{se}}(\hat\mu_{co})$ inherited from $\hat\mu_{co}$
itself. Standardizing $\delta^\star$ relative to the pooled estimate, i.e., letting $\delta_{\mathrm{std}}^{\star}\equiv\delta^{\star}/\hat{\mu}_{co}$, and
writing $t_{co}\equiv\hat\mu_{co}/s_{co}$, a two-sided test of $H_0:\tau=0$
using $\hat\tau(\delta^\star)$ fails to reject over the interval
\[
\delta_{\mathrm{std}}^\star\in\left[2-\frac{2z_{1-\alpha/2}}{|t_{co}|},\ 2+\frac{2z_{1-\alpha/2}}{|t_{co}|}\right]
\]
(proof in \ref{app:sensitivity-derivation}). For an
originally significant pooled estimate $\hat\mu_{co}$ with a $t$-statistic $|t_{co}|>z_{1-\alpha/2}$, the
standardized minimum carryover gap required to overturn significance is the
interval's lower endpoint,
\[
\delta_{\min,\mathrm{std}}^*
\equiv
2\left(1-\frac{z_{1-\alpha/2}}{|t_{co}|}\right).
\]
This is the smallest same-direction carryover gap, expressed as a multiple of
the pooled estimator, for which the bias-adjusted confidence interval includes
zero. 

For example, at the 5\% level, if $|t_{co}|=6$, then
$\delta_{\min,\mathrm{std}}^*=1.35$, and the full non-rejection region is
$[1.35,2.65]$. The lower endpoint, 1.35, is the smallest same-direction gap that removes
statistical significance. That is, $\delta$, the difference in average carryover effects across sequences, needs to be at least 1.35 times as large as the estimated $\hat\mu_{co}$ for the bias-adjusted estimator to be statistically indifferent from zero.

\paragraph{Calibrating the Sensitivity Parameter}

To calibrate plausible values of $\delta$, we review the 33 papers of
within-subject design identified by \textcite{JORDAN_OLLERENSHAW_TREXLER_2026}
as citing \textcite{clifford2021increasing}. Where the design and data recover
both period-specific treatment contrasts, we calculate
$\hat\delta=\hat\mu_2-\hat\mu_1$ and express it relative to the pooled
estimate for comparison with $\delta_{\min,\mathrm{std}}^*$. 

Table~\ref{tab:litreview-design-classes} organizes the 33 papers into four
design classes. Group 1 contains exact or directly reducible two-sequence
within-subject designs. Group 2 contains multiperiod or multicondition designs that we reduce via pairwise comparisons. We pair the first period with each later period. A pair contributes only when both treatment sequences are present, i.e. some respondents received treatment in Period 1 and control in the later period (sequence $(1,0)$) while others received the reverse (sequence $(0,1)$), so we can form the same two-sequence contrast used throughout the paper. We average with equal weight over all such pairs. We call these reductions history-marginal because a respondent's outcome in the later period is averaged over whatever treatments they received in between. Group 3
contains within-subject designs whose Period 1 exposure is fixed across respondents. Everyone receives the same (or no) stimulus in Period 1, and only in Period 2 are they assigned to different treatment conditions. Group 3, further elaborated later, is not subject to the differential-carryover concern that motivates our analysis, since every respondent shares the same Period 1 exposure. It remains a within-subject design, but it lacks the two-sequence structure that our analysis addresses.  Group 4 contains the remaining
non-reducible or unresolved designs, most often for lack of usable replication data, such as a covariate recording the order in which treatments were assigned\footnote{A small number of Group 4 papers fall outside this core reason instead: some have paired items answered in one sitting rather than at two temporally separated occasions, so a two-sequence estimate is unsupported despite the data being available; others are papers that \textcite{JORDAN_OLLERENSHAW_TREXLER_2026}'s own table misclassifies as within-subject, which are out of scope for this review rather than non-reducible on structural terms.}. Only Groups 1 and 2 have the two-sequence
structure needed to compute $\hat\delta$ and $\delta_{\min,\mathrm{std}}^*$.  \ref{app:litreview-parameters} gives the full classification
rules.
\input{tables/litreview_design_classes}

Researchers need to defend a pooled estimator only when it is itself significant, so we restrict the sensitivity
analysis to $\hat\mu_{co}$ with $|t_{co}|>z_{.975}$. Six papers from Groups 1 and 2 meet this criterion, each contributing one qualifying study. We supplement them with two further papers, \textcite{clifford2021increasing}'s Study 1 and \textcite{JORDAN_OLLERENSHAW_TREXLER_2026}'s Studies 4-6. Ten exact
two-sequence within-subject designs are reconstructed from the latter three studies. Altogether the analysis covers eight papers, 10 studies, and 25 significant treatment-control comparisons across 24 distinct pairs of treatment sequences\footnote{One pair yields two significant treatment-control comparisons.}, averaged with equal weight within each paper.

Table~\ref{tab:litreview-sensitivity-summary}
reports the within-paper averages of $\hat\delta/\hat\mu_{co}$, its absolute
value $|\hat\delta/\hat\mu_{co}|$, and $\delta_{\min,\mathrm{std}}^*$. Giving each paper equal weight, the carryover ratio averages
\LitReviewAcrossPaperDeltaRatio{} in absolute value, well below
\LitReviewAcrossPaperDeltaMinStd{}, the average threshold that would be needed to overturn
significance. The sign of $\hat\delta/\hat\mu_{co}$ is also informative: a negative ratio guarantees that pooling has attenuated
the post-only estimate\footnote{Writing
$x\equiv\hat\delta/(2\hat\mu_{co})$, the identity
$\hat\mu_{co}=\hat\mu_1+\tfrac12\hat\delta$ gives $\hat\mu_1=\hat\mu_{co}(1-x)$,
so $|\hat\mu_1|>|\hat\mu_{co}|$ whenever $x<0$, i.e.\ whenever the ratio is
negative. A positive ratio has no equally clean converse.
$|\hat\mu_1|<|\hat\mu_{co}|$ only while $0<x<2$, and reverts to
$|\hat\mu_1|>|\hat\mu_{co}|$ (with a sign flip) while $x>2$.}. In 5 of the 8
papers the ratio is negative, so in exactly those cases, a researcher who reports $\hat\mu_{co}$
should be worried about the carryover gap, unless they also conduct our sensitivity analysis.

\input{tables/litreview_sensitivity_summary}
\FloatBarrier

The coverage failure of the pooled estimator in Section~\ref{sec:testing-carryover} does not contradict the sensitivity analysis in the current section: the former asks whether the interval contains $\tau$, whereas the latter asks whether a plausible carryover gap could overturn a significant estimated effect. As the pooled estimate becomes more precise with larger $n$, its sign can be robust even when coverage is not nominal. To be clear, this does not restore point identification or guarantee correct coverage: the sensitivity analysis only shows how large the carryover gap would need to be before the estimated effect is no longer distinguishable from zero in its original direction, enough to answer a qualitative question about the effect's sign, even without identification or valid coverage.

We therefore recommend that researchers report $\delta_{\min,\mathrm{std}}^*$ rather than treat a failed carryover test as evidence that $\delta=0$. Across the studies we review, pooled estimates are generally robust by this criterion, supporting the recommendation of \textcite{clifford2021increasing} and \textcite{JORDAN_OLLERENSHAW_TREXLER_2026} that repeated-measures designs can improve precision.

\section{Practical Recommendations and Concluding Remarks}\label{sec:Practical-Advice}
We conclude with practical guidance for researchers using within-subject
designs, organized around three decisions: design, estimation, and
interpretation.

\advice{Design: Collect the second period when it adds useful precision}
Researchers should use pilot data or substantive knowledge to assess the
likely precision gain. Pooling is more efficient only when
$\sigma_{12}>\frac{1}{2}(\sigma_2^2-3\sigma_1^2)$, a sufficient condition of which is $\sigma_2^2<\sigma_1^2$, regardless of the covariance
(Section~\ref{sec:Within-Subject-Design-in}, \ref{app:sampling}).

When the efficiency cannot be assessed before
fielding, we generally recommend collecting it anyway. Period 1 remains
a valid post-only fallback at often-modest marginal cost, so researchers can
use the pooled estimator when it is robust to
plausible carryover and is proved to be more efficient, and otherwise report the Period-1 estimate.

\advice{Estimation: Use the period-interacted regression with
unit-clustered standard errors}
Estimate the period-interacted regression of
Equation~\eqref{eq:period-interacted} and construct the pooled estimator as
$\hat\mu_{co}=\hat\beta_1+\tfrac12\hat\beta_3$. 

This specification reproduces
the design-based point estimate and unit-clustered standard error under
any allocation, not confined to the case where $n_1=n_2$. OLS and random effects give the same coefficient map and
clustered covariance, so only the choice of unit-clustered standard errors matters (Section~\ref{sec:regression-implementation},
\ref{app:bare-pooled}, \ref{app:saturated}). 

Do not let the carryover test
determine whether to pool. The test of $H_0:\delta=0$ is no more powerful,
under plausible conditions, than the Period-1 average treatment-effect test, so by
the time it can reliably detect a consequential gap the post-only fallback
would already detect the treatment effect directly
(Section~\ref{sec:testing-carryover}). Conditioning the reported
estimator on this preliminary test creates its own problem, since the
resulting two-step interval undercovers exactly when pooling is more
efficient, even at $\delta=0$ (\ref{app:grizzle}).

\advice{Interpretation: Defend or assess a small carryover gap}

The identification condition for the treatment effect with a pooled estimator is not to establish whether any individual carryover
effect exists, or even whether either sequence has a nonzero average
carryover effect. It is more modest: assess whether the \emph{average} carryover effects for the two sequences, $E[\delta_{i,0}]$ and $E[\delta_{i,1}]$, are likely to be different. Individual effects
and both sequence averages may be nonzero, heterogeneous, or difficult to
estimate, but their difference could be plausibly negligible.

Group 3 in Section~\ref{sec:Sensitivity-Analysis} illustrates why this
average-difference claim may be easier to defend than zero carryover for every unit: all
respondents share the same pretreatment measurement history, so its average
effect may plausibly be common across treatment arms even though individual
or arm-specific average effects are hard to evaluate. Suitable conjoint
designs offer a parallel case: pooling a respondent's multiple choice tasks
requires the same stability and no-carryover-effects assumption, tested by
comparing AMCEs across tasks \parencite[Sec.~5.3.1]{hainmueller2014causal}.
Individual respondents' choices may fluctuate across tasks even for the
same set of profiles, but identification may only require that this
fluctuation not differ, on average, between respondents randomly assigned
to see the tasks in different orders.
Neither design feature identifies the gap by itself, but each can support a
transparent argument for a negligible difference.

Whether this matters depends on what the researcher wants from the estimate.
If the goal is the \emph{magnitude} of $\tau$, a nonzero carryover gap breaks
identification outright. Whenever $\delta\ne0$, the pooled estimator
converges to $\tau+\delta/2$ rather than $\tau$, however large the sample
(Section~\ref{sec:testing-carryover}). If the goal is only the
\emph{sign} of the treatment effect, the within-subject design remains
useful, and our review finds this sign is generally robust to plausible
carryover in practice (Section~\ref{sec:Sensitivity-Analysis}). Accordingly,
when the pooled estimate is statistically significant, researchers should
report $\delta_{\min,\mathrm{std}}^*$: the smallest same-direction carryover
gap, expressed as a multiple of the pooled estimator, that would overturn
that significance. This does not restore point identification or guarantee
correct coverage of $\tau$, but states transparently how much carryover the estimated direction can withstand, while preserving the efficiency gain of the repeated measurements.

\subsection*{Concluding Remarks}
We close by clarifying the
scope of our analysis, and identifying directions for future research. Within-subject designs remain valuable, but their precision gains do not
eliminate the need to confront carryover. The appropriate response is neither
automatic pooling nor an automatic return to post-only designs. Researchers
should use the correct design-based estimator, justify plausible carryover
substantively, and report how much carryover their conclusions can withstand.

Two limitations define the scope of these conclusions. First, we study
unadjusted estimators and do not consider the inclusion of pretreatment
covariates, which may alter the efficiency comparison and the corresponding
sensitivity calculations.

Second, our formal results concern within-subject designs reducible
to the same two-sequence structure. Concerns about carryover across
repeated exposures, however, arise more broadly. Conjoint experiments are
one important example: when respondents complete multiple choice tasks,
the stability and no-carryover-effects assumption -- that earlier tasks do
not affect responses to later ones -- plays the same identifying role,
testable by comparing an attribute's Average Marginal Component Effect
(AMCE) across tasks \parencite[Sec.~5.3.1]{hainmueller2014causal}. The same
design-based and sensitivity arguments developed above may often apply
there, but our present results do not establish the exact sensitivity
mapping for conjoint estimands. Extending them to conjoint and other
repeated-treatment designs is an important direction for future research.

%% ---------------------------------------------------------------------------
%% Working note, carried over from notes/paper.lyx. Not paper prose -- kept as a
%% comment so it does not render. Address these, then delete this block.
%%
%% Teppei's comment: (1) (are you saying the t-test is underpowered?) when you
%% want to test the carry-over effect, to make the power large enough, the
%% sample size is good enough for the post-only experiment; (2) conjoint should
%% be fine as it's hard to believe there are carry-over effects - Yes, advise
%% people to make substantive arguments (don't defer substantive judgements to
%% statistical criterion); (3) can discuss conjoint in the lit review, but focus
%% on two period (Clifford example).
%% ---------------------------------------------------------------------------

%TC:ignore
\paragraph{Acknowledgments} We thank Diana Jordan, Jule Krüger, Yuehong
Cassandra Tai, and Amanda Weiss, and participants at the Visions in
Methodology pre-PolMeth workshop 2026, for helpful comments.

\paragraph{Use of Artificial Intelligence} Anthropic's Claude (Sonnet 5,
Opus 4.8, and Opus 5) and OpenAI's Codex (Terra 5.6 and Sol 5.6) assisted
the authors with language editing, \LaTeX{} formatting, drafting, coding,
debugging, and reproducibility checks between July 21 and August 25, 2026.
The authors supplied the original research design, data, and code, made
all substantive decisions, reviewed and verified all AI-assisted output,
and take full responsibility for the manuscript's scientific content.

\paragraph{Funding Statement} S.L. received funding for this work from Peking University.

\paragraph{Competing Interests} The authors declare none.

\paragraph{Data Availability Statement} Replication materials for this
article are available at this Dropbox
\href{https://www.dropbox.com/scl/fo/scmlm9g7x3kjr1mfupz9f/AHhYrlzm58uQxDbdFzkiBVE?rlkey=ithp220x70nh2si6gp486b0wv\&st=nr68iuqj\&dl=0}{Link}
for the duration of the review
process. Upon acceptance, the package will be deposited at Harvard Dataverse
and this statement updated with the resulting persistent DOI. This
replication package bundles two third-party datasets under their own open
licenses: \textcite{clifford2021data} (CC0 1.0) and \textcite{jordan2026data}
(CC BY 4.0).

\paragraph{Research with Human Subjects} This study analyzes previously
collected, de-identified data and publicly available replication materials. The
authors did not recruit participants, interact with participants, or collect
new human-subject data for this research. The original studies' ethical-review
and informed-consent procedures are described in their respective publications
and replication materials.

\paragraph{Author Contributions} S.L. and J.Z. contributed equally to this work. Authors are listed in alphabetical order.
%TC:endignore

\printbibliography

\appendix

%% Appendix equations number A.1, A.2, ... rather than continuing the body's
%% flat count. \numberwithin is unusable here: the class sets \thesection to the
%% literal string "Appendix N", which would render "(Appendix 1.3)". \Alph{section}
%% sidesteps \thesection; the starred \counterwithin adds the per-section reset
%% without overriding \theequation, and makes a future second appendix section
%% number B.1 rather than continuing A.21.
\counterwithin*{equation}{section}
\renewcommand{\theequation}{\Alph{section}.\arabic{equation}}

%% Appendix figures use a separate A-prefixed sequence so their location is
%% apparent from every caption and cross-reference (Figure A.1, A.2, ...).
\setcounter{figure}{0}
\renewcommand{\thefigure}{A.\arabic{figure}}

%% Appendix tables use the same A-prefixed sequence as figures, so their
%% location is apparent from every caption and cross-reference (Table A.1,
%% A.2, ...), distinguishing them from the main-text Tables 1-4.
\setcounter{table}{0}
\renewcommand{\thetable}{A.\arabic{table}}

\begin{refsection}
\input{appendix}

\clearpage
\printbibliography[title={References Cited Only in the Appendix}]
\end{refsection}

\end{document}

%% file: tables/litreview_sensitivity_macros.tex
\newcommand{\LitReviewAcrossPaperDeltaRatio}{0.352}
\newcommand{\LitReviewAcrossPaperDeltaMinStd}{1.168}
\newcommand{\LitReviewZetaCount}{30}
\newcommand{\LitReviewZetaMedian}{0.271}
\newcommand{\LitReviewZetaMin}{-0.005}
\newcommand{\LitReviewZetaMax}{0.930}
\newcommand{\LitReviewZetaPTen}{0.044}
\newcommand{\LitReviewZetaPNinety}{0.537}
\newcommand{\LitReviewZetaNonnegCount}{29}
\newcommand{\LitReviewZetaFaithfulCount}{27}
\newcommand{\LitReviewZetaFaithfulMedian}{0.221}
\newcommand{\LitReviewZetaFaithfulMin}{-0.005}
\newcommand{\LitReviewZetaFaithfulMax}{0.724}
\newcommand{\LitReviewZetaFaithfulNonnegCount}{26}
\newcommand{\LitReviewOmegaFaithfulMin}{0.69}
\newcommand{\LitReviewOmegaFaithfulMax}{1.21}

%% file: tables/spec_comparison.tex
\begin{tabular}{lccccccc}
\toprule
 & $\hat{\beta}_1$ & $\hat{\beta}_3$ & ATE & ATE/design & SE & SE/design \\
\midrule
\multicolumn{7}{l}{\emph{Panel A: original data} ($n_1 = 227$, $n_2 = 226$)} \\
Design-based & & &   -0.278 & exact &    0.033 & exact \\
(0) Paired $t$-test & & &   -0.278 & +0.02\% &    0.033 & +0.08\% \\
(1) Treatment-only OLS, clustered &   -0.278 & &   -0.278 & +0.02\% &    0.033 & +0.08\% \\
(2) Treatment-only RE, clustered &   -0.278 & &   -0.278 & +0.02\% &    0.033 & +0.08\% \\
(3) Period-interacted OLS, clustered &   -0.289 &    0.021 &   -0.278 & exact &    0.033 & exact \\
(4) Period-interacted RE, clustered &   -0.289 &    0.021 &   -0.278 & exact &    0.033 & exact \\
(5) Period-1 OLS &   -0.289 & &   -0.289 & -3.72\% &    0.070 & +114.79\% \\
\addlinespace
\multicolumn{7}{l}{\emph{Panel B: synthetic data} ($n_1 = 454$, $n_2 = 226$)} \\
Design-based & & &   -0.278 & exact &    0.028 & exact \\
(0) Paired $t$-test & & &   -0.269 & +3.27\% &    0.027 & -4.81\% \\
(1) Treatment-only OLS, clustered &   -0.269 & &   -0.269 & +3.27\% &    0.027 & -4.81\% \\
(2) Treatment-only RE, clustered &   -0.269 & &   -0.269 & +3.27\% &    0.027 & -4.81\% \\
(3) Period-interacted OLS, clustered &   -0.289 &    0.021 &   -0.278 & exact &    0.028 & exact \\
(4) Period-interacted RE, clustered &   -0.289 &    0.021 &   -0.278 & exact &    0.028 & exact \\
(5) Period-1 OLS &   -0.289 & &   -0.289 & -3.72\% &    0.059 & +108.46\% \\
\bottomrule
\end{tabular}
\par\smallskip
\begin{minipage}{\linewidth}
\footnotesize\raggedright
\textit{Notes:} 
ATE/design is the percentage difference between a specification's ATE and the design-based point estimate. SE/design is the percentage difference between its reported standard error and the design-based standard error. ``Exact'' denotes equality at the displayed precision. Point-estimate and standard-error equivalence is proved in \ref{app:bare-pooled} and \ref{app:saturated}; the finite-sample corrections applied to compare each row against the design-based standard error are given in \ref{app:se-corrections}.

\end{minipage}

%% file: tables/litreview_design_classes.tex
\begin{table}[!h]
\centering
\caption{Design classification of the 33 papers citing Clifford, Sheagley, and Piston (2021)}
\label{tab:litreview-design-classes}
\footnotesize
\begin{tabular}{@{}p{10.4cm} r@{}}
\toprule
Group & $N$ \\
\midrule
1. Exact or directly reducible two-sequence designs & 3 \\
2. Multiperiod or multicondition designs reduced via pairwise comparison & 5 \\
3. Within-subject designs with fixed Period 1 exposure & 14 \\
4. Other non-reducible or unresolved designs & 11 \\
\midrule
Total & 33 \\
\bottomrule
\end{tabular}
\end{table}

%% file: tables/litreview_sensitivity_summary.tex
\begin{table}[!h]
\centering
\caption{Sensitivity quantities for papers with significant pooled treatment-effect estimates}
\label{tab:litreview-sensitivity-summary}
\scriptsize
\begin{tabular}{@{}>{\raggedright\arraybackslash}p{3.0cm} r r r r r r r r r@{}}
\toprule
Paper & $K$ & $\bar n$ & $\bar t_{co}$ & $\bar\sigma_1^2$ & $\bar\sigma_2^2$ & $\bar\sigma_{12}$ & $\overline{\hat\delta/\hat\mu_{co}}$ & $\overline{|\hat\delta/\hat\mu_{co}|}$ & $\overline{\delta_{\min,\mathrm{std}}^*}$ \\
\midrule
\textcite{carnahan2022correcting} & 1 & 348.0 & 4.70 & 0.121 & 0.111 & 0.075 & 0.100 & 0.100 & 1.166 \\
\textcite{tappin2023estimating} & 6 & 516.3 & 4.74 & 0.433 & 0.446 & 0.060 & -0.593 & 0.686 & 1.106 \\
\textcite{ozer2022partisan} & 2 & 24.0 & -2.53 & 0.162 & 0.133 & 0.065 & -0.518 & 0.518 & 0.434 \\
\textcite{vantrappen2023biased} & 2 & 597.0 & -3.22 & 3.238 & 3.343 & 1.152 & 0.040 & 0.158 & 0.806 \\
\textcite{velez2023latino} & 2 & 700.0 & 5.68 & 4.527 & 4.345 & 1.889 & -0.334 & 0.334 & 1.256 \\
\textcite{halling2024frontline} & 2 & 522.0 & 5.38 & 35.508 & 34.376 & 4.776 & 0.240 & 0.240 & 1.270 \\
\textcite{clifford2021increasing} & 1 & 453.0 & 8.50 & 2.226 & 1.995 & 1.145 & -0.074 & 0.074 & 1.539 \\
\textcite{JORDAN_OLLERENSHAW_TREXLER_2026} & 9 & 2918.9 & 18.76 & 0.539 & 0.552 & 0.379 & -0.348 & 0.702 & 1.765 \\
\bottomrule
\end{tabular}
\vspace{0.4em}
\begin{minipage}{\textwidth}
\scriptsize \textit{Notes}: The table includes outcome-specific treatment-control comparisons with $|t_{co}|>z_{.975}$. $K$ is the number of included comparisons, and every reported quantity is an equal-weight within-paper mean. A negative $\overline{\hat\delta/\hat\mu_{co}}$ guarantees the post-only estimate is larger in magnitude than the pooled one (pooling attenuates); a positive value does not have a symmetric interpretation. The 25 comparisons represent 24 distinct treatment-control pairs because one paper contributes two outcomes for one pair.
\end{minipage}
\end{table}

%% file: appendix.tex
%% Appendix. Included from main.tex via \input, after \appendix.
%%
%% NOTE: the class redefines \thesection to "Appendix N", so \ref{app:proofs}
%% already renders as "Appendix 1" -- write "see \ref{app:proofs}", never
%% "see Appendix~\ref{app:proofs}" (that yields "Appendix Appendix 1").
%%
%% Results supported by notes/Outline.lyx and the regression-equivalence notes
%% are stated and proved below. Remaining TK markers concern simulations or the
%% Grizzle derivation.

\section{Proofs}\label{app:proofs}

\subsection{Assumptions}\label{app:assumptions}

\begin{assumption}[SUTVA]\label{as:sutva}
For every unit $i$, period $t$, and treatment sequence assignments that assign the same sequence to unit $i$,  the potential outcome for unit $i$
in period $t$ is the same. Thus
there is no interference across units, and the potential outcome can be
written as a function only of unit $i$'s own treatment
sequence, $Y_{it}(d_{1},d_{2})$. The observed outcome can be writte as $Y_{it}=Y_{it}(D_{i1},D_{i2})$.\end{assumption}

\begin{assumption}[Random assignment of treatment sequence]\label{as:random}
The treatment-sequence assignments are independent of the potential outcomes of all units. Exactly $n_{1}$ units are assigned to sequence $(1,0)$
and $n_{2}$ units are assigned to sequence $(0,1)$, where $n=n_{1}+n_{2}$.
As $n\to\infty$, $n_1\to rn$ and $n_2\to (1-r)n$, with $0<r<1$.
\end{assumption}

\begin{assumption}[Independence across units and regularity conditions]\label{as:regularity}
The vectors
\[
\mathcal{Y}_{i}\equiv
\left(Y_{i1}(1,0),Y_{i1}(0,1),Y_{i2}(1,0),Y_{i2}(0,1)\right)
\]
are independently and identically distributed across units. The variance of each component is strictly positive and finite.
\end{assumption}

\subsection{Identification results}\label{app:identification}

Define the average treatment effect using period-1 potential outcomes:
\[\tau
\equiv\mathbb{E}\left[Y_{i1}(1,0)-Y_{i1}(0,1)\right].
\]
Define the carryover gap as the difference between the average treatment effects in periods 1 and 2:
\[\delta
\equiv\mathbb{E}\left[Y_{i2}(0,1)-Y_{i2}(1,0)\right]-\tau.
\]
Define the population treated-minus-control contrasts in the observed outcome in the two periods as
\begin{align*}
\mu_{1}
&\equiv\mathbb{E}\left[Y_{i1}\mid D_{i1}=1\right]
 -\mathbb{E}\left[Y_{i1}\mid D_{i1}=0\right],\\
\mu_{2}
&\equiv\mathbb{E}\left[Y_{i2}\mid D_{i2}=1\right]
 -\mathbb{E}\left[Y_{i2}\mid D_{i2}=0\right].
\end{align*}

\begin{prop}[Period-1 difference-in-means identifies the ATE]\label{prop:period1}
Under Assumptions~\ref{as:sutva} and~\ref{as:random}, $\mu_{1}=\tau$.
\end{prop}

\begin{proof}
There are only two possible treatment sequences. For unit $i$ with treatment sequence $(D_{i1},D_{i2})=(1,0)$, $Y_{i1}=Y_{i1}(1,0)$. For unit $i$ with treatment sequence  $(D_{i1},D_{i2})=(0,1)$, $Y_{i1}=Y_{i1}(0,1)$. 
Therefore
\[\mu_{1}=\mathbb{E}\left[Y_{i1}(1,0)\mid(D_{i1},D_{i2})=(1,0)\right]-\mathbb{E}\left[Y_{i1}(0,1)\mid(D_{i1},D_{i2})=(0,1)\right].
\]
Because the treatment sequence is randomly assigned, we have
\[\mu_1=\mathbb{E}\left[Y_{i1}(1,0)\right]
  -\mathbb{E}\left[Y_{i1}(0,1)\right]=\tau.
\]
\end{proof}

\begin{prop}[Period-2 difference-in-means identifies the ATE plus the carryover gap]\label{prop:period2}
Under Assumptions~\ref{as:sutva} and~\ref{as:random},
$\mu_{2}=\tau+\delta$.
\end{prop}

\begin{proof}
There are only two possible treatment sequences. For unit $i$ with treatment sequence  $(D_{i1},D_{i2})=(0,1)$, $Y_{i2}=Y_{i2}(0,1)$. For unit $i$ with treatment sequence $(D_{i1},D_{i2})=(1,0)$, $Y_{i2}=Y_{i2}(1,0)$. 
Therefore
\[\mu_{2}=\mathbb{E}\left[Y_{i2}(0,1)\mid(D_{i1},D_{i2})=(0,1)\right]-\mathbb{E}\left[Y_{i2}(1,0)\mid(D_{i1},D_{i2})=(1,0)\right].
\]
Because the treatment sequence is randomly assigned, we have
\[\mu_2=\mathbb{E}\left[Y_{i2}(0,1)\right]
  -\mathbb{E}\left[Y_{i2}(1,0)\right]=\tau+\delta.
\]
\end{proof}

\begin{prop}[The pooled estimator identifies the ATE plus half the carryover gap]\label{prop:pooled}
Under Assumptions~\ref{as:sutva} and~\ref{as:random},
\[
\mu_{co}=\frac{1}{2}(\mu_{1}+\mu_{2})=\tau+\frac{\delta}{2}.
\]
Consequently, $\mu_{co}=\tau$ if and only if $\delta=0$, equivalently if and
only if $\mu_{1}=\mu_{2}$.
\end{prop}

\begin{proof}
Propositions~\ref{prop:period1} and~\ref{prop:period2} give
\[
\mu_{co}=\frac{1}{2}\{\tau+(\tau+\delta)\}
=\tau+\frac{\delta}{2}.
\]
The two equivalences follow because $\mu_{2}-\mu_{1}=\delta$.
\end{proof}

The saturated model in Equation~\eqref{eq:saturated-model} is one useful
illustration of these potential-outcome estimands. Under that parameterization,
\[
\tau=\mathbb{E}[\tau_{i}],
\qquad
\delta=\mathbb{E}\left[\delta_{i,0}-\delta_{i,1}\right].
\]
The identification arguments above do not otherwise rely on that
parameterization.

\subsection{Sampling distributions and variances}\label{app:sampling}

Define $\bar{Y}_{t,(1,0)}$ and
$\bar{Y}_{t,(0,1)}$ as the sample means of the observed period-$t$ outcome among units with treatment sequences $(1,0)$ and $(0,1)$. The period-specific estimators are
\[
\hat{\mu}_{1}=\bar{Y}_{1,(1,0)}-\bar{Y}_{1,(0,1)},
\qquad
\hat{\mu}_{2}=\bar{Y}_{2,(0,1)}-\bar{Y}_{2,(1,0)}.
\]
Define the variance scales
\begin{align*}
\sigma_{1}^{2}
&\equiv\frac{n}{n_{1}}\mathbb{V}\left(Y_{i1}(1,0)\right)
 +\frac{n}{n_{2}}\mathbb{V}\left(Y_{i1}(0,1)\right),\\
\sigma_{2}^{2}
&\equiv\frac{n}{n_{2}}\mathbb{V}\left(Y_{i2}(0,1)\right)
 +\frac{n}{n_{1}}\mathbb{V}\left(Y_{i2}(1,0)\right).
\end{align*}

\paragraph*{Marginal asymptotic distributions}
\begin{prop}[Marginal asymptotic distributions of the period-specific estimators]\label{prop:marginal-asymptotic}
Under Assumptions~\ref{as:sutva}--\ref{as:regularity},
\begin{align*}
\frac{\sqrt{n}(\hat{\mu}_{1}-\mu_{1})}{\sigma_{1}}
&\overset{d}{\longrightarrow}\mathcal{N}(0,1),\\
\frac{\sqrt{n}(\hat{\mu}_{2}-\mu_{2})}{\sigma_{2}}
&\overset{d}{\longrightarrow}\mathcal{N}(0,1).
\end{align*}
\end{prop}

\begin{proof}
The two sequence groups are disjoint collections of i.i.d. units, so their sample means are
independent. Substituting the expressions for $\hat\mu_1$ and $\hat\mu_2$ gives
\begin{align*}
\mathbb{V}\left(\hat{\mu}_{1}\right)
&=\mathbb{V}\left(
  \bar{Y}_{1,(1,0)}-\bar{Y}_{1,(0,1)}\right)\\
&=\mathbb{V}\left(\bar{Y}_{1,(1,0)}\right)
 +\mathbb{V}\left(\bar{Y}_{1,(0,1)}\right)\\
&=\frac{1}{n_{1}}\mathbb{V}\left(Y_{i1}(1,0)\right)
 +\frac{1}{n_{2}}\mathbb{V}\left(Y_{i1}(0,1)\right)
 =\frac{\sigma_{1}^{2}}{n},\\
\mathbb{V}\left(\hat{\mu}_{2}\right)
&=\mathbb{V}\left(
  \bar{Y}_{2,(0,1)}-\bar{Y}_{2,(1,0)}\right)\\
&=\mathbb{V}\left(\bar{Y}_{2,(0,1)}\right)
 +\mathbb{V}\left(\bar{Y}_{2,(1,0)}\right)\\
&=\frac{1}{n_{2}}\mathbb{V}\left(Y_{i2}(0,1)\right)
 +\frac{1}{n_{1}}\mathbb{V}\left(Y_{i2}(1,0)\right)
 =\frac{\sigma_{2}^{2}}{n}.
\end{align*}
The second equality in each calculation uses independence across the two
sequence groups; the third uses the i.i.d. assumption and the group sizes.
The corresponding expectations for $\hat\mu_1$ and $\hat\mu_2$ are $\mu_{1}$ and $\mu_{2}$. Applying the central limit theorem to the two independent group means in each period and taking their difference yields the stated distributions.
\end{proof}

\begin{prop}[Asymptotic distribution of the carryover-gap estimator]\label{prop:delta-asymptotic}
Let $\hat{\delta}\equiv\hat{\mu}_{2}-\hat{\mu}_{1}$. Under
Assumptions~\ref{as:sutva}--\ref{as:regularity},
\[
\frac{\sqrt{n}(\hat{\delta}-\delta)}{\sigma_{\delta}}
\overset{d}{\longrightarrow}\mathcal{N}(0,1),
\]
where
\[
\sigma_{\delta}^{2}
\equiv\frac{n}{n_{1}}\mathbb{V}\left(Y_{i1}(1,0)+Y_{i2}(1,0)\right)
 +\frac{n}{n_{2}}\mathbb{V}\left(Y_{i1}(0,1)+Y_{i2}(0,1)\right).
\]
\end{prop}

\begin{proof}
Substituting the expressions for $\hat\mu_1$ and $\hat\mu_2$ gives
\begin{align*}
\hat{\delta}
&=\left\{\bar{Y}_{1,(0,1)}+\bar{Y}_{2,(0,1)}\right\}
 -\left\{\bar{Y}_{1,(1,0)}+\bar{Y}_{2,(1,0)}\right\}.
\end{align*}
The two braces are means from independent sequence groups, of sizes $n_{2}$
and $n_{1}$ respectively. Therefore,
\begin{align*}
\mathbb{V}\left(\hat{\delta}\right)
&=\mathbb{V}\left(
 \left\{\bar{Y}_{1,(0,1)}+\bar{Y}_{2,(0,1)}\right\}
 -\left\{\bar{Y}_{1,(1,0)}+\bar{Y}_{2,(1,0)}\right\}
 \right)\\
&=\mathbb{V}\left(
 \bar{Y}_{1,(0,1)}+\bar{Y}_{2,(0,1)}\right)+\mathbb{V}\left(
 \bar{Y}_{1,(1,0)}+\bar{Y}_{2,(1,0)}\right)\\
&=\frac{1}{n_{2}}
  \mathbb{V}\left(Y_{i1}(0,1)+Y_{i2}(0,1)\right)+\frac{1}{n_{1}}
  \mathbb{V}\left(Y_{i1}(1,0)+Y_{i2}(1,0)\right)
 =\frac{\sigma_{\delta}^{2}}{n}.
\end{align*}
The second equality uses independence between the two sequence-group means,
and the third uses the i.i.d. assumption within each group.  Since $\hat\mu_1$ and $\hat\mu_2$ are independent and asymptotically normal, $\hat{\delta}$ is asymptotically normal as well, yielding the stated distribution.
\end{proof}

\begin{prop}[Asymptotic distribution of the pooled estimator]\label{prop:pooled-asymptotic} Let $\hat{\mu}_{co}\equiv\frac{1}{2}(\hat{\mu}_{1}+\hat{\mu}_{2})$. Under
Assumptions~\ref{as:sutva}--\ref{as:regularity}, 
\begin{equation}
\frac{2\sqrt{n}(\hat{\mu}_{co}-\mu_{co})}{\sigma_{co}}
\overset{d}{\longrightarrow}\mathcal{N}(0,1),
\label{eq:pooled-asymptotic}
\end{equation}
where
\[
\sigma_{co}^{2}
\equiv\frac{n}{n_{1}}\mathbb{V}\left(Y_{i1}(1,0)-Y_{i2}(1,0)\right)
 +\frac{n}{n_{2}}\mathbb{V}\left(Y_{i1}(0,1)-Y_{i2}(0,1)\right).
\]\end{prop}

\begin{proof}
Substituting the expressions for $\hat\mu_1$ and $\hat\mu_2$ gives
\[
\hat{\mu}_{co}=\frac{1}{2}
\left\{\bar{Y}_{1,(1,0)}-\bar{Y}_{2,(1,0)}\right\}
+\frac{1}{2}\left\{\bar{Y}_{2,(0,1)}-\bar{Y}_{1,(0,1)}\right\}.
\]
The two braces are means from independent sequence groups, of sizes $n_{2}$
and $n_{1}$ respectively. Therefore,
\begin{align*}
\mathbb{V}\left(\hat{\mu}_{co}\right)
&=\frac{1}{4}\mathbb{V}\left(
 \left\{\bar{Y}_{1,(1,0)}-\bar{Y}_{2,(1,0)}\right\}
 +\left\{\bar{Y}_{1,(0,1)}-\bar{Y}_{2,(0,1)}\right\}
 \right)\\
&=\frac{1}{4}\mathbb{V}\left(
 \bar{Y}_{1,(1,0)}-\bar{Y}_{2,(1,0)}\right)+\frac{1}{4}\mathbb{V}\left(
 \bar{Y}_{1,(0,1)}-\bar{Y}_{2,(0,1)}\right)\\
&=\frac{1}{4n_{1}}
  \mathbb{V}\left(Y_{i1}(1,0)-Y_{i2}(1,0)\right)+\frac{1}{4n_{2}}
  \mathbb{V}\left(Y_{i1}(0,1)-Y_{i2}(0,1)\right)
 =\frac{\sigma_{co}^{2}}{4n}.
\end{align*}
The second equality uses independence between the two sequence-group means,
and the third uses the i.i.d. assumption within each group.  Since $\hat\mu_1$ and $\hat\mu_2$ are independent and asymptotically normal, $\hat{\mu}_{co}$ is asymptotically normal as well, yielding the stated distribution.
\end{proof}

\paragraph*{Relationship between variances and covariances}
Applying
$\mathbb{V}(A+B)=\mathbb{V}(A)+\mathbb{V}(B)+2\mathrm{Cov}(A,B)$ gives
\begin{equation}
\begin{aligned}
\sigma_{\delta}^{2}
&=\frac{n}{n_{1}}\left\{
  \mathbb{V}\left(Y_{i1}(1,0)\right)
 +\mathbb{V}\left(Y_{i2}(1,0)\right)
 +2\mathrm{Cov}\left(Y_{i1}(1,0),Y_{i2}(1,0)\right)\right\}\\
&\quad+\frac{n}{n_{2}}\left\{
  \mathbb{V}\left(Y_{i1}(0,1)\right)
 +\mathbb{V}\left(Y_{i2}(0,1)\right)
 +2\mathrm{Cov}\left(Y_{i1}(0,1),Y_{i2}(0,1)\right)\right\}\\
&=\sigma_{1}^{2}+\sigma_{2}^{2}+2\sigma_{12}.
\end{aligned}
\label{eq:sigma-delta-identity}
\end{equation}
Likewise, applying
$\mathbb{V}(A-B)=\mathbb{V}(A)+\mathbb{V}(B)-2\mathrm{Cov}(A,B)$ gives
\begin{equation}
\begin{aligned}
\sigma_{co}^{2}
&=\frac{n}{n_{1}}\left\{
  \mathbb{V}\left(Y_{i1}(1,0)\right)
 +\mathbb{V}\left(Y_{i2}(1,0)\right)
 -2\mathrm{Cov}\left(Y_{i1}(1,0),Y_{i2}(1,0)\right)\right\}\\
&\quad+\frac{n}{n_{2}}\left\{
  \mathbb{V}\left(Y_{i1}(0,1)\right)
 +\mathbb{V}\left(Y_{i2}(0,1)\right)
 -2\mathrm{Cov}\left(Y_{i1}(0,1),Y_{i2}(0,1)\right)\right\}\\
&=\sigma_{1}^{2}+\sigma_{2}^{2}-2\sigma_{12}.
\end{aligned}
\label{eq:variance-identities}
\end{equation}

\begin{prop}[A sufficient condition for pooling to be more efficient]\label{prop:pooling-sufficient}
$\sigma_2^2<\sigma_1^2$ is sufficient for $\sigma_{co}^2<4\sigma_1^2$
regardless of $\sigma_{12}$; the weak inequality $\sigma_2^2\le\sigma_1^2$ is
sufficient except when $\sigma_2=\sigma_1$ and $\sigma_{12}=-\sigma_1\sigma_2$.
\end{prop}

\begin{proof}
First bound $\sigma_{12}$. Within each sequence, Cauchy--Schwarz applied to
the two potential outcomes gives
\[
\left|\mathrm{Cov}\left(Y_{i1}(1,0),Y_{i2}(1,0)\right)\right|
\le
\sqrt{\mathbb V\left(Y_{i1}(1,0)\right)\mathbb V\left(Y_{i2}(1,0)\right)},
\]
and likewise for sequence $(0,1)$. Writing
$a_1=\sqrt{(n/n_1)\mathbb V(Y_{i1}(1,0))}$,
$a_2=\sqrt{(n/n_2)\mathbb V(Y_{i1}(0,1))}$,
$b_1=\sqrt{(n/n_1)\mathbb V(Y_{i2}(1,0))}$,
$b_2=\sqrt{(n/n_2)\mathbb V(Y_{i2}(0,1))}$, the two within-sequence bounds
give 
\[
\begin{aligned}
|\sigma_{12}|&\le\frac{n}{n_1}\left|\mathrm{Cov}\left(Y_{i1}(1,0),Y_{i2}(1,0)\right)\right|+\frac{n}{n_2}\left|\mathrm{Cov}\left(Y_{i1}(0,1),Y_{i2}(0,1)\right)\right|\\
&\le a_1b_1+a_2b_2.
\end{aligned}
\]
The Cauchy--Schwarz inequality gives $a_1b_1+a_2b_2\le\sqrt{a_1^2+a_2^2}\sqrt{b_1^2+b_2^2}
=\sigma_1\sigma_2$. Hence $|\sigma_{12}|\le\sigma_1\sigma_2$.

By Equation~\ref{eq:variance-identities}, $\sigma_{co}^2<4\sigma_1^2$ is
equivalent to $\sigma_2^2-2\sigma_{12}<3\sigma_1^2$. Since
$\sigma_{12}\ge-\sigma_1\sigma_2$, the left side is at most
$\sigma_2^2+2\sigma_1\sigma_2$. Writing $x\equiv\sigma_2/\sigma_1$,
$\sigma_2^2+2\sigma_1\sigma_2<3\sigma_1^2$ reduces to $x^2+2x-3<0$, i.e.
$(x-1)(x+3)<0$, which holds for $x\in(0,1)$.
Hence $\sigma_2^2<\sigma_1^2$ guarantees $\sigma_{co}^2<4\sigma_1^2$ for every
admissible $\sigma_{12}$. When $\sigma_2^2=\sigma_1^2$, $x=1$, the inequality becomes equality only at $\sigma_{12}=-\sigma_1\sigma_2$, in which case $\sigma_{co}^2=4\sigma_1^2$.
\end{proof}

\paragraph*{Design-based variance estimators.}
Let $S_{(d_1,d_2)}^2(X)$ denote the sample variance of $X$ among respondents
assigned to sequence $(d_1,d_2)$), computed with mean squared deviations divided by 
$n$ (not $n-1$). The sample analogs of $\sigma_1^2$, $\sigma_2^2$, $\sigma_{co}^2$ and $\sigma_{\delta}^2$ are
\[
\begin{aligned}
\hat\sigma_1^2
&=
\frac{n}{n_1}S_{(1,0)}^2(Y_{i1})
+
\frac{n}{n_2}S_{(0,1)}^2(Y_{i1}),\\
\hat\sigma_2^2
&=
\frac{n}{n_2}S_{(0,1)}^2(Y_{i2})
+
\frac{n}{n_1}S_{(1,0)}^2(Y_{i2}),\\
\hat\sigma_{co}^2
&=
\frac{n}{n_1}S_{(1,0)}^2(\Delta_i)
+
\frac{n}{n_2}S_{(0,1)}^2(\Delta_i),\\
\hat\sigma_{\delta}^2
&=
\frac{n}{n_1}S_{(1,0)}^2(Y_{i1}(1,0)+Y_{i2}(1,0))
+
\frac{n}{n_2}S_{(0,1)}^2(Y_{i1}(0,1)+Y_{i2}(0,1)).
\end{aligned}
\]
The resulting design-based standard errors are
\[
\widehat{\operatorname{se}}(\hat\mu_1)
=
\frac{\hat\sigma_1}{\sqrt n},
\qquad
\widehat{\operatorname{se}}(\hat\mu_2)
=
\frac{\hat\sigma_2}{\sqrt n},
\qquad
\widehat{\operatorname{se}}(\hat\mu_{co})
=
\frac{\hat\sigma_{co}}{2\sqrt n},\qquad \widehat{\operatorname{se}}(\hat\mu_{\delta})
=
\frac{\hat\sigma_{\delta}}{2\sqrt n}.
\]

\paragraph*{Joint asymptotic distributions}

The results so far are related to marginal distributions. The analysis of the two-step procedure needs the joint asymptotic distributions of $W_1=(\hat{\delta},\hat{\mu}_{1})^{\top}$ and $W_2=(\hat\delta,\hat{\mu}_{co})^{\top}$.

\begin{prop}[Joint asymptotic distributions]\label{prop:joint-asymptotic}
Let $W_1=(\hat{\delta},\hat{\mu}_{1})^{\top}$, $W_2=(\hat\delta,\hat{\mu}_{co})^{\top}$, $m_1=(\delta,\mu_{1})^{\top}$ and $m_2=(\delta,\mu_{co})^{\top}$. Under
Assumptions~\ref{as:sutva}--\ref{as:regularity},
\[
\sqrt{n}\left(W_1-m_1\right)\overset{d}{\longrightarrow}\mathcal{N}\left(0,\Sigma_1\right),
\qquad
\Sigma_1=\begin{pmatrix}
\sigma_{\delta}^{2} & -\left(\sigma_{1}^{2}+\sigma_{12}\right)\\[3pt]
-\left(\sigma_{1}^{2}+\sigma_{12}\right) & \sigma_{1}^{2}
\end{pmatrix};
\]
and
\[
\sqrt{n}\left(W_2-m_2\right)\overset{d}{\longrightarrow}\mathcal{N}\left(0,\Sigma_2\right),
\qquad
\Sigma_2=\begin{pmatrix}
\sigma_{\delta}^{2} & \tfrac{1}{2}\left(\sigma_{2}^{2}-\sigma_{1}^{2}\right)\\[3pt]
\tfrac{1}{2}\left(\sigma_{2}^{2}-\sigma_{1}^{2}\right) & \tfrac{1}{4}\sigma_{co}^{2}
\end{pmatrix}.
\]
\end{prop}

\begin{proof}
Each component of $W_1$ is a fixed linear combination of the four within-sequence sample means $\bar{Y}_{t,(d_{1},d_{2})}$. Hence any scalar linear combination $a_1^{\top}W_1$ is also a linear combination of the sample means. Since each sample mean is an average of i.i.d. terms, it follows from the central limit theorem that every scalar linear combination of $W_1$ is asymptotically normal, which implies that $W_1$ is asymptotically bivariate normal. The same argument applies to $W_2$.

The variance terms follow directly from the marginal distributions. For the covariance terms, we first compute $\mathrm{Cov}\left(\hat{\mu}_{1},\hat{\mu}_{2}\right)$. Substituting the expressions for $\hat\mu_1$ and $\hat\mu_2$ and using the independence between the two sequence groups yields
\[
\mathrm{Cov}\left(\hat{\mu}_{1},\hat{\mu}_{2}\right)
=-\frac{1}{n_{1}}\mathrm{Cov}\left(Y_{i1}(1,0),Y_{i2}(1,0)\right)
 -\frac{1}{n_{2}}\mathrm{Cov}\left(Y_{i1}(0,1),Y_{i2}(0,1)\right)
=-\frac{\sigma_{12}}{n}.
\]
Since $\hat{\delta}=\hat{\mu}_{2}-\hat{\mu}_{1}$ and
$\hat{\mu}_{co}=\tfrac{1}{2}(\hat{\mu}_{1}+\hat{\mu}_{2})$, it follows that
\begin{align*}
\mathrm{Cov}\left(\hat{\delta},\hat{\mu}_{1}\right)&=\mathrm{Cov}\left(\hat\mu_2,\hat\mu_1\right)-\mathbb V(\hat\mu_1)=-\frac{1}{n}\left(\sigma_1^2+\sigma_{12}\right),\\
\mathrm{Cov}\left(\hat{\delta},\hat{\mu}_{co}\right)&=\frac{1}{2}\left(\mathbb V(\hat\mu_2)-\mathbb V(\hat\mu_1)\right)=\frac{1}{2n}\left(\sigma_2^2-\sigma_1^2\right).
\end{align*}
Combining the above arguments yields the stated asymptotic distribution.
\end{proof}

\subsection{Treatment-only  regressions}\label{app:bare-pooled}

Consider the treatment-only regression,
\begin{equation}
Y_{it}=\beta_{0}+\beta_{1}D_{it}+e_{it}.
\label{eq:bare-pooled}
\end{equation}
Every unit contributes one treated and one control observation. 

\paragraph*{Treatment-only OLS}
First, consider the OLS model that leaves the within-unit dependence unrestricted when inference is clustered by unit. 

There are $n$ observations for treatment and $n$ observations under control, so the estimator for $\beta_1$ equals the difference in mean observed outcomes between treatment and control observations. The mean observed outcomes under treatment and control are 
\begin{align*}
m_1&=\frac{1}{n}\left(\sum_{i:\ (D_{i1},D_{i2})=(1,0)}Y_{i1}+\sum_{i:\ (D_{i1},D_{i2})=(0,1)}Y_{i2}\right),\\
m_0&=\frac{1}{n}\left(\sum_{i:\ (D_{i1},D_{i2})=(1,0)}Y_{i2}+\sum_{i:\ (D_{i1},D_{i2})=(0,1)}Y_{i1}\right).
\end{align*}
Therefore $\hat{\beta}_1=m_1-m_0$. 

Using the sequence-specific treated-minus-control differences defined in the main text gives
\begin{align*}
\hat{\beta}_{1}
&=\frac{1}{n}\sum_{i=1}^n\Delta_i\\
&=\frac{1}{n}\left\{
 \sum_{i:\ (D_{i1},D_{i2})=(1,0)}\Delta_i(1,0)
 +\sum_{i:\ (D_{i1},D_{i2})=(0,1)}\Delta_i(0,1)
 \right\}\\
&=\frac{n_{1}}{n}\bar{\Delta}_{(1,0)}
 +\frac{n_{2}}{n}\bar{\Delta}_{(0,1)}.
\end{align*}
The pooled estimator can be written as
\[
\hat{\mu}_{co}=\frac{1}{2}\left\{
\bar{\Delta}_{(1,0)}+\bar{\Delta}_{(0,1)}\right\},
\]
The difference between these two estimators is
\[
\hat{\beta}_{1}-\hat{\mu}_{co}
=\frac{2n_1-n}{2n}\bar{\Delta}_{(1,0)}+\frac{2n_2-n}{2n}\bar{\Delta}_{(0,1)}\\
=\frac{n_{1}-n_{2}}{2n}\left\{
\bar{\Delta}_{(1,0)}-\bar{\Delta}_{(0,1)}\right\}.
\]
The second equality can be obtained by plugging in $n=n_1+n_2$.
Thus equal allocation ($n_1=n_2$) guarantees $\hat\beta_{1}=\hat\mu_{co}$. With unequal allocation ($n_1\neq n_2$), the two may coincide in a particular sample, but they are not the same estimator.

Applying the saturated potential-outcome representation in
Equation~\ref{eq:saturated-model} gives
\begin{align*}
\mathbb{E}\left[\Delta_i(1,0)\right]
&=\tau+\mathbb{E}\left[\pi_{i1}-\pi_{i2}\right]
  -\mathbb{E}\left[\delta_{i,1}\right],\\
\mathbb{E}\left[\Delta_i(0,1)\right]
&=\tau-\mathbb{E}\left[\pi_{i1}-\pi_{i2}\right]
  +\mathbb{E}\left[\delta_{i,0}\right].
\end{align*}
Therefore
\begin{align}
\mathbb{E}\left[\hat\beta_{1}\right]
&=\tau+\frac{n_{1}-n_{2}}{n}
  \mathbb{E}\left[\pi_{i1}-\pi_{i2}\right]
  +\frac{n_{2}}{n}\mathbb{E}\left[\delta_{i,0}\right]
  -\frac{n_{1}}{n}\mathbb{E}\left[\delta_{i,1}\right].
\label{eq:bare-pooled-expectation}
\end{align}
With unequal allocation ($n_1\neq n_2$), $\mathbb{E}[\hat\beta_1]=\tau$ only when
\begin{equation}
(n_{1}-n_{2})\mathbb{E}\left[\pi_{i1}-\pi_{i2}\right]
+n_{2}\mathbb{E}\left[\delta_{i,0}\right]
-n_{1}\mathbb{E}\left[\delta_{i,1}\right]=0.
\label{eq:bare-pooled-restriction}
\end{equation}
Because $\mathbb{E}[\hat\beta_1]$ in Equation~\eqref{eq:bare-pooled-expectation} contains $\mathbb{E}[\pi_{i1}-\pi_{i2}]$, it is not, in general, any fixed weighted average of
$\mu_{1}$ and $\mu_{2}$.

Next consider the unadjusted cluster-robust (CR0) variance estimator. 
Stacking the independent variables (the constant 1 and the treatment indicator $D_{it}$) for the $2n$ observations unit by unit (period 1 then period 2), with units in sequence (1,0) first and those in 
(0,1) second, yields the design matrix
$$
X=\left(\begin{array}{c}
\begin{array}{cc}
\\
1 & 1\\
1 & 0\\\vdots & \vdots\\
1 & 1\\
1 & 0\\
\end{array}\\
\hline
\begin{array}{cc}
1 & 0\\
1 & 1\\
\vdots & \vdots\\
1 & 0\\
1 & 1\\
\end{array}
\end{array}\right).
$$
The upper block contains $n_1$ units and $2n_1$ rows, and the lower block contains $n_2$ units and $2n_2$ rows. We can obtain
$$
X^{\prime}X=\left(\begin{array}{cc}
2n & n\\
n & n
\end{array}\right),
\qquad
\left(X^{\prime}X\right)^{-1}=\frac{1}{n}\left(\begin{array}{cc}
1 & -1\\
-1 & 2
\end{array}\right).
$$
Note that $(X'X)^{-1}$ does not depend on how $n$ splits into $n_{1}$ and
$n_{2}$, because there is one treated and one control row per unit
whichever sequence the unit is in. 

Write $\hat\beta=(\hat\beta_0,\hat\beta_1)^{\top}$.
The CR0 variance estimator for $\hat\beta$ is
$$
\widehat{\mathbb{V}}_{\mathrm{CR0}}(\hat\beta)
=\left(X^{\prime}X\right)^{-1}\Bigl(\textstyle\sum_{i=1}^{n}X_{i}^{\prime}\hat{e}_{i}\hat{e}_{i}^{\prime}X_{i}\Bigr)\left(X^{\prime}X\right)^{-1}=\left(X^{\prime}X\right)^{-1}\Bigl(\textstyle\sum_{i=1}^{n}s_is_i'\Bigr)\left(X^{\prime}X\right)^{-1},
$$
where $X_{i}$ is the $2\times2$ block of $X$ for unit $i$, $\hat e_{i}$ is the
residual vector for unit $i$, and $s_{i}\equiv X_{i}^{\prime}\hat e_{i}$. For units with
$(D_{i1},D_{i2})=(1,0)$,
\[
s_i=\begin{pmatrix}1&1\\1&0\end{pmatrix}
\begin{pmatrix}Y_{i1}-m_{1}\\ Y_{i2}-m_{0}\end{pmatrix}
=\begin{pmatrix}(Y_{i1}-m_{1})+(Y_{i2}-m_{0})\\ Y_{i1}-m_{1}\end{pmatrix}.
\]
For units with
$(D_{i1},D_{i2})=(0,1)$,
\[
s_i=\begin{pmatrix}1&1\\0&1\end{pmatrix}
\begin{pmatrix}Y_{i1}-m_{0}\\ Y_{i2}-m_{1}\end{pmatrix}
=\begin{pmatrix}(Y_{i1}-m_{0})+(Y_{i2}-m_{1})\\ Y_{i2}-m_{1}\end{pmatrix}.
\]

We need only the $(2,2)$ entry of the variance estimator matrix. The second row
of $(X'X)^{-1}$ is $v^{\prime}=\frac{1}{n}(-1,\,2)$, so
$$
\widehat{\mathbb{V}}_{\mathrm{CR0}}(\hat\beta_1)
=v^{\prime}\Bigl(\textstyle\sum_{i=1}^n s_{i}s_{i}^{\prime}\Bigr)v=\sum_{i=1}^n (v's_i)^2
=\frac{1}{n^{2}}\sum_{i=1}^{n}\bigl(s_{i,1}-2s_{i,2}\bigr)^{2},
$$
where $s_{i,1}$ and $s_{i,2}$ are the first and second elements of $s_i$.

Evaluating $w_{i}\equiv s_{i,1}-2s_{i,2}$ in each group: for unit $i$ in treatment sequence (1,0), $$w_{i}=(Y_{i2}-m_{0})-(Y_{i1}-m_{1})=(m_{1}-m_{0})-(Y_{i1}-Y_{i2})=\hat\beta_{1}-\Delta_{i};$$ for unit $i$ in treatment sequence (0,1), $$w_{i}=(Y_{i1}-m_{0})-(Y_{i2}-m_{1})=(m_{1}-m_{0})-(Y_{i2}-Y_{i1})=\hat\beta_{1}-\Delta_{i}.$$ So $w_{i}=\hat\beta_{1}-\Delta_{i}$ in both groups. Moreover write $\bar{\Delta}\equiv\frac{1}{n}\sum_{i=1}^n\Delta_i$, we have $\hat\beta_{1}=\bar{\Delta}$. Therefore
$$
\widehat{\mathbb{V}}_{\mathrm{CR0}}(\hat\beta_1)
=\frac{1}{n^{2}}\sum_{i=1}^{n}\bigl(\Delta_{i}-\bar\Delta\bigr)^{2}.
$$

Let
\[
S_{1}^{2}\equiv S^{2}_{(1,0)}\left(\Delta_i(1,0)\right),
\qquad
S_{2}^{2}\equiv S^{2}_{(0,1)}\left(\Delta_i(0,1)\right).
\]
Decomposing $\widehat{\mathbb{V}}_{\mathrm{CR0}}(\hat\beta_1)$ into within- and between-group parts:
\begin{equation}
\widehat{\mathbb{V}}_{\mathrm{CR0}}\left(\hat\beta_{1}\right)
=\frac{n_{1}S_{1}^{2}+n_{2}S_{2}^{2}}{n^{2}}
 +\frac{n_{1}n_{2}}{n^{3}}
  \left\{\bar\Delta_{(1,0)}-\bar\Delta_{(0,1)}\right\}^{2}.
\label{eq:bare-pooled-cr0}
\end{equation}
The design-based variance of $\hat\mu_{co}$ can instead be written as
\begin{equation}
\widehat{\mathbb{V}}\left(\hat\mu_{co}\right)
=\frac{\hat\sigma_{co}^{2}}{4n}=\frac{1}{4}\left(\frac{S_{1}^{2}}{n_{1}}
 +\frac{S_{2}^{2}}{n_{2}}\right).
\label{eq:design-pooled-sample-variance}
\end{equation}
When $n_{1}=n_{2}=n/2$, Equations~\eqref{eq:bare-pooled-cr0}
and~\eqref{eq:design-pooled-sample-variance} imply
\begin{equation}
\widehat{\mathbb{V}}_{\mathrm{CR0}}\left(\hat\beta_{1}\right)
=\frac{\hat\sigma_{co}^{2}}{4n}
 +\frac{1}{4n}
  \left\{\bar\Delta_{(1,0)}-\bar\Delta_{(0,1)}\right\}^{2}.
\label{eq:bare-pooled-wedge}
\end{equation}
Therefore under equal allocation the treatment-only regression's clustered variance is therefore weakly conservative
relative to the design-based variance, and is exact only when the two sequence-specific treated-minus-control mean
differences ($\bar\Delta_{(1,0)}$ and $\bar\Delta_{(0,1)}$) happen to coincide. 

With unequal allocation, subtracting
Equation~\eqref{eq:design-pooled-sample-variance} from
Equation~\eqref{eq:bare-pooled-cr0} gives
\begin{align*}
&\left(\frac{n_{1}}{n^{2}}-\frac{1}{4n_{1}}\right)S_{1}^{2}
 +\left(\frac{n_{2}}{n^{2}}-\frac{1}{4n_{2}}\right)S_{2}^{2}\\
&\qquad+\frac{n_{1}n_{2}}{n^{3}}
 \left\{\bar\Delta_{(1,0)}-\bar\Delta_{(0,1)}\right\}^{2}.
\end{align*}
The coefficients in the first two terms can have opposite signs, so the total discrepancy
has no uniform sign: with unequal allocation the treatment-only regression's clustered
variance can be conservative or anti-conservative relative to the design-based variance.

\paragraph*{Treatment-only random effects model} The regression specification is the same as Equation~\eqref{eq:bare-pooled}, but the error $e_{it}$ can be decomposed into a random intercept and an idiosyncratic component:
$$
e_{it}=u_{i}+\varepsilon_{it},$$
where $u_{i}\sim(0,\sigma_{u}^{2})$ and $\varepsilon_{it}\sim(0,\sigma_{\varepsilon}^{2})$ are mutually independent  with $\sigma_u^2>0$ and $\sigma_{\varepsilon}^2>0$. The covariance matrix of $e_i=(e_{i1},e_{i2})'$ for each unit is
$$
\Sigma_e=\mathrm{Var}(e)
=\left(\begin{array}{cc}
\sigma_u^2+\sigma_{\varepsilon}^2 & \sigma_u^2\\
\sigma_u^2 & \sigma_u^2+\sigma_{\varepsilon}^2
\end{array}\right).
$$
The covariance matrix of error terms for all units is 
$$\Omega=I_{n}\otimes\Sigma_e,$$
where $I_n$ is the $n\times n$ identity matrix, and $\otimes$ is the Kronecker product.

The fact that the treatment-only random effects model produces the same coefficient estimates and standard errors as the treatment-only OLS model can be proved using a method similar to that used for the period-interacted regression in \ref{app:saturated}.

\subsection{The period-interacted regression}\label{app:saturated}

Consider the period-interacted regression,
\begin{equation}
Y_{it}=\beta_{0}+\beta_{1}D_{it}+\beta_{2}I(t=2)
 +\beta_{3}D_{it}I(t=2)+e_{it}.
\label{eq:saturated-regression-appendix}
\end{equation}
Write
\[
\beta=(\beta_0,\beta_1,\beta_2,\beta_3)^\prime.
\]

 \paragraph*{Period-interacted OLS}
Stacking the independent variables (the constant 1, the treatment indicator $D_{it}$, the period indicator $I(t=2)$, and the interaction between the treatment indicator and the period indicator $D_{it}I(t=2)$) for the $2n$ observations unit by unit (period 1 then period 2), with units in sequence (1,0) first and those in 
(0,1) second, yields the design matrix
\[
X=
\left(
\begin{array}{c}
\begin{array}{cccc}
1&1&0&0\\
1&0&1&0\\
\vdots&\vdots&\vdots&\vdots\\
1&1&0&0\\
1&0&1&0
\end{array}\\
\hline
\begin{array}{cccc}
1&0&0&0\\
1&1&1&1\\
\vdots&\vdots&\vdots&\vdots\\
1&0&0&0\\
1&1&1&1
\end{array}
\end{array}
\right).
\]
The upper block contains $n_1$ units and $2n_1$ rows, and the lower block contains $n_2$ units and $2n_2$ rows. We can obtain
\[
X^\prime X=
\begin{pmatrix}
2n&n&n&n_2\\
n&n&n_2&n_2\\
n&n_2&n&n_2\\
n_2&n_2&n_2&n_2
\end{pmatrix},
\qquad
(X^\prime X)^{-1}=
\begin{pmatrix}
\frac{1}{n_2}&-\frac{1}{n_2}&-\frac{1}{n_2}&\frac{1}{n_2}\\
-\frac{1}{n_2}&k&\frac{1}{n_2}&-k\\
-\frac{1}{n_2}&\frac{1}{n_2}&k&-k\\
\frac{1}{n_2}&-k&-k&2k
\end{pmatrix},
\]
where $k\equiv \tfrac{1}{n_1}+\tfrac{1}{n_2}$.

Define the four cell totals as
\begin{align*}
&A\equiv\sum_{i:(D_{i1},D_{i2})=(1,0)}Y_{i1},\qquad
B\equiv\sum_{i:(D_{i1},D_{i2})=(1,0)}Y_{i2},\\
&C\equiv\sum_{i:(D_{i1},D_{i2})=(0,1)}Y_{i1},\qquad
E\equiv\sum_{i:(D_{i1},D_{i2})=(0,1)}Y_{i2},
\end{align*}
and the corresponding four cell means are
\[
\bar A=\bar{Y}_{1,(1,0)},\qquad \bar B=\bar{Y}_{2,(1,0)},\qquad
\bar C=\bar{Y}_{1,(0,1)},\qquad \bar E=\bar{Y}_{2,(0,1)}.
\]
The four columns of $X$ select, respectively, all observations, treated
observations, Period-2 observations, and treated Period-2 observations. Hence
\begin{equation}
\begin{aligned}
\hat\beta
&=(X^\prime X)^{-1}X^\prime Y\\
&=(X^\prime X)^{-1}
\begin{pmatrix}
A+B+C+E\\ A+E\\ B+E\\ E
\end{pmatrix}
=
\begin{pmatrix}
\bar C\\
\bar A-\bar C\\
\bar B-\bar C\\
\bar C+\bar E-\bar A-\bar B
\end{pmatrix}
=
\begin{pmatrix}
\bar{Y}_{1,(0,1)}\\
\bar{Y}_{1,(1,0)}-\bar{Y}_{1,(0,1)}\\
\bar{Y}_{2,(1,0)}-\bar{Y}_{1,(0,1)}\\
\bar{Y}_{1,(0,1)}+\bar{Y}_{2,(0,1)}-\bar{Y}_{1,(1,0)}-\bar{Y}_{2,(1,0)}
\end{pmatrix}.
\end{aligned}
\label{eq:saturated-beta-matrix}
\end{equation}
The matrix calculation therefore yields
\[
\hat\beta_1=\hat\mu_1,\qquad
\hat\beta_1+\hat\beta_3=\hat\mu_2,\qquad
\hat\beta_3=\hat\mu_2-\hat\mu_1=\hat\delta,\qquad
\hat\beta_1+\frac12\hat\beta_3=\hat\mu_{co}.
\]

We next derive the CR0 variance estimator for $\hat\beta$ when inference is clustered by unit. It is easy to see that the fitted values and residuals from the model are
\[
\hat Y=
\begin{pmatrix}
\bar A\\ \bar B\\ \vdots\\ \bar A\\ \bar B\\
\hline
\bar C\\ \bar E\\ \vdots\\ \bar C\\ \bar E
\end{pmatrix},
\qquad
\hat e=
\begin{pmatrix}
Y_{11}-\bar A\\ Y_{12}-\bar B\\ \vdots\\
Y_{n_1,1}-\bar A\\ Y_{n_1,2}-\bar B\\
\hline
Y_{n_1+1,1}-\bar C\\ Y_{n_1+1,2}-\bar E\\ \vdots\\
Y_{n,1}-\bar C\\ Y_{n,2}-\bar E
\end{pmatrix}.
\]
Write the two within-unit residuals as
\[
a_i\equiv Y_{i1}-\bar A,\quad b_i\equiv Y_{i2}-\bar B,\quad \text{if }(D_{i1},D_{i2})=(1,0);
\]
\[
c_i\equiv Y_{i1}-\bar C,\quad d_i\equiv Y_{i2}-\bar E, \quad \text{if }(D_{i1},D_{i2})=(0,1).
\]
The CR0 variance estimator for $\hat\beta$ is
$$
\widehat{\mathbb{V}}_{\mathrm{CR0}}(\hat\beta)
=\left(X^{\prime}X\right)^{-1}\Bigl(\textstyle\sum_{i=1}^{n}X_{i}^{\prime}\hat{e}_{i}\hat{e}_{i}^{\prime}X_{i}\Bigr)\left(X^{\prime}X\right)^{-1}=\left(X^{\prime}X\right)^{-1}\Bigl(\textstyle\sum_{i=1}^{n}s_is_i'\Bigr)\left(X^{\prime}X\right)^{-1},
$$
where $X_{i}$ is the $2\times4$ block of $X$ for unit $i$, $\hat e_{i}$ is the
residual vector for unit $i$, and $s_{i}\equiv X_{i}^{\prime}\hat e_{i}$.

For units with $(D_{i1},D_{i2})=(1,0)$,
\begin{align*}
s_i&=\begin{pmatrix}
a_i+b_i\\ a_i\\ b_i\\ 0
\end{pmatrix},\\
s_{i}s^{\prime}_{i}&=\left(\begin{array}{cccc}
\left(a_{i}+b_{i}\right)^{2} & a^{2}_{i}+a_{i}b_{i} & b^{2}_{i}+a_{i}b_{i} & 0\\
a^{2}_{i}+a_{i}b_{i} & a^{2}_{i} & a_{i}b_{i} & 0\\
b^{2}_{i}+a_{i}b_{i} & a_{i}b_{i} & b^{2}_{i} & 0\\
0 & 0 & 0 & 0
\end{array}\right).
\end{align*}
For units with $(D_{i1},D_{i2})=(0,1)$,
\begin{align*}
s_i&=
\begin{pmatrix}
c_i+d_i\\ d_i\\ d_i\\ d_i
\end{pmatrix},\\
s_{i}s^{\prime}_{i}&=\left(\begin{array}{cccc}
\left(c_{i}+d_{i}\right)^{2} & d^{2}_{i}+c_{i}d_{i} & d^{2}_{i}+c_{i}d_{i} & d^{2}_{i}+c_{i}d_{i}\\
d^{2}_{i}+c_{i}d_{i} & d^{2}_{i} & d^{2}_{i} & d^{2}_{i}\\
d^{2}_{i}+c_{i}d_{i} & d^{2}_{i} & d^{2}_{i} & d^{2}_{i}\\
d^{2}_{i}+c_{i}d_{i} & d^{2}_{i} & d^{2}_{i} & d^{2}_{i}
\end{array}\right).
\end{align*}

Define
\[
\begin{gathered}
F\equiv\sum_{i:(D_{i1},D_{i2}=(1,0)}a_i^2,\quad
G\equiv\sum_{i:(D_{i1},D_{i2}=(1,0)}b_i^2,\quad
H\equiv\sum_{i:(D_{i1},D_{i2}=(1,0)}a_ib_i,\\
U\equiv\sum_{i:(D_{i1},D_{i2}=(0,1)}c_i^2,\quad
V\equiv\sum_{i:(D_{i1},D_{i2}=(0,1)}d_i^2,\quad
W\equiv\sum_{i:(D_{i1},D_{i2}=(0,1)}c_id_i.
\end{gathered}
\]
We have
\[
\sum_{i=1}^n s_is_i'=
\begin{pmatrix}
F+G+2H+U+V+2W&F+H+V+W&G+H+V+W&V+W\\
F+H+V+W&F+V&H+V&V\\
G+H+V+W&H+V&G+V&V\\
V+W&V&V&V
\end{pmatrix}.
\]

We first extract the variance for $\hat\beta_3=\hat\delta$. The fourth row of $(X^\prime X)^{-1}$ is
\[
v^\prime=\left(\frac1{n_2},-k,-k,2k\right),
\]
so
$$
\widehat{\mathbb{V}}_{\mathrm{CR0}}(\hat\beta_3)
=v^{\prime}\Bigl(\textstyle\sum_{i=1}^n s_{i}s_{i}^{\prime}\Bigr)v=\sum_{i=1}^n(v^\prime s_i)^2.
$$
For the two sequences,
\begin{align*}
v^\prime s_i
&=-\frac{a_i+b_i}{n_1}, \quad \text{if }(D_{i1},D_{i2})=(1,0);\\
v^\prime s_i
&=\frac{c_i+d_i}{n_2},  \quad \text{if }(D_{i1},D_{i2})=(0,1).
\end{align*}
It follows that
\begin{align}
\widehat{\mathbb V}_{\mathrm{CR0}}(\hat\beta_3)
&=\frac{F+G+2H}{n_1^2}+\frac{U+V+2W}{n_2^2}\nonumber\\
&=\frac1{n_1}S^2_{(1,0)}\!\left(Q_i(1,0)\right)
 +\frac1{n_2}S^2_{(0,1)}\!\left(Q_i(0,1)\right),
\label{eq:saturated-delta-cr0}
\end{align}
where
\[
Q_i(d_1,d_2)\equiv Y_{i1}(d_1,d_2)+Y_{i2}(d_1,d_2).
\]
This exactly equals $\hat\sigma_\delta^2/n$, the design-based variance of $\hat\delta$.

In order to calculate the variance for $\hat\mu_{co}=\hat\beta_1+\frac{1}{2}\hat{\beta}_3$, we define
\[
\ell_{co}\equiv
\begin{pmatrix}0\\1\\0\\1/2\end{pmatrix},
\qquad
q^\prime\equiv\ell_{co}^\prime(X^\prime X)^{-1}
=\left(
-\frac1{2n_2},\frac{k}{2},
\frac1{n_2}-\frac{k}{2},0
\right).
\]
Then
\[
\widehat{\mathbb V}_{\mathrm{CR0}}
\left(\hat\mu_{co}\right)=q^{\prime}\Bigl(\textstyle\sum_{i=1}^n s_{i}s_{i}^{\prime}\Bigr)q
=\sum_{i=1}^n(q^\prime s_i)^2,
\]
with
\begin{align*}
q^\prime s_i
&=\frac{a_i-b_i}{2n_1}
=\frac{\Delta_i(1,0)-\bar\Delta_{(1,0)}}{2n_1},
\quad \text{if }(D_{i1},D_{i2})=(1,0),\\
q^\prime s_i
&=\frac{d_i-c_i}{2n_2}
=\frac{\Delta_i(0,1)-\bar\Delta_{(0,1)}}{2n_2},
\quad \text{if }(D_{i1},D_{i2})=(0,1).
\end{align*}
Consequently,
\begin{equation}
\begin{aligned}
\widehat{\mathbb V}_{\mathrm{CR0}}
\left(\hat\mu_{co}\right)
&=\frac1{4n_1^2}\sum_{i:(1,0)}
\left\{\Delta_i(1,0)-\bar\Delta_{(1,0)}\right\}^2\\
&\quad+\frac1{4n_2^2}\sum_{i:(0,1)}
\left\{\Delta_i(0,1)-\bar\Delta_{(0,1)}\right\}^2\\
&=\frac14\left(\frac{S_1^2}{n_1}+\frac{S_2^2}{n_2}\right).
\end{aligned}
\label{eq:saturated-pooled-cr0}
\end{equation}
This exactly equals $\hat\sigma_{co}^2/(4n)$, the design-based variance of $\hat\mu_{co}$.

\paragraph*{Period-interacted random effects model} 
The same equivalence holds if Equation~\eqref{eq:saturated-regression-appendix} is estimated as a random-effects model, where the error $e_{it}$ can be decomposed into a random intercept and an idiosyncratic component:
$$
e_{it}=u_{i}+\varepsilon_{it},$$
where $u_{i}\sim(0,\sigma_{u}^{2})$ and $\varepsilon_{it}\sim(0,\sigma_{\varepsilon}^{2})$ are mutually independent with $\sigma_u^2>0$ and $\sigma_{\varepsilon}^2>0$.

To see why, write the covariance matrix of $e_i=(e_{i1},e_{i2})'$ for each unit as
\[
\Sigma_e=\sigma_{e}^{2}I_{2}
 +\sigma_{u}^{2}\boldsymbol{1}_{2}\boldsymbol{1}_{2}^{\prime},
\]
where $I_d$ is an $d\times d$ identity matrix, and $\boldsymbol{1}_d$ is a vector of $d$ ones. The covariance matrix of error terms for all units is 
$$\Omega=I_{n}\otimes\Sigma_e,$$
where $I_n$ is the $n\times n$ identity matrix, and $\otimes$ is the Kronecker product. The random effects estimator of the coefficients is given by the generalized least squares formula
$$
\hat\beta^{\mathrm{RE}}=\left(X^{\prime}\Omega^{-1}X\right)^{-1}X^{\prime}\Omega^{-1}Y,
$$
where $Y$ is the vector of observed outcomes arranged in the same order as the rows of 
$X$. The CR0 variance estimator of $\hat\beta^{\mathrm{RE}}$ is
$$\widehat{\mathbb V}_{CR0}^{\mathrm{RE}}(\hat\beta^{\mathrm{RE}})
=\left( X^{\prime}\Omega^{-1}X\right)^{-1}\Bigl(\sum_{i=1}^n X_{i}^{\prime}\Sigma_{e}^{-1}\hat e_{i}
 \hat e_{i}^{\prime}\Sigma_{e}^{-1}X_{i}\Bigr)\left(X^{\prime}\Omega^{-1}X\right)^{-1}.
$$

Let $$\bar{J}=I_n\otimes\left(\begin{array}{cc}\frac{1}{2}&\frac{1}{2}\\\frac{1}{2}&\frac{1}{2}\end{array}\right).$$ It is easy to verify that 
\[\bar J'=J,\quad \bar J^2=\bar J,\quad \bar JX=XK,\]
where 
$$K=\begin{pmatrix}
1&\tfrac12&\tfrac12&\tfrac12\\
0&0&0&-\tfrac12\\
0&0&0&-\tfrac12\\
0&0&0&1
\end{pmatrix}.$$
It is easy to verify that $K^2=K$.

Note that $\Omega$ can be written as 
\[
\Omega=\sigma_e^2I_{2n}+2\sigma_u^2\bar J=\sigma_e^2 (I_{2n}-\bar J)+\sigma_t^2\bar J,\quad \sigma_t^2=\sigma_e^2+2\sigma_u^2.
\]
The second equality for $\Omega$ gives the spectral decomposition of $\Omega$, with $(I_{2n}-\bar J)$ and $\bar J$ being complementary idempotent projection matrices. Therefore we can write
\[
\Omega^{-1}=\frac{1}{\sigma_e^2} (I_{2n}-\bar J)+\frac{1}{\sigma_t^2}\bar J.
\]

Let 
$M\equiv\sigma_{e}^{2}(I_{4}-K)+\sigma_t^{2}K$. Then
\[
\Omega X = \sigma_{\varepsilon}^2X+2\sigma_u^2\bar JX=\sigma_{\varepsilon}^2X + 2\sigma_u^2XK=X\left(\sigma_{e}^{2}I_{4}+2\sigma_u^{2}K\right)=XM.
\]
Because $I_4-K$ and $K$ are complementary idempotent matrices, $M$ is invertible. Multiplying both sides of $\Omega X = XM$ on the left by $\Omega^{-1}$ and on the right by $M^{-1}$, we can get 
$$\Omega^{-1}\Omega X M^{-1}=\Omega^{-1}XM M^{-1},$$
or $XM^{-1}=\Omega^{-1}X$.

Because $\Omega$ is symmetric, $X^{\prime}\Omega X$ is symmetric. Plugging in $\Omega X=XM$, we obtain that $X^{\prime}\Omega X=X^{\prime}XM$ is symmetric, so
$X^{\prime}XM=M^{\prime}X^{\prime}X$. Multiplying both sides on the left by $(X'X)^{-1}$ and on the right by $(X'X)^{-1}$ gives
$M(X^{\prime}X)^{-1}=(X^{\prime}X)^{-1}M^{\prime}$. Applying this identity and substituting $\Omega^{-1}X=XM^{-1}$ and its transposed version, 
$X'\Omega^{-1}=(XM^{-1})^{\prime}$, into the generalized least squares coefficient map,
\begin{align*}
&\left(X^{\prime}\Omega^{-1}X\right)^{-1}X^{\prime}\Omega^{-1}
=\left(X'XM^{-1}\right)^{-1}(XM^{-1})'
=M\left(X^{\prime}X\right)^{-1}(M^{-1})^{\prime}X^{\prime}\\
&=\left(X^{\prime}X\right)^{-1}M^{\prime}(M^{-1})^{\prime}X^{\prime}
=\left(X^{\prime}X\right)^{-1}X^{\prime}.
\end{align*}
The generalized least squares coefficient map, $\left(X^{\prime}\Omega^{-1}X\right)^{-1}X^{\prime}\Omega^{-1}$, is therefore identical to the OLS coefficient map, $\left(X^{\prime}X\right)^{-1}X^{\prime}$. Consequently, the coefficient estimates from the random effects model are identical to those from the OLS model.

The equality of the coefficient maps can also be used to show that the CR0 variance estimators from the two models are equal.
Let $L_i$ be the $4\times 2$ block of columns of $(X'X)^{-1}X'$ associated with unit $i$, so that $L_i=(X'X)^{-1}X_i'$. Then
\[
\hat {\mathbb V}_{CR0}^{\textrm{OLS}}(\hat\beta)
=\bigl(X^{\prime}X\bigr)^{-1}
 \Bigl(\sum_{i}X_{i}^{\prime}\hat e_{i}\hat e_{i}^{\prime}X_{i}\Bigr)
 \bigl(X^{\prime}X\bigr)^{-1}
=\sum_{i}L_{i}\hat e_{i}\hat e_{i}^{\prime}L_{i}^{\prime}
\]
The CR0 variance estimator from the random effects model is
\[
\hat{\mathbb V}_{CR}^{\textrm{RE}}(\hat\beta)
=(X'\Omega^{-1}X)^{-1}\Bigl(\sum_{i}X_{i}^{\prime}\Sigma_{e}^{-1}\hat e_{i}
 \hat e_{i}^{\prime}\Sigma_{e}^{-1}X_{i}\Bigr)(X'\Omega^{-1}X)^{-1},
\]
which likewise factors as $\sum_{i}L_{i}^{\textrm{RE}}\hat e_{i}\hat e_{i}^{\prime}
(L_{i}^{\textrm{RE}})^{\prime}$ with
$L_{i}^{\textrm{RE}}=(X'\Omega^{-1}X)^{-1}X_{i}^{\prime}\Sigma_{e}^{-1}$. 
$L_{i}^{\textrm{RE}}$ is exactly the block of columns of $(X'\Omega^{-1}X)^{-1}X'\Omega^{-1}$ associated with unit $i$. The equality of the coefficient maps thus gives $L_{i}^{\textrm{RE}}=L_{i}$ for every $i$.

Two remarks serve to clarify this result. First, CR0 is the
unadjusted sandwich corresponding to the sample-variance convention used in
the main text. Software-specific CR1 or CR2 corrections can produce small
finite-sample differences unless the same correction is applied to both fits.
Second, the equivalence relies on complete two-period pairs and on the stated
column space. Attrition, time-varying covariates, random slopes, or other
changes that destroy its invariance can make the random-effects and OLS point
estimates differ. 

The same arguments can be used to show that the treatment-only random effects model yields identical coefficient and variance estimates to the treatment-only OLS model. It suffices to replace $K$ and $M$ with
\[
K_{0}=\begin{pmatrix}1&\tfrac12\\[2pt]0&0\end{pmatrix},\quad M_{0}=\sigma_{\varepsilon}^{2}(I_{2}-K_{0})
+\sigma_{t}^{2}K_{0}.
\]

\subsection{Power comparison}\label{app:power}

\begin{prop}[The carryover test is weakly less powerful]\label{prop:power-comparison}
Consider level-$\alpha$ two-sided tests of $H_{0}:\delta=0$ based on
$\hat{\delta}$ and of $H_{0}:\tau=0$ based on $\hat{\mu}_{1}$. Suppose
$|\delta|\le|\tau|$ and, within both treatment sequences,
\begin{align*}
\mathbb{V}\left(Y_{i1}(1,0)\right)
&\le\mathbb{V}\left(Y_{i1}(1,0)+Y_{i2}(1,0)\right),\\
\mathbb{V}\left(Y_{i1}(0,1)\right)
&\le\mathbb{V}\left(Y_{i1}(0,1)+Y_{i2}(0,1)\right).
\end{align*}
Then the asymptotic power of the carryover test is no greater than that of
the treatment-effect test.
\end{prop}

\begin{proof}
Let $z=z_{1-\alpha/2}$. For a standardized normal statistic with
noncentrality $\lambda$, the two-sided rejection probability equals
\[
\Psi(\lambda)=1-\Phi(z-\lambda)+\Phi(-z-\lambda).
\]
Because the standard normal distribution is symmetric around zero, $\Psi(-\lambda)=\Psi(\lambda)$. Therefore, the two-sided rejection probability only dependens on $a=|\lambda|$, and equals
\[
\Psi(a)=1-\Phi(z-a)+\Phi(-z-a),\quad a\geq 0.
\]
It is nondecreasing in $a$ because
\[
\Psi'(a)=\phi(z-a)-\phi(z+a)\ge0.
\]
The two tests have noncentrality magnitudes
\[
|\lambda_{\delta}|=\frac{\sqrt{n}|\delta|}{\sigma_{\delta}},
\qquad
|\lambda_{1}|=\frac{\sqrt{n}|\tau|}{\sigma_{1}}.
\]
The two variance inequalities, multiplied by the positive weights
$n/n_{1}$ and $n/n_{2}$, imply $\sigma_{\delta}\ge\sigma_{1}$. Together with
$|\delta|\le|\tau|$, this implies
$|\lambda_{\delta}|\le|\lambda_{1}|$, so monotonicity of $\Psi$ proves the
power ordering.
\end{proof}

The proof only requires the aggregate ordering
$\sigma_{\delta}\ge\sigma_{1}$. By
Equation~\ref{eq:sigma-delta-identity}, $\sigma_{12}\ge0$ is one sufficient
condition for that ordering. Nonnegative cross-period covariance within each
sequence is, in turn, sufficient for the two variance inequalities stated in the proposition.

\subsection{Formal coverage results for the two-step procedure}\label{app:grizzle}

\paragraph*{The two-step procedure.} 
To test $H_0:\delta=0$, we reject the null 
hypothesis when $|\hat\delta|>z\sqrt{\hat\sigma_{\delta}^2/n},$
where $z\equiv z_{1-\alpha/2}$.

The point estimator is:
$$
\hat\theta=\begin{cases}
\hat\mu_{co} & \text{if } |\hat\delta|\leq z\sqrt{\hat\sigma_{\delta}^2/n},\\[4pt]
\hat\mu_1 & \text{if } |\hat\delta|>z\sqrt{\hat\sigma_{\delta}^2/n}.
\end{cases}
$$
The reported confidence interval is 
$$
CI=\begin{cases}
\hat\mu_{co}\pm z\sqrt{\hat \sigma_{co}^2/(4n)} & \text{if } |\hat\delta|\leq z\sqrt{\hat\sigma_{\delta}^2/n},\\[4pt]
\hat\mu_1\pm z\sqrt{\hat \sigma_1/n} & \text{if } |\hat\delta|>z\sqrt{\hat\sigma_{\delta}^2/n}.
\end{cases}
$$

\paragraph*{Reduction to known variances.} The estimators
$\hat{\sigma}_{1}^{2},\hat{\sigma}_{2}^{2},\hat{\sigma}_{12}$ based on the
sample moments are consistent under
Assumption~\ref{as:regularity}, and therefore
$\hat{\sigma}_{\delta}^{2}$ and $\hat{\sigma}_{co}^{2}$ are also consistent by
Equations~\eqref{eq:sigma-delta-identity}--\eqref{eq:variance-identities}. Both the selection event (whether $|\hat\delta|>z\sqrt{\hat{\sigma}_{\delta}^2/n}$) and the reported endpoints for CI are continuous at the true values of the variances. Therefore, asymptotically, the two-step procedure is equivalent to the corresponding procedure with the variance estimators replaced by the true variances.

We can therefore study the asymptotically equivalent procedure defined by
\[
\hat{\theta}^{*}=\begin{cases}
\hat{\mu}_{co} & \text{if }|\hat{\delta}|\le z_{\alpha_{1}}\sqrt{\sigma_{\delta}^{2}/n},\\[4pt]
\hat{\mu}_{1} & \text{if }|\hat{\delta}|>z_{\alpha_{1}}\sqrt{\sigma_{\delta}^{2}/n},
\end{cases}
\qquad
CI^{*}=\begin{cases}
\hat{\mu}_{co}\pm z_{\alpha}\sqrt{\sigma_{co}^{2}/(4n)} & \text{if }|\hat{\delta}|\le z_{\alpha_{1}}\sqrt{\sigma_{\delta}^{2}/n},\\[4pt]
\hat{\mu}_{1}\pm z_{\alpha}\sqrt{\sigma_{1}^{2}/n} & \text{if }|\hat{\delta}|>z_{\alpha_{1}}\sqrt{\sigma_{\delta}^{2}/n}.
\end{cases}
\]

\paragraph*{Standardization.} Define
\[
Z_{\delta}=\frac{\sqrt{n}\hat{\delta}}{\sigma_{\delta}}-\eta,
\qquad
Z_{1}=\frac{\sqrt{n}(\hat{\mu}_{1}-\mu_{1})}{\sigma_{1}},
\qquad
Z_{co}=\frac{2\sqrt{n}(\hat{\mu}_{co}-\mu_{1})}{\sigma_{co}}-b_{co},
\]
where
\[
\eta=\frac{\sqrt{n}\,\delta}{\sigma_{\delta}},
\qquad
b_{co}=\frac{\sqrt{n}\,\delta}{\sigma_{co}}.
\] By Proposition~\ref{prop:joint-asymptotic},
$(Z_{\delta},Z_{1})\overset{d}{\to}\mathcal{N}(0,R_1)$ with
\[
R_1=\begin{pmatrix}1 & \rho_{\delta,1}\\ \rho_{\delta,1} & 1\end{pmatrix},
\quad \text{where }
\rho_{\delta,1}=-\frac{\sigma_{1}^{2}+\sigma_{12}}{\sigma_{1}\sigma_{\delta}};
\]
$(Z_{\delta},Z_{co})\overset{d}{\to}\mathcal{N}(0,R_2)$ with
\[
R_2=\begin{pmatrix}1 & \rho_{\delta,co}\\ \rho_{\delta,co} & 1\end{pmatrix},
\quad\text{where }
\rho_{\delta,co}=\frac{\sigma_{2}^{2}-\sigma_{1}^{2}}{\sigma_{co}\sigma_{\delta}}.
\]
Asymptotically, the carryover test fails to reject when
$|Z_{\delta}+\eta|\le z_{\alpha_{1}}$; the pooled interval covers $\mu_{1}$
when $|Z_{co}+b_{co}|\le z_{\alpha}$; the post-only interval covers
$\mu_{1}$ when $|Z_{1}|\le z_{\alpha}$.

\paragraph*{The asymptotic coverage probability.} Asymptotically, $CI^{*}$ covers $\mu_{1}$
when one of these two disjoint events occurs: the test fails to reject
and the pooled interval covers $\mu_1$, or the test rejects and the post-only interval
covers $\mu_1$. Therefore, the asymptotic coverage probability is 
\[
\text{ACP}=\Pr\!\left(|Z_{\delta}+\eta|\le z,\ |Z_{co}+b_{co}|\le z\right)
+\Pr\!\left(|Z_{\delta}+\eta|>z,\ |Z_{1}|\le z\right).
\]
Since $\Pr(|Z_{1}|\le z)=1-\alpha$, we can also write
\begin{equation}
\text{ACP}=\left(1-\alpha\right)+T_{\text{pool}}-T_{1},
\label{eq:grizzle-coverage}
\end{equation}
\begin{equation}
T_{\text{pool}}=\Pr\!\left(|Z_{\delta}+\eta|\le z,\ |Z_{co}+b_{co}|\le z\right),
\qquad
T_{1}=\Pr\!\left(|Z_{\delta}+\eta|\le z,\ |Z_{1}|\le z\right).
\label{eq:coverage-terms}
\end{equation}

\begin{lem}\label{lem:rectangle}
Define $R(\rho)\equiv\Pr\!\big(|X|\le z,|Y|\le z\big)$
for $(X,Y)$ standard bivariate normal with correlation $\rho$. Then $R$ is even in $\rho$ and strictly increasing in $|\rho|$, with
$R(0)=(1-\alpha)^2$.
\end{lem}
\begin{proof}
Let $\Phi_2(a,b;\rho)$ denote the 
cumulative distribution function of a standard bivariate normal with correlation $\rho$. 
We can write 
$$R(\rho)=\Phi_2(z,z;\rho)-\Phi_2(-z,z;\rho)
-\Phi_2(z,-z;\rho)+\Phi_2(-z,-z;\rho).$$
Therefore $R(-\rho)=R(\rho)$ due to the symmetry of the bivariate normal distribution. When $\rho=0$, $X$ and $Y$ are independent, so $R(0)=\Pr(|X|\le z)\cdot \Pr(|Y|\le z)=(1-\alpha)^2$.

By Plackett's formula $\partial_\rho\Phi_2(a,b;\rho)=\phi_2(a,b;\rho)$, 
where $\phi_2(a,b;\rho)$ is the probability density function of a standard 
bivariate normal with correlation $\rho$:
\[
\phi_2(a,b;\rho) = 
\frac{1}{2\pi\sqrt{1-\rho^2}}
\exp\left(
-\frac{a^2 - 2\rho ab + b^2}{2(1-\rho^2)}
\right).
\]
Apparently $\phi_2(-a,-b;\rho)=\phi_2(a,b;\rho)$, therefore we have
$$R'(\rho)=2\big[\phi_2(z,z;\rho)-\phi_2(z,-z;\rho)\big].$$
Since
$$\phi_2(z,z;\rho)=\frac{1}{2\pi\sqrt{1-\rho^2}}\exp\!\left(-\frac{z^2}{1+\rho}\right),\qquad
\phi_2(z,-z;\rho)=\frac{1}{2\pi\sqrt{1-\rho^2}}\exp\!\left(-\frac{z^2}{1-\rho}\right)$$
we have $R'(\rho)>0$ for $\rho>0$. By symmetry, $R'(\rho)<0$ for $\rho<0$, and 
$R$ is strictly increasing in $|\rho|$.
\end{proof}

\paragraph*{A reduced parameterization.} We normalize $\sigma_1^2=1$ and write
\begin{equation}
r\equiv\frac{\sigma_2^2}{\sigma_1^2},\qquad g\equiv\frac{\sigma_{12}}{\sigma_1^2},
\label{eq:omega-zeta}
\end{equation}
so that $\sigma_\delta^2=1+r+2g$, $\sigma_{co}^2=1+r-2g$,
$\rho_{\delta,1}=-\dfrac{1+g}{\sqrt{1+r+2g}}$, and
$\rho_{\delta,co}=\dfrac{r-1}{\sqrt{(1+r+2g)(1+r-2g)}}$.

The asymptotic covariance matrix of $(\sqrt{n}\hat\mu_1,\sqrt{n}\hat\mu_2)$ has
diagonal terms $(\sigma_1^2,\sigma_2^2)$ and off-diagonal term $-\sigma_{12}$,
so positive semidefiniteness requires $\sigma_{12}^2\le\sigma_1^2\sigma_2^2$, i.e.
\begin{equation}
g^2\le r,
\label{eq:feasible}
\end{equation}
with equality only in the degenerate case $|\mathrm{Corr}(\hat\mu_1,\hat\mu_2)|=1$.

The pooled estimator has smaller variance than the first half effect estimator when
$$\tfrac14\sigma_{co}^2<\sigma_1^2\iff\sigma_{12}>\tfrac12\!\left(\sigma_2^2-3\sigma_1^2\right)\iff g>\tfrac{r-3}{2},$$

\begin{prop}[With a zero carryover gap, undercoverage holds exactly when pooling is efficient]\label{prop:undercoverage}
Suppose $\delta=0$ and the non-degenerate case with $g^{2}<r$. Then, $$ACP<1-\alpha\iff\tfrac14\sigma_{co}^2<\sigma_1^2,\qquad ACP=1-\alpha\iff\tfrac14\sigma_{co}^2=\sigma_1^2,\qquad ACP>1-\alpha\iff\tfrac14\sigma_{co}^2>\sigma_1^2.$$
\end{prop}

\begin{proof}
At $\delta=0$, $\eta=b_{co}=0$, so
$T_{\text{pool}}=R(\rho_{\delta,co})$ and $T_{1}=R(\rho_{\delta,1})$ in the
notation of Lemma~\ref{lem:rectangle}. The lemma therefore gives
\[
\mathrm{sign}\left\{\text{ACP}-(1-\alpha)\right\}
=\mathrm{sign}\left\{|\rho_{\delta,co}|-|\rho_{\delta,1}|\right\}
=\mathrm{sign}\left\{\rho_{\delta,co}^{2}-\rho_{\delta,1}^{2}\right\}.
\]
In the parameterization of Equation~\eqref{eq:omega-zeta},
$$\rho_{\delta,co}^2-\rho_{\delta,1}^2=
\frac{(r-1)^2}{(1+r+2g)(1+r-2g)}-\frac{(1+g)^2}{1+r+2g}=\frac{(r-1)^2-(1+r-2g)(1+g)^2}{(1+r+2g)(1+r-2g)}.$$
The denominator is positive, so the sign is that of the numerator, which factors:
$$\;(r-1)^2-(1+r-2g)(1+g)^2=(2g+3-r)\,(g^2-r).\;$$
Because $g^2<r$, the second factor is strictly negative, so the sign of $\rho_{\delta,co}^{2}-\rho_{\delta,1}^{2}$ is that of
$-(2g+3-r)$. Finally
$2g+3-r>0\iff g>(r-3)/2\iff\tfrac{1}{4}\sigma_{co}^{2}<\sigma_{1}^{2}$,
and likewise for the other two cases.
\end{proof}

\begin{rem}
Proposition \ref{prop:undercoverage} says that when $\delta=0$, the two-step confidence interval asymptotically undercovers if and only if the pooled estimator is strictly more efficient than the post-only estimator. This means that the two-step procedure undercovers precisely when it is worth running.
\end{rem}

\begin{rem}\label{rem:fixed-delta}
If $\delta\neq 0$, then asymptotically $\eta\to \infty$ and $b_{co}\to \infty$ when
$\delta>0$, whereas $\eta\to -\infty$ and $b_{co}\to -\infty$ when
$\delta<0$. This implies $T_{\text{pool}}\to 0$ and $T_{1}$$\to 0$, and hence $ACP\to 1-\alpha$.
\end{rem}

We report the finite-sample behavior of Equation~\eqref{eq:grizzle-coverage} in~\ref{app:undercoverage-simulation}, both under the procedure actually used by researchers (which uses
$\hat{\sigma}$'s rather than $\sigma$'s) and under non-normal outcomes.

\subsection{Derivation of the sensitivity threshold}\label{app:sensitivity-derivation}

Under the hypothesis that the average carryover gap equals $\delta^{\star}$,
the bias-adjusted estimator is
$\hat{\tau}(\delta^{\star})\equiv\hat{\mu}_{co}-\delta^{\star}/2$. By
Equation~\eqref{eq:pooled-asymptotic}, a two-sided test of $H_{0}:\tau=0$ using
$\hat\tau(\delta^\star)$ fails to reject when
\[
\left|\frac{\hat{\mu}_{co}-\delta^{\star}/2}{s_{co}}\right|\le z_{1-\alpha/2},
\qquad \text{where }s_{co}\equiv\widehat{\mathrm{se}}(\hat{\mu}_{co}).
\]
Writing $\delta_{\mathrm{std}}^{\star}\equiv\delta^{\star}/\hat{\mu}_{co}$ and
$t_{co}\equiv\hat{\mu}_{co}/s_{co}$, this condition becomes
\[
|t_{co}|\left|1-\frac{\delta_{\mathrm{std}}^{\star}}{2}\right|\le z_{1-\alpha/2},
\qquad\text{i.e.}\qquad
\delta_{\mathrm{std}}^{\star}\in
\left[2-\frac{2z_{1-\alpha/2}}{|t_{co}|},\;2+\frac{2z_{1-\alpha/2}}{|t_{co}|}\right].
\]
For an originally significant pooled estimate ($|t_{co}|>z_{1-\alpha/2}$), the
lower endpoint of this interval is the standardized minimum carryover gap
required to overturn significance,
$\delta_{\min,\mathrm{std}}^{*}\equiv2(1-z_{1-\alpha/2}/|t_{co}|)$, reported in
Section~\ref{sec:Sensitivity-Analysis}. 

\input{appendix_regression_simulation}
\input{appendix_undercoverage_simulation}
\input{appendix_litreview_parameters}
\FloatBarrier
\input{appendix_csp_calibration}

%% file: appendix_regression_simulation.tex
\section{Simulations} \label{app:simulations}

\subsection{Finite-sample corrections used in Table~\ref{tab:Comparison-of-Sample}}\label{app:se-corrections}

All standard errors in Table~\ref{tab:Comparison-of-Sample} are unit-clustered
sandwich estimators, except Design-based, which is the raw, unadjusted CR0
formula of \ref{app:bare-pooled} and \ref{app:saturated} with no finite-sample
correction. Writing $n=n_1+n_2$ for the number of unit clusters and
$N_{\mathrm{obs}}$ for the number of rows in the fitted regression: (0)
Paired $t$-test uses the classical paired-difference correction, CR0
$\times\sqrt{n/(n-1)}$.

(1)--(4) use Stata's default cluster-robust (CR1S)
correction, CR0 $\times\sqrt{(n/(n-1))\cdot(N_{\mathrm{obs}}-1)/(N_{\mathrm{obs}}-p)}$,
with $p=2$ for (1)--(2) and $p=4$ for (3)--(4), the number of columns in
the two design matrices of \ref{app:bare-pooled} and \ref{app:saturated}
respectively ($p$ here is unrelated to the $k\equiv1/n_1+1/n_2$ of
\ref{app:saturated}), while $N_{\mathrm{obs}}=2n$ in all four (both periods
pooled).

(5) uses the same CR1S correction with $p=2$, but
$N_{\mathrm{obs}}=n$ (period-1 only), so each cluster contains a single
observation. \ref{app:saturated} shows the period-interacted specification's
CR0 estimator equals the design-based CR0 exactly for any $n_1,n_2$, so applying its CR1S factor
reproduces the printed clustered SE to floating-point precision, i.e. the
``exact'' entries in rows (3)--(4).

\subsection{Simulation of regression specifications}\label{app:regression-simulation}

This simulation evaluates the specifications of
Section~\ref{sec:regression-implementation} against a known truth. Data
follow the saturated model in Equation~\eqref{eq:saturated-model},
for $n=800$ units split between the two sequences:

\begin{equation*}
    Y_{it}(d_1,d_2)=\pi_{it}+d_t\tau_i+(t-1)\delta_{i,d_1}.
\end{equation*}

The unit treatment effect is
$\tau_i\sim\mathcal N(0.50,0.40^2)$, matching the true effect used
throughout the paper's worked examples. The carryover effects
$\delta_{i,0}$ and $\delta_{i,1}$ are drawn independently across units and
sequences from $\mathcal N(\cdot,0.20^2)$: both have mean $0$ except when a
carryover gap is present, in which case $\mathbb E[\delta_{i,0}]=0.30$ while
$\mathbb E[\delta_{i,1}]$ remains $0$. The baseline outcome is
$\pi_{it}=\alpha_i+\gamma_2 I(t=2)+\varepsilon_{it}$, where
$\alpha_i\sim\mathcal N(0,0.90^2)$ is a unit-level random intercept shared by
both periods, inducing the positive within-unit correlation discussed in
Section~\ref{sec:Within-Subject-Design-in}.  $\varepsilon_{it}\sim\mathcal
N(0,0.70^2)$ is independent idiosyncratic noise drawn separately for each
period, and $\gamma_2=0.40$ when a period effect is present and $0$
otherwise. 

The simulation crosses three such factors: whether the average
carryover gap is present, whether allocation is equal ($n_1=n_2=400$) or
$2{:}1$ ($n_1=534$, $n_2=266$), and whether the common period effect is
present, for $2\times2\times2=8$ cells at $5{,}000$ replications each. Each
cell reports bias for the average treatment effect, the ratio of the mean
reported standard error to the true sampling standard deviation, and
coverage of the nominal 95\% interval. 

We have two core findings.  First, the interacted specification with unit-clustered standard errors
reproduces the design-based estimator and its standard error \emph{exactly} in
every cell, as \ref{app:saturated} proves. Table~\ref{tab:regression-specs-recovery}
restricts to the four cells without a carryover gap, and reports point estimates and standard
errors for all six specifications against the design-based benchmark. The treatment-only
specification does not recover the design-based estimator: under equal
allocation with no period effect the two coincide, but with $2{:}1$ allocation
and a period effect, the treatment-only specification carries a bias of
$-0.13$ against a true effect of $0.50$. The period-interacted specification is
unaffected because a common period effect is included in the model.

\begin{table}[!h]
\begin{centering}
\input{tables/s7_regression_specs_recovery}
\par\end{centering}
\caption{Recovery of the Design-based Estimator, No Carryover Gap, Averaged Across 5{,}000 Iterations}\label{tab:regression-specs-recovery}
\end{table}

The consequence for inference is the incorrect coverage. Table~\ref{tab:regression-specs-coverage}
reports bias and coverage of the nominal 95\% interval across all eight
cells, including the carryover-gap cells, where pooling is biased for every
specification alike due to an identification failure. Coverage for the treatment-only
specification falls to $7.7\%$ in the $2{:}1$-allocation, period-effect,
no-carryover-gap cell, the same cell behind the bias figure above, while
design-based and the period-interacted specification coverage track each other throughout.

\begin{table}[!h]
\begin{centering}
\input{tables/s7_regression_specs_coverage}
\par\end{centering}
\caption{Bias and Coverage of the Design-based, Interacted, and Treatment-only Specifications, All Eight Cells}\label{tab:regression-specs-coverage}
\end{table}

Second, OLS and random effects give numerically identical point estimates and
identical unit-clustered CR0 standard errors, to machine precision, in every
cell and in both parameterizations. Where they differ is the \emph{unclustered,
model-based} standard error that panel software reports by default. When no
cluster-robust option is requested, OLS's default assumes i.i.d.\ errors,
while random effects' default is built from the within-unit correlation
$\hat\sigma_u^2$ it estimates, so the two rest on different variance formulas
even though they share the same point estimate. The choice that matters is therefore the parameterization and the clustering,
not OLS versus random effects.

%% file: tables/s7_regression_specs_recovery.tex
\begin{tabular}{lcccc}
\toprule
 & ATE & ATE/design & SE & SE/design \\
\midrule
\multicolumn{5}{l}{\emph{Panel A: equal allocation ($n_1=n_2=400$)}} \\
\multicolumn{5}{l}{\quad No period effect} \\
Design-based & 0.499 & exact & 0.038 & exact \\
Paired $t$-test & 0.499 & 0.00\% & 0.038 & +0.13\% \\
Treatment-only OLS & 0.499 & 0.00\% & 0.038 & +0.06\% \\
Treatment-only RE & 0.499 & 0.00\% & 0.038 & +0.06\% \\
Interacted OLS & 0.499 & exact & 0.038 & exact \\
Interacted RE & 0.499 & exact & 0.038 & exact \\
\multicolumn{5}{l}{\quad Period effect} \\
Design-based & 0.500 & exact & 0.038 & exact \\
Paired $t$-test & 0.500 & 0.00\% & 0.041 & +6.72\% \\
Treatment-only OLS & 0.500 & 0.00\% & 0.041 & +6.65\% \\
Treatment-only RE & 0.500 & 0.00\% & 0.041 & +6.65\% \\
Interacted OLS & 0.500 & exact & 0.038 & exact \\
Interacted RE & 0.500 & exact & 0.038 & exact \\
\addlinespace
\multicolumn{5}{l}{\emph{Panel B: $2{:}1$ allocation ($n_1=534$, $n_2=266$)}} \\
\multicolumn{5}{l}{\quad No period effect} \\
Design-based & 0.500 & exact & 0.041 & exact \\
Paired $t$-test & 0.500 & -0.06\% & 0.038 & -5.64\% \\
Treatment-only OLS & 0.500 & -0.06\% & 0.038 & -5.69\% \\
Treatment-only RE & 0.500 & -0.06\% & 0.038 & -5.69\% \\
Interacted OLS & 0.500 & exact & 0.041 & exact \\
Interacted RE & 0.500 & exact & 0.041 & exact \\
\multicolumn{5}{l}{\quad Period effect} \\
Design-based & 0.501 & exact & 0.041 & exact \\
Paired $t$-test & 0.367 & -26.76\% & 0.041 & -0.14\% \\
Treatment-only OLS & 0.367 & -26.76\% & 0.041 & -0.20\% \\
Treatment-only RE & 0.367 & -26.76\% & 0.041 & -0.20\% \\
Interacted OLS & 0.501 & exact & 0.041 & exact \\
Interacted RE & 0.501 & exact & 0.041 & exact \\
\bottomrule
\end{tabular}
\par\smallskip
\begin{minipage}{\linewidth}
\footnotesize\raggedright
\textit{Notes:} No between-sequence average carryover gap in any cell
($\delta=0$); true treatment effect $\tau=0.50$; 5{,}000 replications per
cell. ATE and SE are averaged across replications: ATE is the mean point
estimate and SE is the mean reported standard error. ATE/design and
SE/design are the percentage difference from the design-based point
estimate and standard error in the same cell.
\emph{Exact} denotes equality to machine precision, as
\ref{app:saturated} proves for the interacted specifications.
\end{minipage}

%% file: tables/s7_regression_specs_coverage.tex
\footnotesize
\begin{tabular}{lcc cc cc cc}
\toprule
 & & & \multicolumn{2}{c}{Design-based} & \multicolumn{2}{c}{Interacted} & \multicolumn{2}{c}{Treatment-only} \\
\cmidrule(lr){4-5} \cmidrule(lr){6-7} \cmidrule(lr){8-9}
Allocation & Carryover Gap & Period Effect & Bias & Coverage & Bias & Coverage & Bias & Coverage \\
\midrule
Equal & No & No & $-0.001$ & 95.0\% & $-0.001$ & 95.0\% & $-0.001$ & 95.0\% \\
Equal & No & Yes & $0.000$ & 95.2\% & $0.000$ & 95.2\% & $0.000$ & 96.4\% \\
Equal & Yes & No & $-0.151$ & 2.3\% & $-0.151$ & 2.3\% & $-0.151$ & 2.4\% \\
Equal & Yes & Yes & $-0.150$ & 2.5\% & $-0.150$ & 2.5\% & $-0.150$ & 4.2\% \\
$2{:}1$ & No & No & $0.000$ & 94.9\% & $0.000$ & 94.9\% & $0.000$ & 94.7\% \\
$2{:}1$ & No & Yes & $0.001$ & 94.6\% & $0.001$ & 94.6\% & $-0.133$ & 7.7\% \\
$2{:}1$ & Yes & No & $-0.150$ & 4.0\% & $-0.150$ & 4.0\% & $-0.200$ & 0.0\% \\
$2{:}1$ & Yes & Yes & $-0.150$ & 4.4\% & $-0.150$ & 4.4\% & $-0.334$ & 0.0\% \\
\bottomrule
\end{tabular}
\par\smallskip
\begin{minipage}{\linewidth}
\footnotesize\raggedright
\textit{Notes:} True treatment effect $\tau=0.50$; 5{,}000 replications
per cell. Bias is the mean point estimate across replications minus
$\tau$. Coverage is the fraction of the 5{,}000 replications in each
cell whose nominal 95\% interval contained $\tau$. \emph{Carryover Gap}
indicates a nonzero between-sequence average carryover gap ($\delta\ne0$);
\emph{Period Effect} indicates a common period shift.
\end{minipage}

%% file: appendix_undercoverage_simulation.tex
\subsection{Simulations of two-step undercoverage and replication of Freeman}\label{app:undercoverage-simulation}

This appendix evaluates the asymptotic coverage formula of
Equations~\eqref{eq:grizzle-coverage}--\eqref{eq:coverage-terms} (Proposition~\ref{prop:undercoverage})
numerically, along five dimensions. We first validate it against
\textcite{freeman1989performance}'s published special case, then trace the
coverage curve's shape away from $\delta=0$ and calibrate it to reductions
recovered from the literature review (\ref{app:litreview-parameters}). The
remaining two checks address the gap between the asymptotic derivation and
applied practice: variances estimated from data rather than known, and
outcomes that are not normally distributed.

\paragraph*{External validation against \textcite{freeman1989performance}.}
Freeman derives the coverage of the two-step interval in an additive normal
model with a patient random effect, a common error variance, equal sequence
sizes, and a baseline measurement, and evaluates it at $\alpha_{1}=0.10$ and
$\alpha=0.05$.  That model class is exactly
the $r=1$ slice of the present framework (Equation~\eqref{eq:omega-zeta}),
with his within-patient intraclass correlation equal to our $g$. Evaluating
Equation~\ref{eq:grizzle-coverage} on that slice reproduces every value he
reports.

\begin{table}[!h]
\centering
\caption{Reproduction of \gentextcite{freeman1989performance} published values
from Equation~\ref{eq:grizzle-coverage}. His $\lambda$ is our $-\delta$.}
\label{tab:freeman-reproduction}
\begin{tabular}{lrrl}
\toprule
Freeman's (1989) quantity & \textcite{freeman1989performance} & This paper & Our $(r,g,\eta)$ \\
\midrule
Coverage at $\lambda=0$, with baseline & $0.917$ & $0.9180$ & $(1,\ \tfrac12,\ 0)$ \\
Minimum coverage, with baseline & $0.56$ & $0.5629$ & $(1,\ \tfrac12,\ 1.53)$ \\
Actual level at $\lambda=0$, no baseline & $0.085$ & $0.0848$ & $(1,\ 0.6,\ 0)$ \\
Maximum actual level, no baseline & ``just over $0.50$'' & $0.5026$ & $(1,\ 0.6,\ 1.42)$ \\
\midrule
\multicolumn{4}{@{}p{355.457pt}@{}}{\footnotesize \textit{Notes}: $\eta=\sqrt{n}\,\delta/\sigma_{\delta}$ is the standardized
carryover gap used in the rightmost column.} \\
\bottomrule
\end{tabular}
\end{table}

\paragraph*{Behavior away from $\delta=0$.}
Figure~\ref{fig:coverage-local} plots Equation~\ref{eq:grizzle-coverage}
against $\eta=\sqrt{n}\,\delta/\sigma_{\delta}$. Coverage falls from
its value at $\delta=0$ to a minimum and then returns to nominal as the
preliminary test gains power, per Remark~\ref{rem:fixed-delta}. The larger
$g$, and hence the larger the efficiency gain that motivates pooling, the
deeper the trough and the smaller the carryover gap at which it occurs.

\begin{table}[!h]
\centering
\begin{tabular}{lrrr}
\toprule
$g$ (at $r=1$) & Coverage at $\delta=0$ & Minimum coverage & at $|\eta|$ \\
\midrule
$0.0$ & $0.9356$ & $0.7250$ & $1.97$ \\
$0.3$ & $0.9302$ & $0.5868$ & $1.80$ \\
$0.5$ & $0.9259$ & $0.4678$ & $1.62$ \\
$0.8$ & $0.9173$ & $0.2373$ & $1.17$ \\
\midrule
\multicolumn{4}{@{}p{267.704pt}@{}}{\footnotesize \textit{Notes}: The
rightmost column reports the value of $|\eta|$ at which each row's
minimum coverage is attained.} \\
\bottomrule
\end{tabular}
\caption{Coverage of the nominal 95\% two-step interval at equal period
variances.}
\label{tab:trough}
\end{table}

\begin{figure}[!h]
\centering
\includegraphics[width=0.78\textwidth]{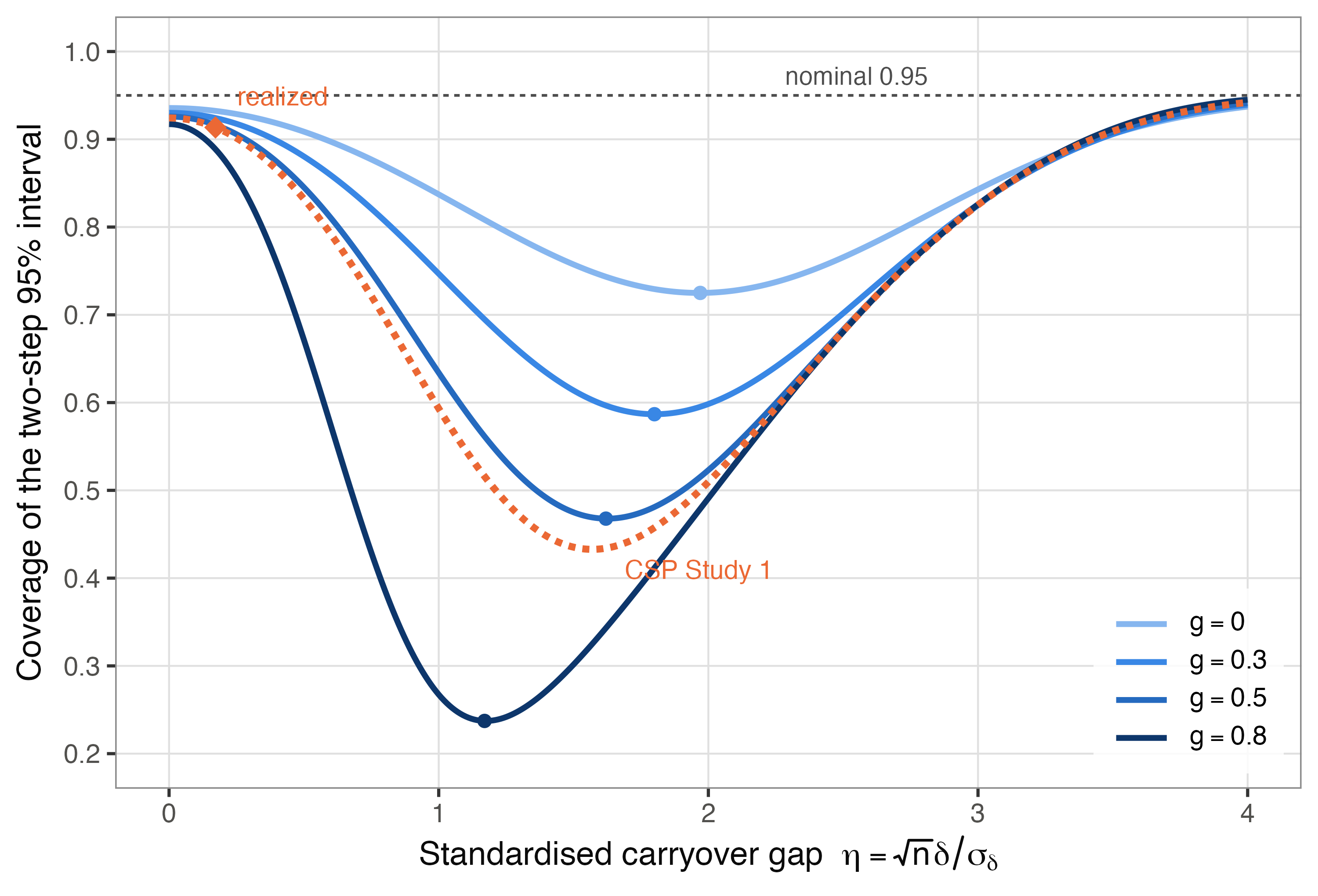}
\caption{Coverage of the nominal 95\% two-step interval as a function of the
standardized carryover gap, at equal period variances, with the calibration to
Study 1 of \textcite{clifford2021increasing} shown as a dashed line. Points mark
each curve's minimum. The diamond marks Study 1's own realized
(plug-in) standardized carryover gap, $|t_{\delta}|$, shown distinctly from
each curve's synthetic worst-case minimum.}
\label{fig:coverage-local}
\end{figure}

\paragraph*{Calibrations from the literature review.}
We apply the same coverage
formula, Equations~\eqref{eq:grizzle-coverage} and~\eqref{eq:coverage-terms}, to all 30 reductions recovered from the literature review with
complete variance components. Of these, 27 are Group 1 + Group 2 rows drawn
from seven papers (\ref{app:litreview-parameters}) and are used to
characterize the empirically relevant region. The remaining 3 are retained in
Table~\ref{tab:litreview-recovered-parameters}, but not included due to smaller observation numbers (e.g. $n<20$ in paper 10).
Figure~\ref{fig:literature-two-step-coverage} plots the 27 primary curves on a
common standardized carryover axis: at $\delta=0$ their coverage ranges from
$0.920$ to $0.936$, and over $|\eta|\leq4$ their minimum coverage
ranges from $0.303$ to $0.722$, with a median minimum of $0.628$. The two-step
undercoverage seen in the CSP and JOT calibrations is therefore not peculiar to
either application, but widly exists. 

\begin{figure}[!h]
\centering
\includegraphics[width=0.78\textwidth]{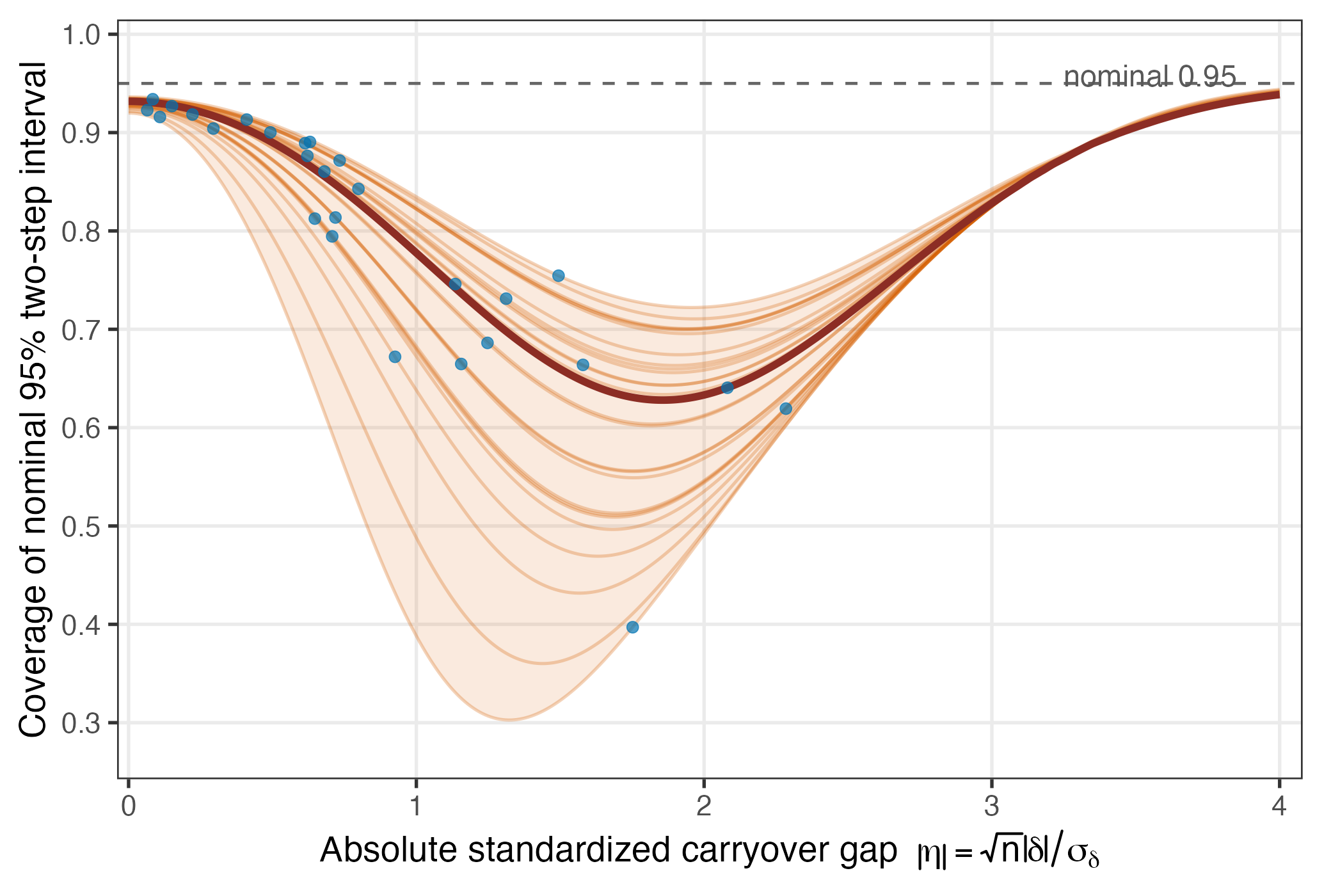}
\caption{Coverage of the nominal 95\% two-step interval for the 27 Group 1 + Group 2 outcome-level reductions recovered from seven papers in the
literature review. Each thin curve is one reduction, the shaded band is their
pointwise range, and the heavy curve is their pointwise median. The
3 excluded
reductions reported in Table~\ref{tab:litreview-recovered-parameters} are not
used to characterize the empirical coverage region.} A point on each
curve marks that reduction's own realized standardized carryover gap,
$|\hat\eta| = |\hat\delta|/\text{se}(\hat\delta)$, the empirical
plug-in analog of $|\eta|$ evaluated from its estimated rather than
population variance components.
\label{fig:literature-two-step-coverage}
\end{figure}

\paragraph*{Estimated rather than known variances.}
The formal results use the population quantities
$\sigma_{\delta},\sigma_{1},\sigma_{co}$; researchers substitute sample
estimates. Figure~\ref{fig:coverage-finite-n} evaluates the consequence by
Monte Carlo at $r=1$, $g=0.6$, with $400{,}000$ replications at each of
$n=100$ and $n=500$. Substituting estimated variances -- the plug-in rule researchers actually
use -- moves coverage only slightly: at $\delta=0$ it falls below the oracle
curve, where the population variance is used, by $0.007$ at $n=100$ and $0.0015$ at $n=500$. The lower panel of
Figure~\ref{fig:coverage-finite-n} plots the plug-in rule's own gap to the analytic curve at each $\eta$: the discrepancy stays small,
within $\pm0.011$ at $n=100$ and $\pm0.003$ at $n=500$, across the full
range of $\eta$.

\paragraph*{Non-normal outcomes.}
Nothing in the derivation assumes a distributional form for the outcome. We only
need the estimators to obey a central limit theorem. 

We simulate three outcome families of
very different shape: normal, a three-point ordinal scale calibrated
to the category shares of \gentextcite{clifford2021increasing} welfare item,
and $t$ with three degrees of freedom, each evaluated against the
analytic curve at its own $(r,g)$. All three reproduce the same U-shaped
coverage curve the formula predicts. Figures~\ref{fig:coverage-nonnormal-normal}--\ref{fig:coverage-nonnormal-t3}
show, for each family, both the simulated coverage using
estimated variances and the analytic prediction. 

\begin{figure}[!h]
\centering
\includegraphics[width=0.78\textwidth]{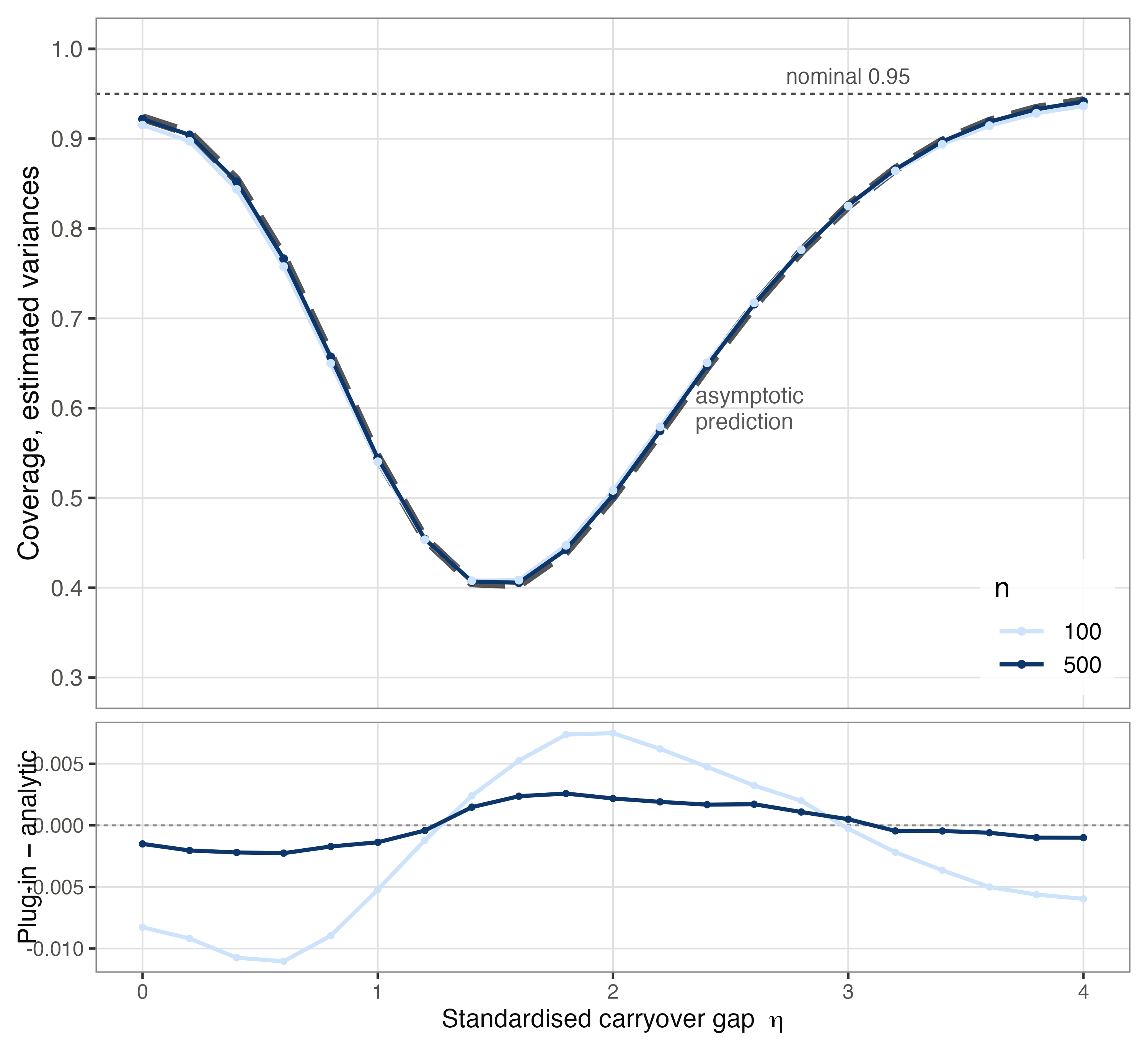}
\caption{Top: coverage of the two-step interval
computed with estimated variances, at $r=1$, $g=0.6$, against the asymptotic
prediction (dashed line), for $n=100$ and $n=500$. Bottom: the plug-in minus
analytic coverage gap at each $\eta$, on the same $x$-axis, showing directly
how small the discrepancy is (within about $\pm0.011$ at $n=100$ and
$\pm0.003$ at $n=500$). $400{,}000$ replications per point.}
\label{fig:coverage-finite-n}
\end{figure}

\begin{figure}[!h]
\centering
\includegraphics[width=0.62\textwidth]{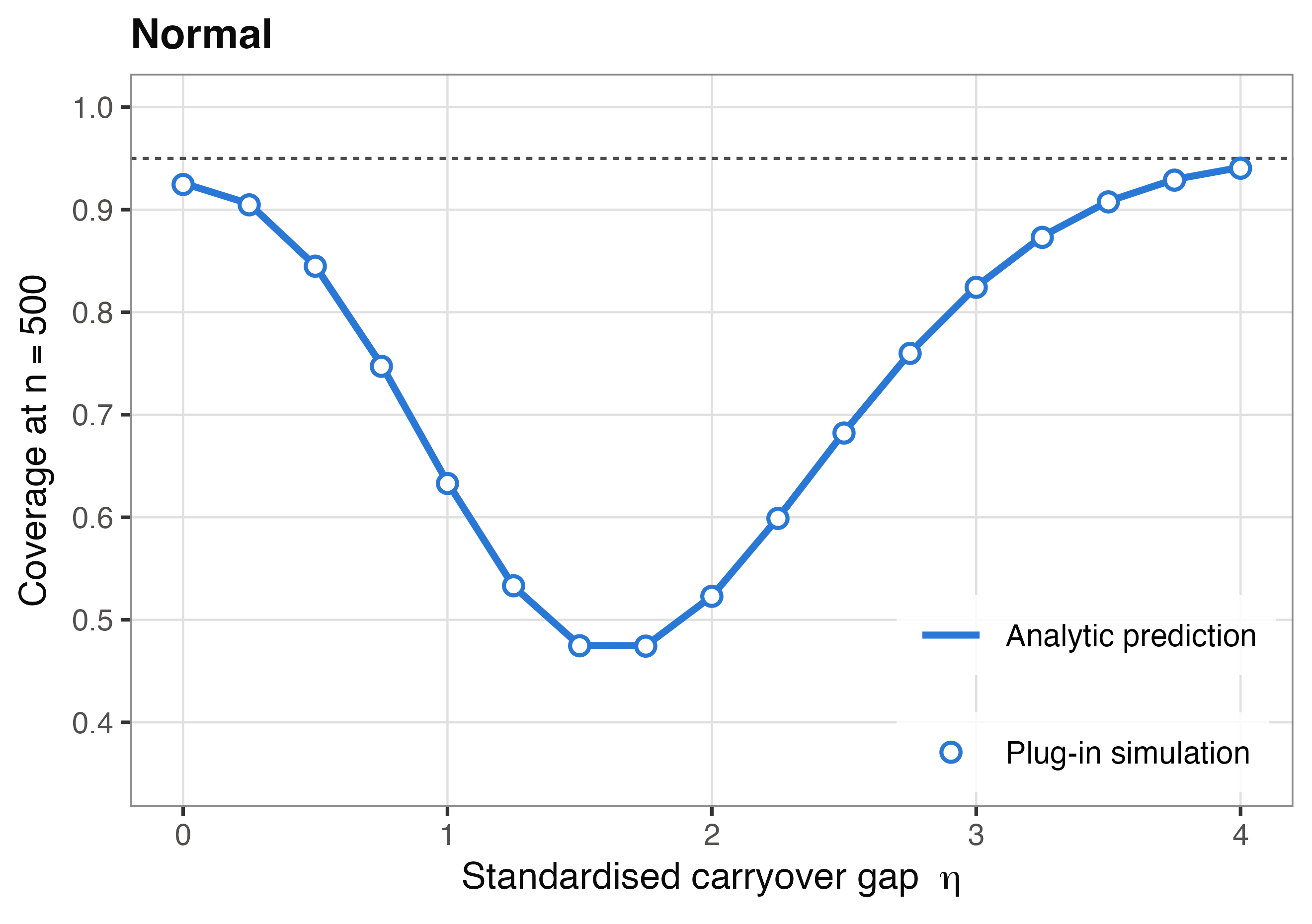}
\caption{Simulated coverage (plug-in simulation, points) against the analytic
prediction (line) at $n=500$, for a normal outcome. The largest discrepancy
over the full range of $\eta$ is $0.002$.}
\label{fig:coverage-nonnormal-normal}
\end{figure}

\begin{figure}[!h]
\centering
\includegraphics[width=0.62\textwidth]{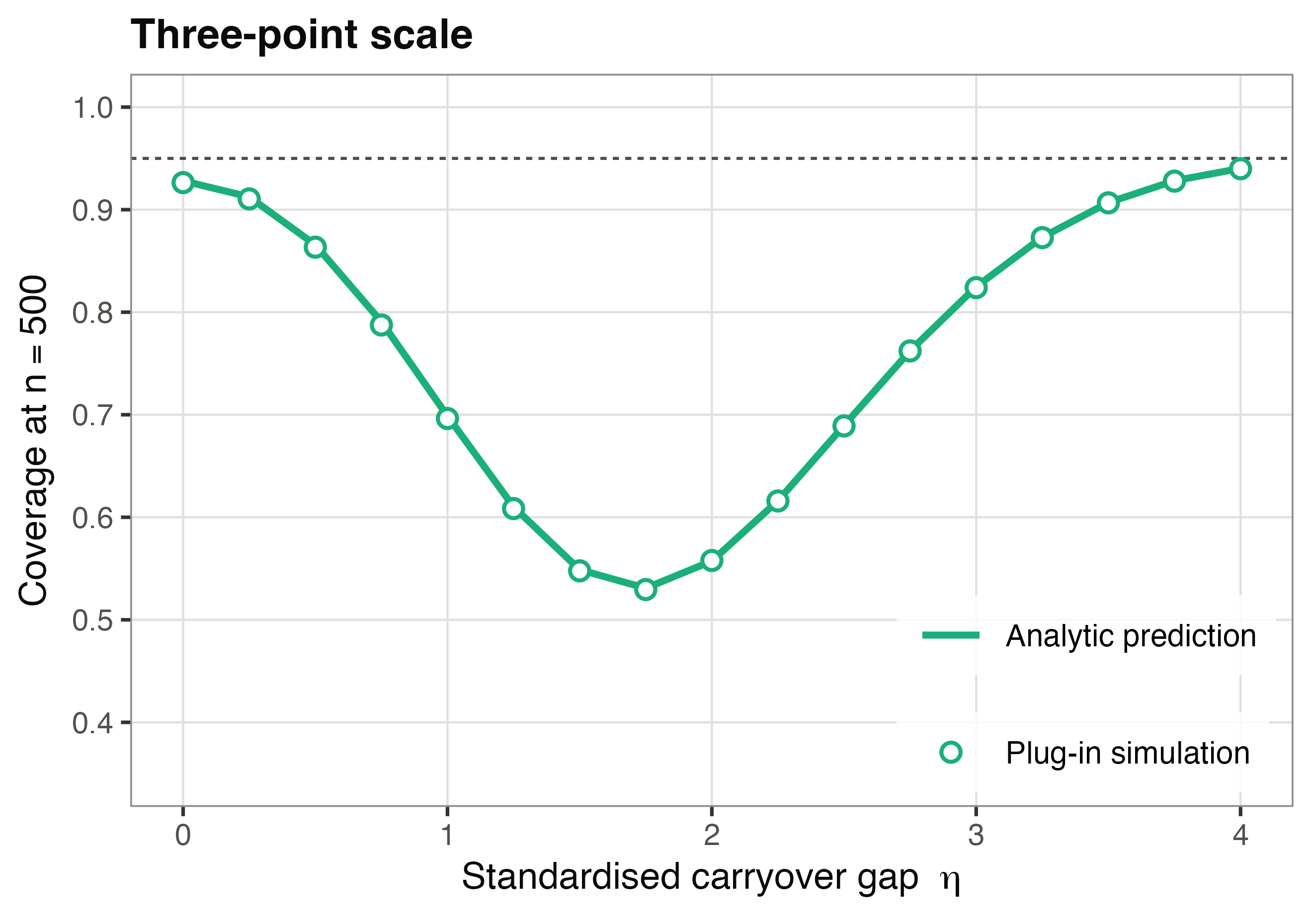}
\caption{Simulated coverage against the analytic prediction at $n=500$, for a
three-point ordinal outcome calibrated to the category shares of
\gentextcite{clifford2021increasing} welfare item. The largest discrepancy is
$0.003$.}
\label{fig:coverage-nonnormal-likert3}
\end{figure}

\begin{figure}[!h]
\centering
\includegraphics[width=0.62\textwidth]{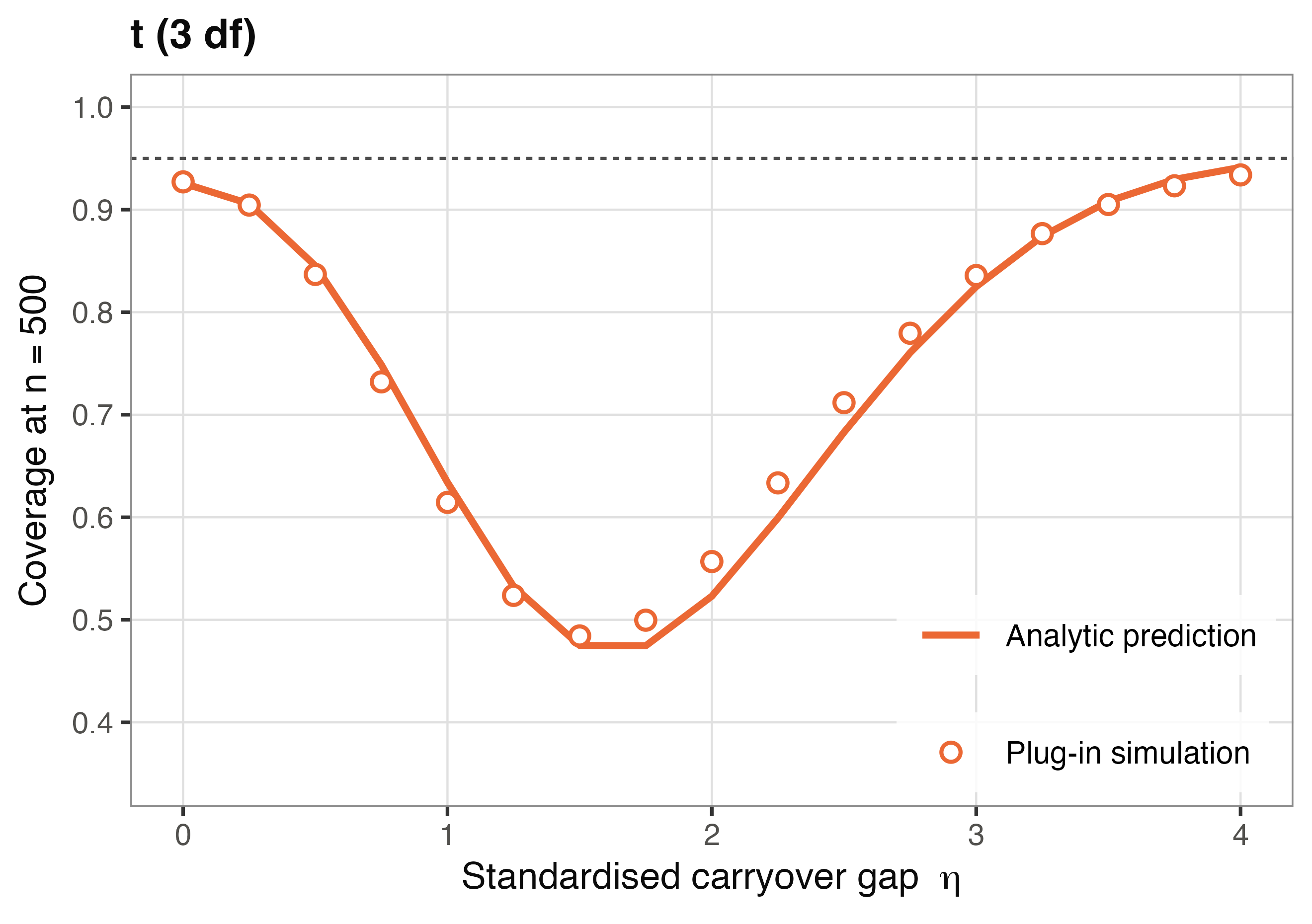}
\caption{Simulated coverage against the analytic prediction at $n=500$, for an
outcome distributed $t$ with three degrees of freedom. The largest discrepancy
is $0.034$, the worst of the three families tested.}
\label{fig:coverage-nonnormal-t3}
\end{figure}

%% file: appendix_litreview_parameters.tex
\section{Parameters recovered from the literature review}\label{app:litreview-parameters}

The sensitivity-analysis review covers 33 published papers that cite
\textcite{clifford2021increasing} and that \textcite{JORDAN_OLLERENSHAW_TREXLER_2026}
classify as adopting a within-subject design. The purpose is to locate
empirical numbers usable for this paper's within-subject sensitivity
analysis. Table~\ref{tab:litreview-paper-audit}
records one headline row for every paper, including designs, and reasons for which the
strict two-sequence estimands are not identified. We categorize the 33 papers into 4 groups. 

Group 1 contains exact two-period crossover designs, in which respondents are
randomized to treatment--control $(1,0)$ or control--treatment $(0,1)$ and the
same outcome is measured after each period. In addition, designs with additional
intervening items from which two observations satisfying this structure can be
isolated are also included here. 

Group 2 contains multiperiod or multicondition designs that we reduce via pairwise comparisons. We pair the first period with each later period. A pair contributes only when both treatment sequences are present, i.e. some respondents received treatment in Period 1 and control in the later period (sequence $(1,0)$) while others received the reverse (sequence $(0,1)$), so we can form the same two-sequence contrast used throughout the paper. We average with equal weight over all such pairs. We call these reductions history-marginal because a respondent's outcome in the later period is averaged over whatever treatments they received in between.

Group 3 contains within-subject designs whose Period 1 exposure is fixed across respondents. Everyone receives the same (or no) stimulus in Period 1, and only in Period 2 are they assigned to different treatment conditions. Group 3, further elaborated later, is not subject to the differential-carryover concern that motivates our analysis, since every respondent shares the same Period 1 exposure; it remains a within-subject design, but it lacks the two-sequence structure that our analysis addresses. 

Group 4 contains the remaining non-reducible or unresolved designs.
These fall into three kinds: (a) designs that
would otherwise support a Group 1 or Group 2 reduction but lack the
replication data needed for it, most often a covariate recording the
order in which treatments were assigned; (b) a single paper whose pairwise
reduction executes mechanically but whose paired items are answered in
one sitting rather than at two temporally separated occasions, so they
cannot support a two-sequence estimate despite the data being available;
and (c) a small number of papers that
\textcite{JORDAN_OLLERENSHAW_TREXLER_2026}'s own table misclassifies
as within-subject, which are simply out of scope for this review rather
than non-reducible on structural terms.

Table~\ref{tab:litreview-paper-audit}
identifies which of
the four specific group it belongs to, with further reasons defined in the table's own notes.
Separately, a dash anywhere in the $N$, $t_{rep}$, $t_{co}$, or
$t_{\delta}$ columns means that quantity could not be recovered for that
paper's headline row.

The three $t$-columns serve different purposes: $t_{rep}$ is the $t$-statistic
for the paper's headline result as originally reported, whatever estimand
that paper used, while $t_{co}$ and $t_{\delta}$ are the $t$-statistics this
paper computes under its own two-sequence framework, for the pooled
estimator $\hat\mu_{co}$ and the carryover gap $\hat\delta$ respectively.
The distinction matters because a paper's headline result can be
reproducible ($t_{rep}$ recovers the original finding) even when its design
does not identify the two-sequence quantities our formal analysis
requires, in which case $t_{co}$ and $t_{\delta}$ are entered as dashes.

Table~\ref{tab:litreview-paper-audit}
reports one headline row per paper, chosen before inspecting any harmonized
estimate: the result the paper or designates primary.
Absent that, the first main-text confirmatory result tied to the
within-subject design. For papers with multiple experiments, the
author-declared main experiment, else the first main-text confirmatory
experiment. 

Table~\ref{tab:litreview-recovered-parameters} then reports every
study-outcome for which the sequence sizes and all three variance components
are recoverable. There are 30 such reductions. For each, it records the pooled
and carryover test statistics, $\sigma_{1}^{2}$, $\sigma_{2}^{2}$,
$\sigma_{12}$, and their normalized versions
$r=\sigma_{2}^{2}/\sigma_{1}^{2}$ and
$g=\sigma_{12}/\sigma_{1}^{2}$. The ``Included'' column identifies the 27
Group 1 + Group 2 rows used as circles in Figure~\ref{fig:r-g-map}. Of the
remaining 3 rows, all have fewer than twenty observations in at least one
sequence (paper 10). Their parameters are retained for transparency but are
shown as crosses in Figure~\ref{fig:r-g-map} rather than used to
characterize the empirically relevant region.

\begin{landscape}
\captionsetup[longtable]{width=0.78\linewidth}
\begingroup
\footnotesize
\setlength{\tabcolsep}{3pt}
\setlength{\LTleft}{0pt}
\setlength{\LTright}{0pt}
\input{tables/litreview_paper_audit}
\endgroup
\begingroup
\scriptsize
\setlength{\tabcolsep}{2.2pt}
\setlength{\LTleft}{0pt}
\setlength{\LTright}{0pt}
\input{tables/litreview_recovered_parameters}
\endgroup
\end{landscape}

\paragraph*{Location in the undercoverage parameter space.}
Figure~\ref{fig:r-g-map} evaluates the coverage deficit at $\delta=0$, from
Equations~\eqref{eq:grizzle-coverage} and~\eqref{eq:coverage-terms}, minus
nominal, on a $241\times241$ grid of $r\in[0.05,6.00]$ and
$g\in[-2.50,2.50]$; restricting to the feasible region $g^{2}\le r$ leaves
$38{,}049$ grid points. The color passes through neutral
exactly on the boundary $g=(r-3)/2$ from
Proposition~\ref{prop:undercoverage}, and its sign agrees with the proposition
at every one of these feasible points.  Every
over-coverage cell has $r>1$, and every over-coverage cell with
$\sigma_{12}\ge0$ has $r>3$.

The \LitReviewZetaFaithfulCount{} Group 1 + Group 2 study-outcomes come from
seven papers and have $r\in[\LitReviewOmegaFaithfulMin{},\LitReviewOmegaFaithfulMax{}]$
and $g\in[\LitReviewZetaFaithfulMin{},\LitReviewZetaFaithfulMax{}]$
(nonnegative in \LitReviewZetaFaithfulNonnegCount{} of
\LitReviewZetaFaithfulCount{}). All lie in the
undercoverage region, and
their coverage at $\delta=0$ ranges from $0.920$ to $0.936$ against a nominal
$0.95$. Empirically, then, we also find that $r$ is close to one, and the slice on which
\textcite{freeman1989performance} worked is the empirically relevant one.
For comparison, Section~\ref{sec:testing-carryover}'s
covariance-plausibility claim draws on the full \LitReviewZetaCount{}
recoverable study-outcomes, for which $g$ ranges more widely,
\LitReviewZetaMin{} to \LitReviewZetaMax{} (nonnegative in
\LitReviewZetaNonnegCount{} of \LitReviewZetaCount{}). The wider range
reflects the 3 rows excluded from the Included column of
Table~\ref{tab:litreview-recovered-parameters}, whose small samples
make them noisier evidence of same-construct
correlation than the Group 1 + Group 2 subset, even though the broader set
still supports the same directional conclusion.

\begin{figure}[tbp]
\centering
\includegraphics[width=0.78\textwidth]{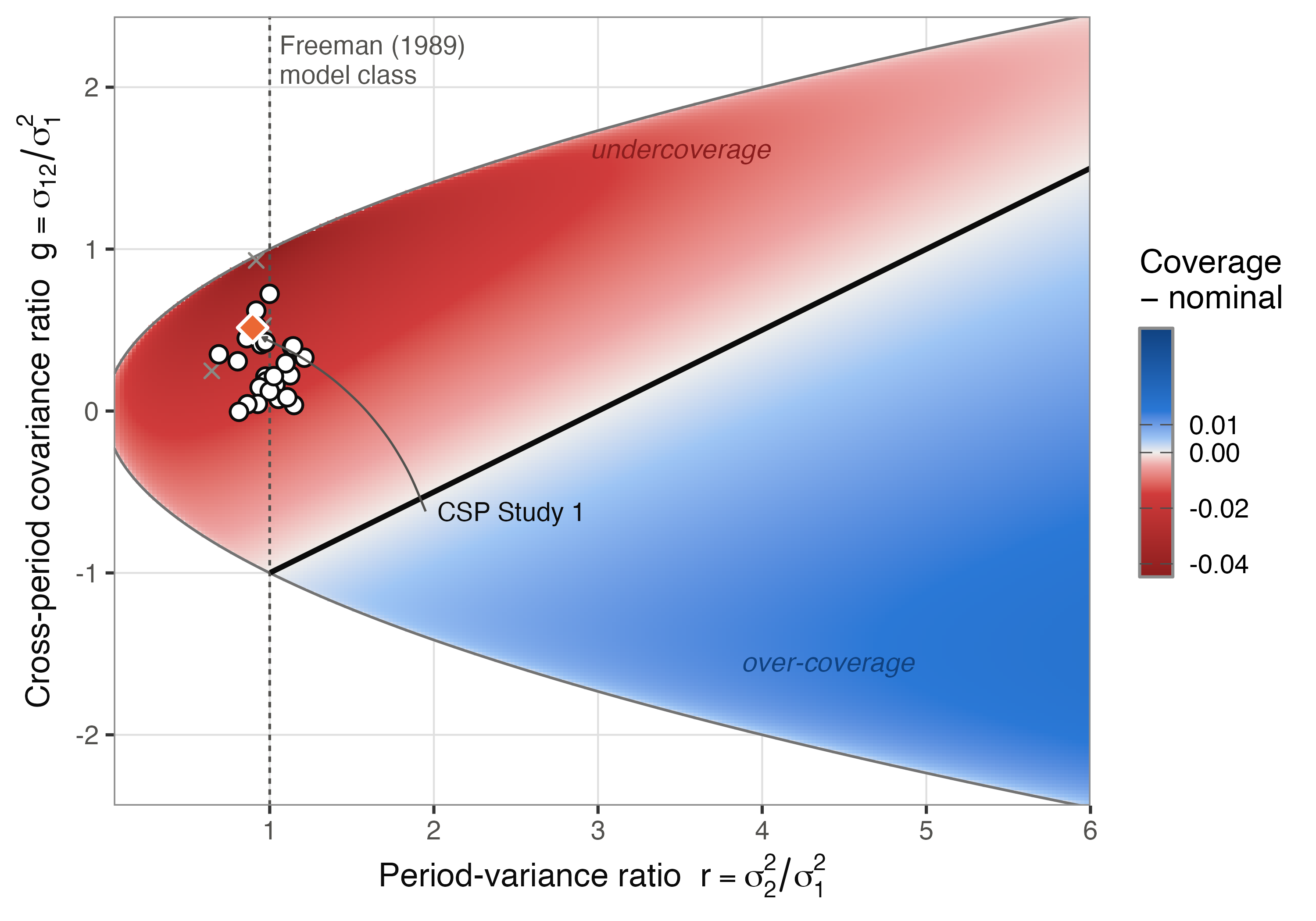}
\caption{Coverage of the nominal 95\% two-step interval at $\delta=0$, minus
nominal, over the feasible region. The heavy line is the efficiency boundary
$\tfrac{1}{4}\sigma_{co}^{2}=\sigma_{1}^{2}$; the thin curves are the
feasibility envelope $g=\pm\sqrt{r}$. Circles are the 27
Group 1 + Group 2 study-outcomes; crosses are the 3 excluded reductions retained
for transparency.}
\label{fig:r-g-map}
\end{figure}

\subsection{From the literature review to Table~\ref{tab:litreview-sensitivity-summary}}\label{app:litreview-crosswalk}
%{\color{red}\small[TK-REVIEW]}

Table~\ref{tab:litreview-recovered-parameters} and
Table~\ref{tab:litreview-sensitivity-summary} report overlapping but
distinct subsets of the literature review, because they answer different
questions rather than one being a further-filtered version of the other.
Of the 33 reviewed papers, only 8 have the two-sequence structure needed
to compute $\hat\delta$ and its variance components at all. Since some of
these 8 measured several outcomes or paired several periods, they together
yield the 30 study-outcomes in Table~\ref{tab:litreview-recovered-parameters}.
That table reports all 30 because its purpose is to characterize how
correlated repeated measurements typically are (Section~\ref{sec:testing-carryover}),
a question that does not require the paper's own pooled estimate to have
been significant. Table~\ref{tab:litreview-sensitivity-summary} instead
asks a narrower question: whether a plausible carryover gap could overturn
an originally significant result, which is only meaningful for results
that were significant to begin with. Restricting the 30 to those with
$|t_{co}|>z_{.975}$ leaves 15; we then add 10 further comparisons
reconstructed independently from exact two-sequence crossovers in
\textcite{clifford2021increasing}'s Study 1 and
\textcite{JORDAN_OLLERENSHAW_TREXLER_2026}'s Studies 4--6, chosen because
their design and data leave no ambiguity about the two sequences. The
precise chain is as follows.

Section~\ref{sec:Sensitivity-Analysis}'s calibration draws on the same 30
study-outcomes in Table~\ref{tab:litreview-recovered-parameters}, filtered
independently of that table's Included column: it
retains rows whose pooled estimate is itself significant,
$|t_{co}|>z_{.975}$. This yields 15 significant comparisons from 6 of the 8
papers represented in Table~\ref{tab:litreview-recovered-parameters} (paper
IDs 04, 08, 10, 16, 22, and 47 in Table~\ref{tab:litreview-paper-audit}). The
remaining two papers (09 and 50) contribute no significant row. We
supplement these 15 with 10 comparisons reconstructed independently from
\textcite{clifford2021increasing}'s Study 1 and
\textcite{JORDAN_OLLERENSHAW_TREXLER_2026}'s Studies 4--6 across their
AmeriSpeak, Prolific, and Lucid samples, exact two-sequence crossovers
external to the 33-paper citation review rather than additional papers
from it. All 10 are significant. Table~\ref{tab:litreview-sensitivity-summary}
reports the resulting 25 comparisons across 8 papers, equal-weighted within
paper. One paper contributes two significant outcomes for the same
treatment-by-period pair, giving 24 distinct pairs.

%% file: tables/litreview_paper_audit.tex
\begin{longtable}{@{}r p{2.6cm} p{4.6cm} p{4.0cm} r r r r p{4.2cm}@{}}
\caption[]{Paper-level audit of recovered quantities.}\label{tab:litreview-paper-audit}\\
\toprule
ID & Author (Year) & Paper & Design class & $N$ & $t_{rep}$ & $t_{co}$ & $t_{\delta}$ & Source \\
\midrule
\endfirsthead
\multicolumn{9}{l}{\tablename\ \thetable\ (continued)}\\
\toprule
ID & Author (Year) & Paper & Design class & $N$ & $t_{rep}$ & $t_{co}$ & $t_{\delta}$ & Source \\
\midrule
\endhead
\midrule \multicolumn{9}{r}{Continued on next page}\\ \endfoot
\midrule
\multicolumn{9}{@{}p{21.0cm}@{}}{\footnotesize \textit{Notes}: The first line of the design-class column gives the four-group manuscript classification. For Group 4, the second line records one of four reasons: ``Group 1 shape, no data'' -- a two-condition, randomized-order design (Group 1's pattern) whose presentation-order or replication data is unavailable; ``Group 2 shape, no data'' -- a multiperiod or multicondition design (Group 2's pattern) blocked the same way; ``JOT misclassification'' -- Jordan, Ollerenshaw, and Trexler's (2026) own source table lists the paper as within-subject, but it is not, on any contrast relevant to this review, and it is therefore out of scope rather than non-reducible on structural terms; or ``Single-occasion, multi-item design'' -- a pairwise reduction executes mechanically (both treatment orders exist across items), but the paired items are answered in one sitting with treatment fixed in advance, so there are no two temporally separated measurement occasions for a two-sequence estimate to describe. A paper's row can show dashes throughout even though a different study-outcome from the same paper contributes real recovered numbers to Table~\ref{tab:litreview-recovered-parameters}: the headline row shown here is fixed by the paper's selection rule before checking what is recoverable, precisely so recoverability cannot influence which row is chosen as headline, and every other eligible outcome for that paper is retained in the underlying detail table regardless. Papers 47 and 50 illustrate this pattern.}\\
\bottomrule \endlastfoot
02 & \textcite{annaka2022can} & Can a constitutional monarch influence democratic preferences? Japanese emperor and the regulation of public expression & 3: fixed Period 1 exposure & 755 & -3.37 & -- & -- & reconstructed from replication data \\
04 & \textcite{carnahan2022correcting} & Correcting the Misinformed: The Effectiveness of Fact-checking Messages in Changing False Beliefs & 1: exact/reducible & 348 & 5.40 & 4.70 & 0.11 & reconstructed from replication data \\
07 & \textcite{kachanoff2022equating} & Equating silence with violence: When White Americans feel threatened by anti-racist messages & 4: Group 1 shape, no data & 428 & 16.89 & -- & -- & copied from text/pdf/table \\
08 & \textcite{tappin2023estimating} & Estimating the between-issue variation in party elite cue effects & 2: multiperiod pairwise & 1240 & -- & -- & -- & copied from text/pdf/table \\
09 & \textcite{vaughn2022mass} & Mass support for proposals to reshape policing depends on the implications for crime and safety & 1: exact/reducible & 1107 & -0.33 & -0.27 & -1.75 & reconstructed from replication data \\
10 & \textcite{ozer2022partisan} & Partisan news versus party cues: The effect of cross-cutting party and partisan network cues on polarization and persuasion & 2: multiperiod pairwise & 797 & -- & -- & -- & copied from text/pdf/table \\
13 & \textcite{reeves2022unilateral} & Unilateral Inaction: Congressional Gridlock, Interbranch Conflict, and Public Evaluations of Executive Power & 4: Group 2 shape, no data & -- & -7.00 & -- & -- & copied from text/pdf/table \\
15 & \textcite{endres2023randomized} & A randomized experiment evaluating survey mode effects for video interviewing & 3: fixed Period 1 exposure & -- & 3.94 & -- & -- & copied from replication data \\
16 & \textcite{vantrappen2023biased} & Biased expectations? An experimental test of which party selectors are more likely to stereotype ethnic minority aspirants as less favorable than ethnic majority aspirants & 1: exact/reducible & 597 & -- & -8.48 & -0.68 & reconstructed from replication data \\
20 & \textcite{findor2023equality} & Equality, Reciprocity, or Need? Bolstering Welfare Policy Support for Marginalized Groups with Distributive Fairness & 4: Group 2 shape, no data & 113 & 3.60 & -- & -- & copied from text/pdf/table \\
22 & \textcite{velez2023latino} & Latino-Targeted Misinformation and the Power of Factual Corrections & 2: multiperiod pairwise & 2768 & 11.07 & -- & -- & reconstructed from replication data \\
25 & \textcite{armaly2023politicized} & Politicized Battles: How Vacancies and Partisanship Influence Support for the Supreme Court & 3: fixed Period 1 exposure & 603 & -- & -- & -- & no available replication data, and not reported in text/pdf/table \\
27 & \textcite{fortunato2023public} & Public Support for Professional Legislatures & 3: fixed Period 1 exposure & 998 & 4.39 & -- & -- & reconstructed from replication data \\
32 & \textcite{ozer2023women} & Women Experts and Gender Bias in Political Media & 4: Group 2 shape, no data & 525 & -- & -- & -- & reconstructed from replication data \\
33 & \textcite{yoon2026anger} & Anger expressions and coercive credibility in international crises & 3: fixed Period 1 exposure & -- & 10.90 & -- & -- & reconstructed from replication data \\
37 & \textcite{demarest2024bureaucracy} & Bureaucracy and Cyber Coercion & 3: fixed Period 1 exposure & -- & -- & -- & -- & copied from text/pdf/table \\
38 & \textcite{briggs2024changes} & Changes in Perceptions of Border Security Influence Desired Levels of Immigration & 4: Group 1 shape, no data & 851 & 3.88 & -- & -- & copied from text/pdf/table \\
47 & \textcite{halling2024frontline} & Frontline employees' responses to citizens' communication of administrative burdens & 2: multiperiod pairwise & -- & 9.23 & -- & -- & reconstructed from replication data \\
49 & \textcite{vanloon2024imagined} & Imagined otherness fuels blatant dehumanization of outgroups & 3: fixed Period 1 exposure & -- & 2.19 & -- & -- & reconstructed from replication data \\
50 & \textcite{clifford2024moral} & Moral Rhetoric, Extreme Positions, and Perceptions of Candidate Sincerity & 2: multiperiod pairwise & -- & -- & -- & -- & copied from text/pdf/table \\
53 & \textcite{stedtnitz2024public} & Public Reactions to Communication of Uncertainty: How Long-Term Benefits Can Outweigh Short-Term Costs & 3: fixed Period 1 exposure & -- & 2.84 & -- & -- & reconstructed from replication data \\
54 & \textcite{kotcher2024role} & Role model stories can increase health professionals' interest and perceived responsibility to engage in climate and sustainability actions & 3: fixed Period 1 exposure & 39 & -- & -- & -- & copied from text/pdf/table \\
55 & \textcite{nam2024scientific} & Scientific supremacy: How do genetic narratives relate to racism? & 3: fixed Period 1 exposure & -- & 2.51 & -- & -- & reconstructed from replication data \\
56 & \textcite{alcocer2024supplemental} & Supplemental online resources improve data literacy education: Evidence from a social science methods course & 4: Single-occasion, multi-item design & -- & 2.77 & -- & -- & copied from text/pdf/table \\
59 & \textcite{sorelle2024policy} & The policy acknowledgement gap: Explaining (mis)perceptions of government social program use & 4: JOT misclassification & -- & -- & -- & -- & no available replication data, and not reported in text/pdf/table \\
61 & \textcite{hanniman2025when} & When partisanship and technocratic credibility collide: mass attitudes and central bank endorsements of fiscal policy in Canada and the USA & 3: fixed Period 1 exposure & -- & 1.73 & -- & -- & copied from text/pdf/table \\
62 & \textcite{keim2025active} & Active Student Responding and Student Perceptions: A Replication and Extension & 4: Group 2 shape, no data & 100 & -- & -- & -- & no available replication data, and not reported in text/pdf/table \\
64 & \textcite{schmidt2025breaking} & Breaking the rules, but for whom? How client characteristics affect frontline professionals' prosocial rule-breaking behavior & 4: Group 2 shape, no data & -- & 3.53 & -- & -- & copied from text/pdf/table \\
70 & \textcite{obrochta2025language} & Language Cues and Perceptions of Nationalism & 4: JOT misclassification & -- & 5.48 & -- & -- & reconstructed from replication data \\
73 & \textcite{aydincakir2025on} & On motives and means: how approach and justification for court-curbing impact public trust & 3: fixed Period 1 exposure & 1014 & -- & -- & -- & copied from text/pdf/table \\
78 & \textcite{argyle2025testing} & Testing theories of political persuasion using AI & 3: fixed Period 1 exposure & -- & 3.94 & -- & -- & reconstructed from replication data \\
81 & \textcite{levin2025to} & To What End? Policy Objectives and US Public Support for Political Warfare & 4: JOT misclassification & -- & 2.47 & -- & -- & copied from replication data \\
82 & \textcite{haenschen2025tweet} & Tweet no harm: Offer solutions when alerting the public to voter suppression efforts & 3: fixed Period 1 exposure & 1085 & -2.80 & -- & -- & copied from text/pdf/table \\
\end{longtable}

%% file: tables/litreview_recovered_parameters.tex
\begin{longtable}{@{}l p{8.4cm} r r r r r r r r r r r r r c@{}}
\caption[]{Recovered quantities by study-outcome.}\label{tab:litreview-recovered-parameters}\\
\toprule
Row & Study and outcome & $N$ & $n_1$ & $n_2$ & $t_{co}$ & $t_{\delta}$ & $\hat\mu_{co}$ & $\hat\delta$ & $|\hat\delta/\hat\mu_{co}|$ & $\sigma_1^2$ & $\sigma_2^2$ & $\sigma_{12}$ & $r$ & $g$ & Included \\
\midrule
\endfirsthead
\multicolumn{16}{l}{\tablename\ \thetable\ (continued)}\\
\toprule
Row & Study and outcome & $N$ & $n_1$ & $n_2$ & $t_{co}$ & $t_{\delta}$ & $\hat\mu_{co}$ & $\hat\delta$ & $|\hat\delta/\hat\mu_{co}|$ & $\sigma_1^2$ & $\sigma_2^2$ & $\sigma_{12}$ & $r$ & $g$ & Included \\
\midrule
\endhead
\midrule \multicolumn{16}{r}{Continued on next page}\\ \endfoot
\midrule
\multicolumn{16}{@{}p{21.0cm}@{}}{\footnotesize \textit{Notes}: $t_{co}$ and $t_{\delta}$ are the $t$-statistics for the pooled estimate $\hat\mu_{co}$ and the carryover gap $\hat\delta$, respectively; $|\hat\delta/\hat\mu_{co}|$ is the absolute carryover-to-effect ratio; $r=\sigma_2^2/\sigma_1^2$ is the period-variance ratio and $g=\sigma_{12}/\sigma_1^2$ is the standardized cross-period covariance (Section~\ref{sec:testing-carryover}). Included indicates whether the row is one of the 27 Group 1 + Group 2 study-outcomes used in Figure~\ref{fig:r-g-map} (\ref{app:litreview-parameters}); the other 3 rows are excluded because a sequence has fewer than 20 observations (paper 10).}\\
\bottomrule \endlastfoot
04\_002 & Wave 1/Wave 2 respondents forming a genuine within-unit fact-check-status crossover (D1 != D2); Right belief x certainty, very certain (0-1) & 336 & 163 & 173 & 0.48 & 0.65 & 0.006 & 0.024 & 4.157 & 0.188 & 0.130 & 0.066 & 0.690 & 0.350 & Yes \\
04\_003 & Wave 1/Wave 2 respondents forming a genuine within-unit fact-check-status crossover (D1 != D2); Right belief x certainty, very or moderately certain (0-1) & 336 & 163 & 173 & 1.21 & 0.93 & 0.019 & 0.050 & 2.594 & 0.352 & 0.303 & 0.158 & 0.860 & 0.449 & Yes \\
04\_004 & Wave 1/Wave 2 respondents forming a genuine within-unit fact-check-status crossover (D1 != D2); Belief accuracy scale, full scale (0-1) & 348 & 168 & 180 & 4.70 & 0.11 & 0.036 & 0.004 & 0.100 & 0.121 & 0.111 & 0.075 & 0.918 & 0.620 & Yes \\
08\_007 & Original experiment; author's question-order rank 1 versus 2; Policy opinion recoded 0-1 toward in-party cue & 621 & 297 & 324 & 4.84 & -2.08 & 0.081 & -0.087 & 1.070 & 0.448 & 0.436 & 0.095 & 0.974 & 0.213 & Yes \\
08\_008 & Original experiment; author's question-order rank 1 versus 3; Policy opinion recoded 0-1 toward in-party cue & 612 & 318 & 294 & 3.83 & -1.31 & 0.068 & -0.054 & 0.793 & 0.447 & 0.449 & 0.066 & 1.005 & 0.147 & Yes \\
08\_009 & Original experiment; author's question-order rank 1 versus 4; Policy opinion recoded 0-1 toward in-party cue & 635 & 301 & 334 & 6.74 & 0.80 & 0.114 & 0.032 & 0.279 & 0.423 & 0.438 & 0.069 & 1.035 & 0.162 & Yes \\
08\_010 & Original experiment; author's question-order rank 1 versus 5; Policy opinion recoded 0-1 toward in-party cue & 570 & 274 & 296 & 5.71 & -0.08 & 0.108 & -0.003 & 0.032 & 0.429 & 0.452 & 0.032 & 1.052 & 0.075 & Yes \\
08\_011 & Original experiment; author's question-order rank 1 versus 6; Policy opinion recoded 0-1 toward in-party cue & 424 & 207 & 217 & 4.50 & -1.58 & 0.092 & -0.078 & 0.847 & 0.440 & 0.431 & 0.081 & 0.980 & 0.184 & Yes \\
08\_012 & Original experiment; author's question-order rank 1 versus 7; Policy opinion recoded 0-1 toward in-party cue & 236 & 119 & 117 & 2.82 & -1.49 & 0.085 & -0.093 & 1.096 & 0.410 & 0.471 & 0.015 & 1.148 & 0.037 & Yes \\
08\_013 & Original experiment; author's question-order rank 1 versus 8; Policy opinion recoded 0-1 toward in-party cue & 45 & 25 & 20 & -0.64 & -1.25 & -0.037 & -0.193 & 5.273 & 0.356 & 0.431 & 0.117 & 1.211 & 0.329 & Yes \\
09\_014 & October 2020 Lucid sample; first Table 1 comparison; Support for police abolition (1-5) & 1107 & 539 & 568 & -0.27 & -1.75 & -0.007 & -0.231 & 32.043 & 5.577 & 5.572 & 4.036 & 0.999 & 0.724 & Yes \\
10\_016 & Main experiment; post 1 versus post 2; Perceived pundit bias (0-1) & 12 & 9 & 3 & -0.95 & -0.82 & -0.022 & -0.311 & 14.000 & 0.367 & 0.337 & 0.342 & 0.918 & 0.930 & No \\
10\_017 & Main experiment; post 1 versus post 3; Perceived pundit bias (0-1) & 16 & 6 & 10 & -2.77 & 0.88 & -0.153 & 0.133 & 0.870 & 0.144 & 0.093 & 0.036 & 0.647 & 0.248 & No \\
10\_018 & Main experiment; post 1 versus post 4; Perceived pundit bias (0-1) & 32 & 17 & 15 & -2.28 & 0.10 & -0.085 & 0.014 & 0.167 & 0.179 & 0.172 & 0.094 & 0.963 & 0.528 & No \\
16\_021 & Flemish local party chairs; paired aspirant evaluations; Perceived competence (1-7) & 597 & 303 & 294 & 1.79 & -1.16 & 0.074 & -0.137 & 1.844 & 3.468 & 2.792 & 1.066 & 0.805 & 0.307 & Yes \\
16\_022 & Flemish local party chairs; paired aspirant evaluations; Perceived trustworthiness (1-7) & 597 & 303 & 294 & 2.04 & -0.06 & 0.066 & -0.008 & 0.118 & 2.887 & 2.651 & 1.513 & 0.918 & 0.524 & Yes \\
16\_023 & Flemish local party chairs; paired aspirant evaluations; Perceived ideological position (1-7) & 597 & 303 & 294 & -8.48 & -0.68 & -0.427 & -0.085 & 0.198 & 3.588 & 4.036 & 0.791 & 1.125 & 0.221 & Yes \\
22\_026 & Pooled trials; chronological trial-order rank 1 versus 2; Factual-belief accuracy (1-4) & 1011 & 507 & 504 & 7.22 & -0.29 & 0.257 & -0.033 & 0.128 & 4.549 & 4.311 & 1.873 & 0.948 & 0.412 & Yes \\
22\_027 & Pooled trials; chronological trial-order rank 1 versus 3; Factual-belief accuracy (1-4) & 389 & 203 & 186 & 4.14 & -0.71 & 0.237 & -0.128 & 0.540 & 4.505 & 4.379 & 1.906 & 0.972 & 0.423 & Yes \\
47\_035 & caseworker vignettes; presentation 1 versus 2; Willingness to reduce burdens (0-10) & 529 & 258 & 271 & 5.13 & 0.62 & 0.868 & 0.245 & 0.282 & 36.695 & 34.429 & 5.370 & 0.938 & 0.146 & Yes \\
47\_036 & caseworker vignettes; presentation 1 versus 3; Willingness to reduce burdens (0-10) & 515 & 248 & 267 & 5.63 & 0.49 & 0.966 & 0.191 & 0.198 & 34.321 & 34.322 & 4.182 & 1.000 & 0.122 & Yes \\
50\_039 & Study 1; issue 1 versus issue 2; Sincerity index & 96 & 49 & 47 & -0.82 & -0.22 & -0.029 & -0.023 & 0.803 & 0.348 & 0.396 & 0.138 & 1.136 & 0.396 & Yes \\
50\_040 & Study 1; issue 1 versus issue 3; Sincerity index & 113 & 59 & 54 & 1.29 & 0.15 & 0.048 & 0.015 & 0.313 & 0.416 & 0.456 & 0.123 & 1.097 & 0.295 & Yes \\
50\_041 & Study 1; issue 1 versus issue 4; Sincerity index & 97 & 48 & 49 & -0.44 & 0.72 & -0.017 & 0.079 & 4.799 & 0.390 & 0.446 & 0.156 & 1.143 & 0.400 & Yes \\
50\_042 & Study 1; issue 1 versus issue 5; Sincerity index & 96 & 48 & 48 & -0.13 & -2.28 & -0.004 & -0.251 & 57.421 & 0.401 & 0.390 & 0.171 & 0.974 & 0.427 & Yes \\
50\_043 & Study 1; issue 1 versus issue 6; Sincerity index & 96 & 48 & 48 & 0.85 & 0.73 & 0.042 & 0.076 & 1.799 & 0.499 & 0.463 & 0.022 & 0.928 & 0.044 & Yes \\
50\_044 & Study 1; issue 1 versus issue 7; Sincerity index & 100 & 50 & 50 & -1.04 & 0.41 & -0.045 & 0.038 & 0.827 & 0.420 & 0.363 & 0.017 & 0.864 & 0.042 & Yes \\
50\_045 & Study 1; issue 1 versus issue 8; Sincerity index & 103 & 53 & 50 & 0.88 & 0.63 & 0.039 & 0.055 & 1.425 & 0.431 & 0.351 & -0.002 & 0.813 & -0.005 & Yes \\
50\_046 & Study 1; issue 1 versus issue 9; Sincerity index & 114 & 57 & 57 & 1.10 & 1.14 & 0.043 & 0.111 & 2.570 & 0.431 & 0.442 & 0.093 & 1.026 & 0.215 & Yes \\
50\_047 & Study 1; issue 1 versus issue 10; Sincerity index & 93 & 46 & 47 & -0.08 & 0.61 & -0.003 & 0.061 & 17.639 & 0.397 & 0.439 & 0.034 & 1.107 & 0.085 & Yes \\
\end{longtable}

%% file: appendix_csp_calibration.tex
\subsection{Calibration to Clifford, Sheagley, and Piston Study 1}\label{app:csp-calibration}
%{\color{red}\small[TK-REVIEW]}

\begin{table}[!h]
\centering
\input{tables/clifford_calibration}
\caption{The two-step procedure applied to the within-subject arm of Study 1 of
\textcite{clifford2021increasing}, $n_{1}=227$, $n_{2}=226$.}
\label{tab:clifford-calibration}
\end{table}

This calibration applies the power and coverage results to the within-subject
arm of Study 1 of \textcite{clifford2021increasing}, using the deposited
respondent-level data. Sequence $(1,0)$ contains $n_{1}=227$ complete cases and
sequence $(0,1)$ contains $n_{2}=226$. The welfare outcome has three support
points. All variance components in Table~\ref{tab:clifford-calibration} are
calculated from the two sequence-specific pairs rather than inferred from a
parametric outcome model.

The estimated components place the study at $r=0.896$ and $g=0.514$: feasible
($g^{2}\le r$: $0.514^{2}=0.264\le0.896$) and, since
Proposition~\ref{prop:undercoverage} places the undercoverage region at
$g>(r-3)/2$ -- here $(0.896-3)/2=-1.052$ -- well inside it, with
$\tfrac{1}{4}\sigma_{co}^{2}/\sigma_{1}^{2}=0.217$. Since the pooled and
post-only intervals are $\hat\mu_{co}\pm z_{\alpha}\sqrt{\sigma_{co}^{2}/(4n)}$
and $\hat\mu_{1}\pm z_{\alpha}\sqrt{\sigma_{1}^{2}/n}$ respectively
(\ref{app:grizzle}), the ratio of their half-widths is
$\sqrt{\tfrac14\sigma_{co}^{2}/\sigma_{1}^{2}}=\sqrt{0.217}\approx0.466$.
Thus the pooled interval is
about $53\%$ shorter than the post-only interval. The same components imply
$\sigma_{\delta}/\sigma_{1}=1.71$, which produces the power comparison in
Figure~\ref{fig:power_function}: at 80\% power, the smallest detectable treatment
effect is $0.20$ scale points, whereas the smallest detectable carryover gap is
$0.34$.

Evaluating Equations~\eqref{eq:grizzle-coverage}--\eqref{eq:coverage-terms} at these components, the two-step interval covers with probability $0.924$ at $\delta=0$. Its
worst-case coverage is $0.433$, attained at a carryover gap of $0.188$ on the
three-point scale, as shown in Figure~\ref{fig:csp-coverage}. At that gap, the
study's own carryover test has power $0.348$, while the treatment-effect
test, evaluating $\psi_1$ from Section~\ref{sec:testing-carryover} at the
realized post-only estimate $\hat\mu_1=-0.289$ (se $0.070$,
$t_1\equiv\hat\mu_1/\mathrm{se}(\hat\mu_1)=-4.12$) in place of the unknown
$\mu_1$, has power $\psi_1(\hat\mu_1)=0.984$. The configuration that does
the most damage to the interval is therefore one the diagnostic would miss
about two times in three in a study amply powered for the treatment effect.

%% file: tables/clifford_calibration.tex
\begin{tabular}{lr}
\toprule
Quantity & Value \\
\midrule
Sequence sizes $(n_1, n_2)$ & $(227,\ 226)$ \\
$\sigma_1^2,\ \sigma_2^2,\ \sigma_{12}$ & $2.226,\ 1.995,\ 1.145$ \\
$r \equiv \sigma_2^2/\sigma_1^2$ & $0.896$ \\
$g \equiv \sigma_{12}/\sigma_1^2$ & $0.514$ \\
$\tfrac14\sigma_{co}^2 / \sigma_1^2$ (pooling is efficient iff $<1$) & $0.217$ \\
\addlinespace
Coverage at $\delta=0$ & $0.924$ \\
Worst-case coverage & $0.433$ \\
Carryover gap at the worst case & $0.188$ \\
Power of the carryover test at that gap & $0.348$ \\
Post-only ATE estimate $\hat\mu_1$ (se, $t$) & $-0.289$ ($0.070$, $-4.12$) \\
Power for the treatment effect at $\hat\mu_{1}$ & $0.984$ \\
$\mathrm{se}(\hat\mu_{co})/\mathrm{se}(\hat\mu_1)$ & $0.466$ \\
\bottomrule
\end{tabular}